\documentclass[12pt]{article}

\RequirePackage{amssymb}
\RequirePackage{amsmath}
\RequirePackage{latexsym}
\RequirePackage{mathrsfs}

\RequirePackage{bbm}
\newcommand{\indicator}{\ensuremath{\mathbbm{1}}}

\DeclareMathOperator*{\supp}{supp}

\newcommand{\R}{\ensuremath{\mathbb{R}}}
\newcommand{\Exp}{\ensuremath{\mathbb{E}}}
\newcommand{\Prob}{\ensuremath{\mathbb{P}}}

\pdfoutput=1
\usepackage{hyperref}
\hypersetup{
    colorlinks=true,
    linkcolor=black,
    citecolor=black,
    filecolor=black,
    urlcolor=black,
}

\usepackage[longnamesfirst]{natbib}
\usepackage{ulem}
\usepackage{graphicx}
\usepackage{pgfplots}
 \pgfplotsset{compat=1.18}

\usepackage{color}
\usepackage{setspace}
\usepackage{pgfplots}
\usepackage{subfig}
\usepackage{epstopdf}

\usepackage{multirow}

\usepackage{array}

\RequirePackage{amsthm}

\usepackage{titlesec}

\theoremstyle{definition}
\newtheorem{assumption}{Assumption}
\newtheorem{proposition}{Proposition}
\newtheorem{lemma}{Lemma}

\newtheorem{definition}{Definition}
\newtheorem{theorem}{Theorem}

\newtheorem{example}{Example}

\makeatletter
\newcommand*{\indep}{%
  \mathbin{%
    \mathpalette{\@indep}{}%
  }%
}
\newcommand*{\nindep}{%
  \mathbin{%
    \mathpalette{\@indep}{\not}%
  }%
}
\newcommand*{\@indep}[2]{%
  \sbox0{$#1\perp\m@th$}%
  \sbox2{$#1=$}%
  \sbox4{$#1\vcenter{}$}%
  \rlap{\copy0}%
  \dimen@=\dimexpr\ht2-\ht4-.2pt\relax
  \kern\dimen@
  {#2}%
  \kern\dimen@
  \copy0 %
} 
\makeatother

\newcommand{\E}{\ensuremath{\mathbb{E}}} 
\renewcommand{\P}{\ensuremath{\mathbb{P}}} 
\newcommand{\1}{\ensuremath{\mathbbm{1}}}
\newcommand{\lc}{\underline{c}}
\newcommand{\uc}{\overline{c}}
\newcommand{\lL}{\underline{\Lambda}}
\newcommand{\uL}{\overline{\Lambda}}

\usepackage{mathtools}

\title{\textbf{A General Approach to \\ Relaxing Unconfoundedness}\footnote{We thank participants at the 2024 Southern Economic Association annual meeting. Masten thanks the National Science Foundation for research support under Grant 1943138.}}
\author{Matthew A. Masten\footnote{Department of Economics, Duke University,
        \texttt{matt.masten@duke.edu}} \qquad Alexandre Poirier\thanks{
    Department of Economics, Georgetown University,
 \texttt{alexandre.poirier@georgetown.edu}} \qquad Muyang Ren\thanks{
    Department of Economics, Duke University,
 \texttt{muyang.ren@duke.edu}}
}
\date{January 26, 2025}

\begin{document}
\maketitle
\begin{abstract}
This paper defines a general class of relaxations of the unconfoundedness assumption. This class includes several previous approaches as special cases, including the marginal sensitivity model of \cite{Tan2006}. This class therefore allows us to precisely compare and contrast these previously disparate relaxations. We use this class to derive a variety of new identification results which can be used to assess sensitivity to unconfoundedness. In particular, the prior literature focuses on average parameters, like the average treatment effect (ATE). We move beyond averages by providing sharp bounds for a large class of parameters, including both the quantile treatment effect (QTE) and the distribution of treatment effects (DTE), results which were previously unknown even for the marginal sensitivity model.
\end{abstract}



\onehalfspacing
\normalsize

\newpage
\section{Introduction}\label{sec:intro_and_litreview}

A large literature studies the identification and estimation of treatment effects when a binary treatment is randomly assigned conditional on covariates. This assumption is called unconfoundedness, conditionally independent treatment assignment, or ignorability, among other terms. With observational data it is often considered very strong, however, so a corresponding literature has developed to relax this assumption. These papers use a variety of different classes of relaxations of unconfoundedness. That is, there are different ways of formalizing the idea that treatment is ``almost'' randomly assigned, given the covariates. This variation raises a question: How do these different relaxations compare to each other? This question is important because empirical researchers are often concerned that the number of robustness checks they must consider is constantly growing; if some of these checks are related, however, then that relationship can potentially be used to simplify the overall analysis. Moreover, mathematically related analyses do not necessarily provide ``independent'' evidence of robustness, a second motivation for better understanding the relationships between different relaxations of an assumption.

With that aim, this paper makes two main contributions. First, we define a general class of relaxations, which includes several previous approaches as special cases. Second, we derive closed form, analytical identification results for treatment effects under this general class of relaxations. This paper therefore unifies several disparate identification results in the literature. In doing so, we also provide a variety of new identification results, because we study an extensive list of parameters, including quantile treatment effects (QTEs) and the distribution of treatment effects (DTEs), whereas most existing papers focus solely on average-type treatment effects. These new results were previously unknown even for the specific types of relaxations that have been considered before. We give a precise discussion of how our results compare to the previous literature in the next subsection.

In section \ref{subsec:setup} we set up the baseline treatment effects model and define the target parameters we study. We define our general class of relaxations at the start of section \ref{subsec:generalClassOfRelax}. We show how this class relates to previous relaxations in sections \ref{sec:MSM} and \ref{sec:cdep}. In section \ref{sec:generalIdentificationResults} we derive general analytical identification results for marginal cdfs of potential outcomes and monotonic functionals of those cdfs. We apply those results in section \ref{sec:mainTreatmentEffectBounds} to obtain analytical bounds on various treatment effect parameters. We conclude in section \ref{sec:conclusion}.

\subsection*{Related Literature}

\nocite{Rosenbaum1995, Rosenbaum2002, Rosenbaum2017}

A vast literature studies unconfoundedness; we do not attempt a comprehensive review here. Instead we discuss the most closely related prior work. Nonparametric relaxations of unconfoundedness were pioneered by Paul Rosenbaum's work (see his 2002 or 2017 books for a survey, for example). His work focuses on sensitivity analysis within the context of finite sample randomization inference (c.f., chapter 5 of \citealt{ImbensRubin2015}). Much of the subsequent literature has instead focused on large population level identification analysis. In particular, inspired by Rosenbaum's approach, \cite{Tan2006} proposed the \emph{marginal sensitivity model} (MSM), a specific nonparametric relaxation of unconfoundedness (which we review in section \ref{subsec:generalClassOfRelax}). Given this relaxation, Tan showed that bounds on parameters of interest can be characterized as the solutions to optimization problems with infinitely many constraints, but did not provide any formal results, proofs, or closed form expressions for these bounds. \cite{ZhaoSmallBhattacharya2019} derived non-sharp bounds on the average potential outcome $\Exp(Y_x)$ and the average treatment effect (ATE) under the MSM, but also did not derive closed form expressions for these bounds. \cite{DornGuo2023} strengthened that result by deriving sharp bounds on $\Exp(Y_x)$, ATE, and the average effect of treatment on the treated (ATT) under the MSM, but again without closed form expressions. \cite{DornGuoKallus2024} subsequently refined that result by obtaining closed form expressions for sharp bounds on $\Exp(Y_x)$ and ATE under the MSM, in addition to developing the concept of double-validity and double sharpness. \cite{Tan2024} gives alternative sharp bound expressions for the ATE under the MSM. \cite{KallusZhou2018} studied policy learning under the MSM, which is related to identification of the average weighted welfare (what they call the ``policy value''), but they do not derive population bounds on this parameter.

This existing literature on the MSM largely focuses on average potential outcomes $\Exp(Y_x)$ or the ATE. Our paper provides the first sharp bounds on a wide variety of target parameters under the MSM, including the quantile treatment effect (QTE), the quantile treatment effect on the treated (QTT), the distribution of treatment effects (DTE), and the average weighted welfare (AWW). Moreover, for many of the parameters we study, our bounds are closed form. The existence of closed form expressions simplifies the construction of estimation and inference procedures, and also allows us to analytically examine how the bounds depend on the distribution of the observed data, and thus which features of the data lead results to be robust.

\cite{MastenPoirier2016,MastenPoirier2018} proposed an alternative relaxation of unconfoundedness called \emph{conditional $c$-dependence}, and derived closed form sharp bounds on a variety of treatment effect parameters under this relaxation, including $\Exp(Y_x)$, ATE, ATT, the QTE, and the DTE (in \citealt{MastenPoirier2019BF}). In the current paper we extend these identification results to a class of parameters that also includes the average weighted welfare (AWW), weighted average treatment effects, and to quantiles of the distribution of conditional average treatment effects (QCATE). Those earlier papers also restricted attention to continuous or binary outcomes whereas our new results apply for any distribution of the outcome, including mixed continous-discrete distributions. We also show how the conditional $c$-dependence relaxation is related to the marginal sensitivity model.

While our general class of relaxations includes several previously proposed relaxations of unconfoundedness, there are alternative relaxations where it is not yet clear if they can be accommodated by our class. This includes \cite{BonviniKennedy2022}, who derive closed-form sharp bounds on ATE under a mixture-style relaxation, and \cite{HuangPimentel2025}, who derive closed-form non-sharp bounds on ATT under an assumption about how much unobserved variables can affect the variance of odds ratios similar to those that arise in the MSM; in appendix A.6 they derive sharp but non-closed form bounds on the ATT under the same relaxation. It also includes \cite{DingVanderWeele2016} and \cite{VanderWeeleDing2017}, who derive closed-form non-sharp bounds on the causal relative risk under assumptions about relative risks involving latent confounders; \cite{Sjolander2024} derives closed-form sharp bounds under the same relaxation.

There are several related papers that provide general methods for deriving bounds. \cite{DornYap2024} show how to derive analytical expressions for sharp bounds on parameters that can be written as certain weighted averages of outcomes under a restriction on a generalized likelihood ratio. Like us, they show that their class of relaxations includes several previous relaxations (such as the MSM of \citealt{Tan2006} and conditional $c$-dependence of \citealt{MastenPoirier2018}). Whereas we only study relaxations of unconfoundedness, they also show how to use their results to do sensitivity analysis for instrumental variables and regression discontinuity designs. Their analysis of unconfoundedness, however, focuses on average potential outcomes and ATE, whereas we also study parameters like the QTE and DTE. A large literature in econometrics has studied how to derive identified sets for a variety of parameters under a variety of assumptions when all observed variables are discretely distributed; see, for example, \cite{Torgovitsky2018} and \cite{GuRussellStringham2024}, and the references therein. \cite{Duarte2024} uses similar ideas to numerically compute identified sets for a variety of sensitivity analyses when all variables are discretely distributed. \cite{RambachanCostonKennedy2023} provide a variety of general sensitivity analyses for binary outcomes. In contrast to this literature, we obtain analytical sharp bounds without any restriction on the distribution of the outcome variable, which allows the outcome to be continuously distributed or even mixed continuous-discretely distributed.

Several prior papers also discuss the relationship between various relaxations of unconfoundedness. \cite{MastenPoirier2023EJ} discuss mean independence, quantile independence assumptions, and a weaker version of quantile independence that they call $\mathcal{U}$-independence. Their focus is on interpreting relaxations in terms of treatment selection models, rather than providing identification results for a broad class of relaxations. \citet[section 7.2]{ZhaoSmallBhattacharya2019} discuss the relationship between the MSM and Rosenbaum's sensitivity model. For binary outcomes, \citet[appendix D]{RambachanCostonKennedy2023} relate their relaxation to the MSM, Rosenbaum's sensitivity model, conditional $c$-dependence, and an approach called Tukey's factorization.

\subsubsection*{Notation}

For random a variable $A$ and a random vector $B$, we let $F_{A \mid B}(a \mid b) \coloneqq \Prob(A \leq a \mid B=b)$ denote the conditional cdf. For $\tau \in (0,1)$, we let $Q_{A \mid B}(\tau \mid b) \coloneqq \inf\{a \in \R : F_{A \mid B}(a \mid b) \geq \tau\}$ denote the left-inverse of this cdf, that is, its conditional quantile function.

\section{Setup and Target Parameters}\label{subsec:setup}

We are interested in the causal impact of a binary treatment variable $X \in \{0,1\}$ on an outcome variable $Y$. Let $(Y_1,Y_0)$ be potential outcomes under treatment and no treatment respectively. Denote the realized outcome by 
\[
	Y = X Y_1 + (1-X) Y_0.
\]
Let $W$ be a vector of covariates with support $\supp(W) \subseteq \R^{d_W}$. We use $p_{x|w}$ to denote $\P(X=x \mid W=w)$. $p_{1|w}$ thus denotes the propensity score. We assume realizations of $(Y,X,W)$ are observed by the researcher. Our identification analysis abstracts from sampling uncertainty and assumes the joint distribution of $(Y,X,W)$ is known.

Throughout the paper we maintain the following assumption, which is a strict overlap assumption. It is also sometimes called strict positivity.

\begin{assumption}\label{assn:overlap}
There exists $\epsilon > 0$ such that $p_{1|w} \in [\epsilon,1-\epsilon]$ for all $w\in\supp(W)$.
\end{assumption}

With observational data, a commonly imposed assumption is unconfoundedness. It is also called selection on observables, ignorability, or conditional independence. This assumption states that potential outcomes are independent of treatment given covariates $W$. This conditional independence is either imposed jointly on both potential outcomes, or on each potential outcome separately:
\begin{equation}\label{eq:unconf}
	Y_x \indep X \mid W \text{ for } x \in \{0,1\} \qquad \text{ or } \qquad (Y_1,Y_0) \indep X \mid W.
\end{equation}
Under Assumption \ref{assn:overlap} and either version of unconfoundedness given in equation \eqref{eq:unconf}, it is well known that the pair of distribution functions $(F_{Y_1 \mid W},F_{Y_0 \mid W})$ is point-identified via
\[
	F_{Y_x \mid W}(y \mid w) = \P(Y \leq y \mid X=x, W=w)
\]
for $x \in \{0,1\}$. This point-identification implies that many parameters that summarize aspects of the distribution of $(Y_1,Y_0,X,W)$ are point-identified. Specifically, parameters that can be expressed as functionals of $F_{Y_1 \mid W}$, $F_{Y_0 \mid W}$, and the known distribution of observables $(Y,X,W)$, are point-identified. For example, we can point-identify the Conditional Average Treatment Effect (CATE) as defined by $\text{CATE}(w) \coloneqq \E[Y_1 - Y_0 \mid W=w]$ because it can be written as 
\begin{align*}
	\E[Y_1 - Y_0 \mid W=w] &= \int y_1 \, dF_{Y_1 \mid W}(y_1 \mid w) - \int y_0 \, dF_{Y_0 \mid W}(y_0 \mid w).
\end{align*}
However, parameters that depend on other aspects of the distribution of potential outcomes may only be partially identified. For example, consider the distribution function of $Y_1 - Y_0$, the unit-level treatment effect:
\begin{align*}
	F_{Y_1 - Y_0}(z) &= \int \1(y_1 - y_0 \leq z) \, dF_{Y_1,Y_0}(y_1,y_0).
\end{align*}
This parameter depends on the structure of the dependence between the two potential outcomes, which is unknown from either version of the unconfoundedness assumption. As discussed in \cite{FanPark2010}, $F_{Y_1 - Y_0}(z)$ is partially identified and sharp bounds can be recovered in terms of the distribution of $(Y,X,W)$.

To help classify treatment effect parameters, consider the decomposition
\begin{align*}
	F_{Y_1,Y_0 \mid X,W}(y_1,y_0 \mid x,w)
	&= C_{1,0 \mid X,W} \Big( F_{Y_1 \mid X,W}(y_1 \mid x,w),F_{Y_0 \mid X,W}(y_0 \mid x,w) \; \big| \; x,w \Big),
\end{align*}
where $C_{1,0 \mid X,W}(\cdot,\cdot \mid x,w)$ is a copula that characterizes the dependence between $Y_1$ and $Y_0$ conditional on $(X,W) = (x,w)$.  By Sklar's Theorem (\citealt{Sklar1959}), such a copula exists. Given that $F_{Y,X,W}$ is known, we consider treatment effect parameters that can be written as a function of $(F_{Y_1 \mid X,W},F_{Y_0 \mid X,W},C_{1,0 \mid X,W},F_{Y,X,W})$. We denote these parameters through the functional 
\begin{align}\label{eq:TE_param_general}
	\theta(F_{Y_1 \mid X,W},F_{Y_0 \mid X,W},C_{1,0 \mid X,W},F_{Y,X,W}),
\end{align}
and we give several examples below. The dependence of $\theta$ on some of its arguments is suppressed if the functional is constant with respect to them. We will sometimes denote these parameters as functionals of $(F_{Y_1 \mid W},F_{Y_0 \mid W},F_{Y,X,W})$ rather than $(F_{Y_1 \mid X,W},F_{Y_0 \mid X,W},F_{Y,X,W})$. These two formulations are equivalent due to the relationship
\begin{align}\label{eq:relationship_F_YXW_FYW}
	F_{Y_x \mid W}(y \mid w) &= F_{Y \mid X,W}(y \mid x,w)p_{x|w} + F_{Y_x \mid X,W}(y \mid 1-x,w)p_{1-x|w}
\end{align}
which holds for all $(y,x,w) \in \R \times \{0,1\}\times \supp(W)$.

Next we consider eleven example target parameters. Our results give sharp bounds for all eleven parameters under our general class of assumptions, including the marginal sensitivity model as a special case. For many of these parameters we also obtain analytical, closed form expressions for the bound functions.

\begin{example}[Conditional Quantile Treatment Effect]\label{ex:CQTE}
	For quantile index $\tau \in (0,1)$ and covariate value $w \in \supp(W)$, the conditional quantile treatment effect (CQTE) can be written as
	\begin{align*}
		\text{CQTE}(\tau \mid w) &\coloneqq Q_{Y_1 \mid W}(\tau \mid w) - Q_{Y_0 \mid W}(\tau \mid w) = \theta_\text{CQTE}(F_{Y_1 \mid W},F_{Y_0 \mid W};\tau,w)
	\end{align*}
	where $\theta_\text{CQTE}(F_{Y_1 \mid W},F_{Y_0 \mid W};\tau,w) \coloneqq F_{Y_1 \mid W}^{-1}(\tau \mid w) - F_{Y_0 \mid W}^{-1}(\tau \mid w)$ where $F_{Y_x \mid W}^{-1}(\tau \mid w) = \inf\{y \in \R:F_{Y_x \mid W}(y \mid w) \geq \tau\}$ is the left-inverse of $F_{Y_x \mid W}(\cdot \mid w)$, or its conditional quantile function.
\end{example}

\begin{example}[Conditional Average Treatment Effect]\label{ex:CATE}
	For covariate value $w \in \supp(W)$, the conditional average treatment effect (CATE) can be written as
	\begin{align*}
		\text{CATE}(w) &\coloneqq \E[Y_1 - Y_0 \mid W = w] = \theta_\text{CATE}(F_{Y_1 \mid W},F_{Y_0 \mid W};w)
	\end{align*}
	where $\theta_\text{CATE}(F_{Y_1 \mid W},F_{Y_0 \mid W};w) \coloneqq \int y_1 \, dF_{Y_1 \mid W}(y_1 \mid w) - \int y_0 \, dF_{Y_0 \mid W}(y_0 \mid w)$.
\end{example}

\begin{example}[Average Treatment Effect]\label{ex:ATE}
Denote the average treatment effect (ATE) as
	\begin{align*}
		\text{ATE} &\coloneqq \E[Y_1 - Y_0] = \theta_\text{ATE}(F_{Y_1 \mid W},F_{Y_0 \mid W},F_W)
	\end{align*}
	where $\theta_\text{ATE}(F_{Y_1 \mid W},F_{Y_0 \mid W},F_W) \coloneqq \int\left(\int y_1 \, dF_{Y_1 \mid W}(y_1 \mid w) - \int y_0 \, dF_{Y_0 \mid W}(y_0 \mid w)\right) \, dF_W(w)$. We can also consider weighted average treatment effects of the kind
	\begin{align*}
		\text{WATE}(\omega) &\coloneqq \E[\omega(W)(Y_1 - Y_0)] 
	\end{align*}
	for an identified function $\omega:\supp(W) \to \R_{\geq 0}$. The ATE is a special case where $\omega(w) = 1$.
\end{example}

\begin{example}[Average Treatment Effect on the Treated]\label{ex:ATT}
Denote the average treatment effect on the treated (ATT) as
	\begin{align*}
		\text{ATT} &\coloneqq\E[Y_1 - Y_0 \mid X=1] = \theta_\text{ATT}(F_{Y_0 \mid X,W},F_{Y,X,W})
	\end{align*}
	where $\theta_\text{ATT}(F_{Y_0 \mid X,W},F_{Y,X,W}) \coloneqq \E[Y \mid X=1] - \int \int y_0 \, dF_{Y_0 \mid X,W}(y_0 \mid 1,w) \, dF_{W \mid X}(w \mid 1)$. 
\end{example}

\begin{example}[Quantile Treatment Effect]\label{ex:QTE}
	For $\tau \in (0,1)$, denote the quantile treatment effect (QTE) as
	\begin{align*}
		\text{QTE}(\tau) &\coloneqq Q_{Y_1}(\tau) - Q_{Y_0}(\tau) = \theta_\text{QTE}(F_{Y_1 \mid W},F_{Y_0 \mid W},F_W;\tau)
	\end{align*}
	where $\theta_\text{QTE}(F_{Y_1 \mid W},F_{Y_0 \mid W},F_W;\tau) \coloneqq F_{Y_1}^{-1}(\tau) - F_{Y_0}^{-1}(\tau)$ and $F_{Y_x}(\cdot) = \int F_{Y_x \mid W}(\cdot \mid w) \, dF_W(w)$ for $x \in \{0,1\}$.
\end{example}

\begin{example}[Quantile Treatment Effect on the Treated]\label{ex:QTT}
For $\tau \in (0,1)$, denote the quantile treatment effect on the treated (QTT) as
\begin{align*}
	\text{QTT}(\tau)
	&\coloneqq Q_{Y_1 \mid X}(\tau \mid 1) - Q_{Y_0 \mid X}(\tau \mid 1) 
	= \theta_\text{QTT}(F_{Y_0 \mid X,W},F_{Y,X,W};\tau)
\end{align*}
where $\theta_\text{QTT}(F_{Y_0 \mid X,W},F_{Y,X,W}; \tau) \coloneqq Q_{Y \mid X}(\tau \mid 1) - F_{Y_0 \mid X}^{-1}(\tau \mid 1)$ with
\[
	F_{Y_0 \mid X}(\cdot \mid 1) = \int F_{Y_0 \mid X,W}(\cdot \mid 1,w) \, dF_{W \mid X}(w \mid 1).
\]
\end{example}

\begin{example}[CATE Distribution]\label{ex:CATE_dist}
	For $\tau \in (0,1)$, denote the quantile of the CATE (QCATE) as
	\begin{align*}
		\text{QCATE}(\tau) &\coloneqq F^{-1}_{\text{CATE}(W)}(\tau) = \theta_\text{QCATE}(F_{Y_1 \mid W},F_{Y_0 \mid W},F_W;\tau)
	\end{align*}
	where $\theta_\text{QCATE}(F_{Y_1 \mid W},F_{Y_0 \mid W},F_W;\tau) \coloneqq F^{-1}_{\text{CATE}(W)}(\tau)$ with
\[
	F_{\text{CATE}(W)}(z) = \int \1(\theta_\text{CATE}(F_{Y_1 \mid W},F_{Y_0 \mid W};w) \leq z) \, dF_W(w).
\]
This parameter is motivated by the sorted effects studied in \cite{ChernozhukovFernandez-ValLuo2018}.
\end{example}

\begin{example}[Average Weighted Welfare]\label{ex:AWW}
	For a weight (or assignment) function $\omega: \supp(W) \to [0,1]$, denote the average weighted welfare (AWW) as
	\begin{align*}
		\text{AWW}(\omega) &\coloneqq \E[\omega(W) Y_1 + (1 - \omega(W)) Y_0] = \theta_\text{AWW}(F_{Y_1 \mid W},F_{Y_0 \mid W}, F_W;\omega)
	\end{align*}
where 
\[
\centering
\scalebox{0.9}{$\displaystyle\
	\theta_\text{AWW}(F_{Y_1 \mid W},F_{Y_0 \mid W}, F_W;\omega) \coloneqq \int\left( \omega(w) \int y_1 \, dF_{Y_1 \mid W}(y_1 \mid w) + (1-\omega(w)) \int y_0 \, dF_{Y_0 \mid W}(y_0 \mid w)\right) \, dF_W(w).
$}
\]
\cite{KallusZhou2018} called the AWW the \emph{policy value}.
\end{example}

\begin{example}[Joint Distribution Function]\label{ex:Joint_dist}
	For $(y_1,y_0) \in \R^2$, the joint cumulative distribution function (cdf) of $(Y_1,Y_0)$ is
	\begin{align*}
		F_{Y_1,Y_0}(y_1,y_0) &\coloneqq \P(Y_1 \leq y_1, Y_0 \leq y_0) = \theta_\text{CDF}(F_{Y_1 \mid X,W},F_{Y_0 \mid X,W},C_{1,0 \mid X,W},F_{Y,X,W};y_1,y_0)
	\end{align*}
	where 
	\begin{align*}
	&\theta_\text{CDF}(F_{Y_1 \mid X,W},F_{Y_0 \mid X,W},C_{1,0 \mid X,W},F_{Y,X,W};y_1,y_0)\coloneqq\\ &\quad \sum_{x \in \{0,1\}}\int C_{1,0 \mid X,W}(F_{Y_1 \mid X,W}(y_1 \mid x, w),F_{Y_0 \mid X,W}(y_0 \mid x,w) \mid x,w)p_{x|w} \, dF_W(w).
	\end{align*}
\end{example}

\begin{example}[Distribution of Treatment Effects]\label{ex:DTE}
For $z \in \R$, the cumulative distribution function for the unit level treatment effect $Y_1 - Y_0$ (called the DTE) is
\begin{align*}
		\text{DTE}(z)
		\coloneqq F_{Y_1 - Y_0}(z) = \P(Y_1 - Y_0 \leq z) = \theta_\text{DTE}(F_{Y_1 \mid X,W},F_{Y_0 \mid X,W},C_{1,0 \mid X,W},F_{Y,X,W};z)
	\end{align*}
	where 
	\begin{align*}
	&\theta_\text{DTE}(F_{Y_1 \mid X,W},F_{Y_0 \mid X,W},C_{1,0 \mid X,W},F_{Y,X,W};z) \coloneqq\\
	& \quad \sum_{x \in \{0,1\}} \int \left( \int_{\{y_1 - y_0 \leq z\}} dC_{1,0 \mid X,W}(F_{Y_1 \mid X,W}(y_1 \mid x,w),F_{Y_0 \mid X,W}(y_0 \mid x,w) \mid x,w)\right) p_{x|w} \, dF_W(w).
	\end{align*}
\end{example}

\begin{example}[Quantiles of Treatment Effects]\label{ex:QDTE}
	For $\tau \in (0,1)$, the quantiles of the distribution function of treatment effect (QDTE) $Y_1 - Y_0$ is
	\begin{align*}
		Q_{Y_1 - Y_0}(\tau) &\coloneqq \inf\{z \in \R: F_{Y_1 - Y_0}(z) \geq \tau\} = \theta_\text{QDTE}(F_{Y_1 \mid X,W},F_{Y_0 \mid X,W},C_{1,0 \mid X,W},F_{Y,X,W};\tau)
	\end{align*}
	where 
\begin{align*}
	&\theta_\text{QDTE}(F_{Y_1 \mid X,W},F_{Y_0 \mid X,W},C_{1,0 \mid X,W},F_{Y,X,W};\tau) \\
	&\qquad \coloneqq \inf\{z \in \R: \theta_\text{DTE}(F_{Y_1 \mid X,W},F_{Y_0 \mid X,W},C_{1,0 \mid X,W},F_{Y,X,W};z) \geq \tau\}.
\end{align*}
\end{example}

The parameters in examples \ref{ex:CQTE}--\ref{ex:AWW} only depend on the distribution of potential outcomes through their marginal distributions given $(X,W)$, while the parameters in examples \ref{ex:Joint_dist}--\ref{ex:QDTE} also depend on their copulas. Under overlap and unconfoundedness, the parameters in \ref{ex:CQTE}--\ref{ex:AWW} are all point-identified. The parameters \ref{ex:Joint_dist}--\ref{ex:QDTE} are partially identified under overlap and unconfoundedness since the conditional copulas $C_{1,0 \mid X,W}$ are not identified from the joint distribution of $(Y,X,W)$. In other words, these parameters depend on the type of dependence between $Y_1$ and $Y_0$, and unconfoundedness does not reveal any information about this dependence. For example, \cite{FanPark2010} show the identified set for $F_{Y_1 - Y_0}(z)$, the DTE, is an interval and they provide a closed-form expression for its lower and upper bounds.

However, if unconfoundedness fails, all these parameters will be partially identified. The identified set for parameters that are partially identified under unconfoundedness becomes larger under failures of unconfoundedness. 

\section{Relaxing Unconfoundedness}\label{subsec:generalClassOfRelax}

We now consider relaxations of the unconfoundedness assumption. We will consider two related, general relaxations of unconfoundedness that encompass several disparate relaxations that were studied in the literature. We begin by considering a class of assumptions on the probabilities of treatment when conditioning on covariates $W$ and one of the potential outcomes.

\begin{assumption}[Marginal $c$-dependence]\label{assn:marginal_cdep}
	Let $(\lc(w,\eta),\uc(w,\eta))$ satisfy $0 < \lc(w,\eta) \leq p_{1|w} \leq \uc(w,\eta) < 1$ for all $w \in \supp(W)$. The potential outcomes satisfy \textit{marginal $c$-dependence} if, for $x \in \{0,1\}$,  
\begin{align*}	
	p_x(Y_x,w) \coloneqq \P(X=1 \mid Y_x,W=w) \in [\lc(w,\eta),\uc(w,\eta)]
\end{align*}	
almost surely conditional on $W=w$ for all $w\in\supp(W)$.
\end{assumption}

This assumption restricts the manner in which potential outcomes affect the treatment probability $p_x(y,w)$, which we call a \emph{latent propensity score}. We will use $c$-dependence assumptions to conduct sensitivity analyses for unconfoundedness. Here the sensitivity parameters are $\lc(w,\eta)$  and $\uc(w,\eta)$, which we refer to as bound functions. Like the notation in \cite{RambachanCostonKennedy2023}, we let $\eta$ be a possibly infinite-dimensional nuisance parameter that is point-identified from the distribution $F_{Y,X,W}$. The bound functions are also allowed to depend on the covariate value $w$. In principle, we can also allow the bounds to differ across $x \in \{ 0,1 \}$, but we do not include an $x$ subscript for simplicity. The specification of $\lc(w,\eta)$ and $\uc(w,\eta)$ is left implicit, which allows them to be functions of low-dimensional or scalar sensitivity parameters. For example, the marginal sensitivity model of \cite{Tan2006}, which depends on a single sensitivity parameter, can be viewed as a special case of marginal $c$-dependence. We show this in section \ref{sec:MSM}.

We can also see that setting $\lc(w,\eta) = \uc(w,\eta) = p_{1|w}$ yields unconfoundedness as a special case of marginal $c$-dependence, while letting $(\lc(w,\eta),\uc(w,\eta))$ approach $(0,1)$ implies that no restrictions on the dependence between $X$ and $Y_x$ (given covariates) are imposed. Note that we restrict the propensity score $p_{1|w}$ to lie within our specified bounds for $p_x(Y_x,w)$. If the propensity score were outside these bounds, then the assumption would be misspecified because, by the law of iterated expectations, $p_{1|w} = \E[\P(X=1 \mid Y_x,W = w) \mid W = w] \in [\lc(w,\eta),\uc(w,\eta)]$.
 
We also consider a closely related class of assumptions that restricts the probability of treatment given both potential outcomes. 
 
\begin{assumption}[Joint $c$-dependence]\label{assn:joint_cdep}
	Let $(\lc(w,\eta),\uc(w,\eta))$ satisfy $0 < \lc(w,\eta) \leq p_{1|w} \leq \uc(w,\eta) < 1$ for all $w \in \supp(W)$. The potential outcomes satisfy \textit{joint $c$-dependence} if  
\begin{align*}	
	p(Y_1,Y_0,w) \coloneqq \P(X=1 \mid Y_1,Y_0,W=w) \in [\lc(w,\eta),\uc(w,\eta)]
\end{align*}	
almost surely conditional on $W=w$ for all $w\in\supp(W)$.
\end{assumption} 

Joint $c$-dependence with bound functions $\lc(w,\eta)$ and $\uc(w,\eta)$ implies marginal $c$-dependence with the same bound functions. This is due to the law of iterated expectations. 

\begin{lemma}\label{lemma:joint_implies_marginal}
	Let Assumption \ref{assn:joint_cdep} hold for $(\lc(w,\eta),\uc(w,\eta))$. Then, Assumption \ref{assn:marginal_cdep} holds for $(\lc(w,\eta),\uc(w,\eta))$.
\end{lemma}

\medskip

We next show that several unconfoundedness relaxations from recent related literature can be viewed as special cases of either marginal or joint $c$-dependence. 
 
\subsection{The Marginal Sensitivity Model}\label{sec:MSM}

\cite{Tan2006} proposed the Marginal Sensitivity Model (MSM), which restricts the odds ratio between propensity scores and treatment probabilities that also condition on the potential outcome $Y_x$, for $x = 0,1$. For $x \in \{0,1\}$, let
\begin{align*}
	R_x(Y_x,W) &\coloneqq \frac{\Prob(X=1\mid Y_x ,W)}{\Prob(X=0\mid Y_x ,W)} \bigg/\frac{\P(X=1 \mid W)}{\P(X=0 \mid W)} 
\end{align*} 
denote this odds ratio. When $Y_x$ is continuously distributed with respect to the Lebesgue measure, this ratio can also be expressed as ratios of conditional densities of $Y_x \mid X=1,W$ and $Y_x \mid X=0,W$.

\cite{Tan2006}'s MSM posits known bounds for these odds ratios.
\begin{definition}[Marginal Sensitivity Model (MSM)]
\label{def:MSM}
Let $\Lambda \in [1,+\infty)$ be known. The potential outcomes satisfy the \textit{Marginal Sensitivity Model} if 
\begin{align*}
	R_x(Y_x,w) \in \left[\Lambda^{-1},\Lambda\right] \text{ for } x \in \{0,1\}
\end{align*}
almost surely conditional on $W = w$ for all $w \in \supp(W)$.
\end{definition}

$\Lambda$ is a scalar sensitivity parameter. Setting $\Lambda = 1$ is equivalent to assuming unconfoundedness, and increasing $\Lambda$ allows for more dependence of latent propensity scores on potential outcomes. Variants of \cite{Tan2006}'s MSM whose odds ratios condition on both potential outcomes have also been considered. Similarly, these odds ratios may instead condition on an abstract confounder $U$ rather than potential outcomes. See, for example, \cite{DornGuo2023} and \cite{DornGuoKallus2024} for recent examples of these two variants. To distinguish it from the case where one conditions on the potential outcomes one at a time, we call the version that conditions on both potential outcomes the \textit{Joint} Sensitivity Model (JSM).

\begin{definition}[Joint Sensitivity Model (JSM)]
\label{def:JSM}
Let $\Lambda \in [1,+\infty)$ be known. The potential outcomes satisfy the \textit{Joint Sensitivity Model}  if 
\begin{align*}
	 R(Y_1,Y_0,w) \in \left[\Lambda^{-1},\Lambda\right] \text{ for } x \in \{0,1\}
\end{align*}
almost surely conditional on $W = w$ for all $w \in \supp(W)$, where
\begin{align*}
	 R(Y_1,Y_0,W) &\coloneqq \frac{\Prob(X=1\mid Y_1, Y_0, W)}{\Prob(X=0\mid Y_1,Y_0,W)} \bigg/ \frac{\P(X=1 \mid W)}{\P(X=0 \mid W)}.
\end{align*} 
\end{definition}
 
We now state generalizations of the MSM and JSM that allow their odds ratios to have arbitrary bounds, as opposed to bounds that have product equal to 1. We also allow their bounds to depend on covariates or nuisance parameters. We will continue distinguishing between \textit{marginal} sensitivity models, which condition on one potential outcome at a time, and \textit{joint} sensitivity models, which condition on both potential outcomes. 

\begin{definition}[Generalized Sensitivity Models] \label{def:GMSM}
Let $\lL(w,\eta) \in (0,1]$ and $\uL(w,\eta) \in [1,+\infty)$ for all $w \in \supp(W)$ where $\lL(w,\eta)$ and $\uL(w,\eta)$ are known. 
The potential outcomes satisfy the \textit{Generalized Marginal Sensitivity Model} (GMSM) if 
\begin{align*}
	R_x(Y_x,w) \in \left[\lL(w,\eta), \uL(w,\eta)\right] \text{ for } x \in \{0,1\}
\end{align*}
almost surely conditional on $W = w$ for all $w \in \supp(W)$.
They satisfy the \textit{Generalized Joint Sensitivity Model} (GJSM) if 
\begin{align*}
	R(Y_1,Y_0,w) \in \left[\lL(w,\eta), \uL(w,\eta)\right]
\end{align*}
almost surely conditional on $W = w$ for all $w \in \supp(W)$.
\end{definition}

\medskip

We can see that the MSM is a special case of the GMSM by setting $[\lL(w,\eta),\uL(w,\eta)] = [\Lambda^{-1},\Lambda]$. Similarly, the JSM is a special case of the GJSM. 

The GMSM is equivalent to marginal $c$-dependence because, for each bound function pair $[\lc(w,\eta),\uc(w,\eta)]$ under marginal $c$-dependence, there exists exactly one corresponding bound function pair $[\lL(w,\eta), \uL(w,\eta)]$ under the GMSM. The same link exists between joint $c$-dependence and the GJSM. We show this in the following proposition.

\begin{proposition}[Equivalence of Sensitivity Models]\label{prop:c-dep_GMSM_equivalence}\hfill
\begin{enumerate}
	\item Let marginal (joint) $c$-dependence hold with bound functions $[\lc(w,\eta),\uc(w,\eta)]$. Then the GMSM (GJSM) holds with bound functions
\begin{align}\label{eq:c-dep_GMSM_equivalence}
	\left[\lL(w,\eta),\uL(w,\eta)\right] &= \left[\frac{\lc(w,\eta)}{1-\lc(w,\eta)}\frac{p_{0|w}}{p_{1|w}},\frac{\uc(w,\eta)}{1-\uc(w,\eta)}\frac{p_{0|w}}{p_{1|w}}\right].
\end{align}\label{eq:c-dep_GMSM_equivalence2}
	\item Let the GMSM (GJSM) hold with bound functions $[\lL(w,\eta),\uL(w,\eta)]$. Then marginal (joint) $c$-dependence holds with bound functions
	\begin{align}
	[\lc(w,\eta),\uc(w,\eta)] &= \left[\frac{p_{1|w}\lL(w,\eta)}{p_{0|w} + p_{1|w}\lL(w,\eta)},\frac{p_{1|w}\uL(w,\eta)}{p_{0|w} + p_{1|w}\uL(w,\eta)}\right].
\end{align}
\end{enumerate}
\end{proposition}

\medskip

This proposition shows that the marginal $c$-dependence is equivalent to the generalized marginal sensitivity model. Similarly, joint $c$-dependence is equivalent to the generalized marginal sensitivity model.
 
\subsection{Conditional $c$-dependence}\label{sec:cdep}

\cite{MastenPoirier2018} studied a relaxation of unconfoundedness they called \textit{conditional $c$-dependence}, which assumed symmetric bounds on the latent propensity score.
\begin{definition}[Conditional $c$-dependence]
Let $c \in [0,1]$ be a known scalar sensitivity parameter. The potential outcomes satisfy \textit{conditional} $c$\textit{-dependence} if 
\begin{align*}
	p_x(Y_x,w) &\coloneqq \P(X=1 \mid Y_x,W = w) \in [p_{1|w} - c, p_{1|w} + c]
\end{align*}
almost surely conditional on $W=w$ for all $w \in \supp(W)$.
\end{definition}

\medskip

This is a special case of marginal $c$-dependence where the bounds equal
\begin{align*}
	\lc(w,\eta) = p_{1|w} - c \qquad \text{ and } \qquad \uc(w,\eta) = p_{1|w} + c.
\end{align*}
Here the nuisance parameter is $p_{1|(\cdot)}$, the propensity score function. Unconfoundedness is obtained by setting $c = 0$, while the no-assumption bounds are obtained for $c$'s equal to or larger than $\sup_{w \in \supp(W)} \max\{p_{1|w},p_{0|w}\}$. \cite{MastenPoirier2018} provided closed-form expressions for sharp bounds on the CQTE, CATE, ATE, QTE, and ATT when potential outcomes are continuously distributed or binary. \cite{MastenPoirierZhang2020} describe flexible parametric estimators of these bounds and provide nonstandard inference methods.

\section{General Identification Results}\label{sec:generalIdentificationResults}

Next we derive sharp bounds on a large class of target parameters under the relaxations described in section \ref{subsec:generalClassOfRelax}. We will study a class of parameters that includes all eleven examples in section \ref{subsec:setup} as special cases. Specifically, we will compute these parameters' sharp bounds, or their identified set, under marginal and joint $c$-dependence, which are equivalent to the GMSM and GJSM, respectively. %

\subsection{Bounds on Marginal Distributions}\label{subsec:cdf_bounds}

Before studying our general class of treatment effect parameters, we first consider bounds on the distribution functions of each potential outcome, given covariates. These cdfs are building blocks for these parameters and, as we will see,  analytical bounds on these cdfs will directly map into analytical bounds on these parameters. 

Specifically, we begin by analyzing the conditional cdf of the potential outcome $Y_x$ given the covariate value $w$, $F_{Y_x \mid W}(y \mid w) \coloneqq \P(Y_x \leq y \mid W=w)$. Under either marginal or joint $c$-dependence, we can show that this cdf is bounded above and below by two cdfs which form an envelope for $F_{Y_x \mid W}(y \mid w)$ for all values of $(y,w) \in \R \times \supp(W)$. 

Define the following functions:
\begin{center}
\resizebox{0.98\textwidth}{!}{$
\begin{aligned}
	\underline{F}_{Y_1 \mid W}(y \mid w) &= \max\left\{F_{Y \mid X,W}(y \mid 1,w)\frac{p_{1|w}}{\uc(w,\eta)}, \frac{\lc(w,\eta) - p_{1|w}}{\lc(w,\eta)} + F_{Y \mid X,W}(y \mid 1,w)\frac{p_{1|w}}{\lc(w,\eta)}\right\}\\
	\overline{F}_{Y_1 \mid W}(y \mid w) &= \min\left\{F_{Y \mid X,W}(y \mid 1,w)\frac{p_{1|w}}{\lc(w,\eta)}, \frac{\uc(w,\eta) - p_{1|w}}{\uc(w,\eta)} + F_{Y \mid X,W}(y \mid 1,w)\frac{p_{1|w}}{\uc(w,\eta)}\right\}\\
	\underline{F}_{Y_0 \mid W}(y \mid w) &= \max\left\{F_{Y \mid X,W}(y \mid 0,w)\frac{p_{0|w}}{1-\lc(w,\eta)}, \frac{p_{1|w}-\uc(w,\eta)}{1-\uc(w,\eta)} + F_{Y \mid X,W}(y \mid 0,w)\frac{p_{0|w}}{1-\uc(w,\eta)}\right\}\\
	\overline{F}_{Y_0 \mid W}(y \mid w) &= \min\left\{F_{Y \mid X,W}(y \mid 0,w)\frac{p_{0|w}}{1-\uc(w,\eta)}, \frac{p_{1|w}-\lc(w,\eta)}{1-\lc(w,\eta)} + F_{Y \mid X,W}(y \mid 0,w)\frac{p_{0|w}}{1-\lc(w,\eta)}\right\}.
\end{aligned}
$}
\end{center}
Viewed as functions of $y$, these four functions are cdfs since they are nondecreasing, right-continuous, and their limits as $y \rightarrow -\infty,+\infty$ equal 0 and 1, respectively. We show these four cdfs form bounds for $F_{Y_x \mid W}$ under marginal or joint $c$-dependence.
\begin{lemma}\label{lemma:cdf_bounds_margcdep}
	Let Assumption \ref{assn:overlap} hold. Let either Assumption \ref{assn:marginal_cdep} or \ref{assn:joint_cdep} hold. Then, for all $(y,w) \in \R \times \supp(W)$, 
\[
		\P(Y_1 \leq y \mid W=w) \in \left[\underline{F}_{Y_1 \mid W}(y \mid w),\overline{F}_{Y_1 \mid W}(y \mid w)\right]
\]
and
\[
	\P(Y_0 \leq y \mid W=w) \in \left[\underline{F}_{Y_0 \mid W}(y \mid w),\overline{F}_{Y_0 \mid W}(y \mid w)\right].
\]
\end{lemma}

\medskip

We note a few properties of these bounds. The bounds for $F_{Y_x \mid W}(y \mid w)$ collapse to a point if either $\lc(w,\eta) = p_{1|w}$ or $\uc(w,\eta) = p_{1|w}$. Also note that $F_{Y \mid X,W}(y \mid x,w)$ always lies within the bounds for $F_{Y_x \mid W}(y \mid w)$. This is because $c$-dependence never rules out unconfoundedness, and unconfoundedness implies that the distribution of $Y$ given $(X,W) = (x,w)$ equals that of $Y_x$ given $W = w$. 

These cdf bounds also yield cdf bounds under the GMSM or GJSM since they are equivalent to $c$-dependence as shown in Proposition \ref{prop:c-dep_GMSM_equivalence}. These bounds are also valid for the standard MSM or JSM, as they are special cases of marginal or joint $c$-dependence.

We now show these cdf bounds are sharp, or that they cannot be improved upon. This is the case under marginal or joint $c$-dependence. The cdf of $Y_x \mid W=w$ can also lie in the interior of these bounds, as we show that any convex linear combination of the upper and lower cdf bounds can be attained. 

Before establishing this, let $\mathcal{C}$ denote the set of all bivariate copulas and let
\begin{align*}
	\mathcal{C}_{1,0 \mid X,W} &= \Big\{ \{C_{1,0|x,w}\}_{x \in \{0,1\}, w \in \supp(W)} \text{ such that } C_{1,0|x,w} \in \mathcal{C} \Big\}
\end{align*}
denote the collection of all bivariate copulas across all treatment and covariate values $(x,w) \in \{0,1\} \times \supp(W)$. We also say that a distribution function for $(Y_1,Y_0) \mid X,W$ is compatible with the observed distribution $F_{Y,X,W}$ if
\begin{align}\label{eq:dist_compatibility}
	F_{Y_1 \mid X,W}(\cdot \mid 1,w) = F_{Y \mid X,W}(\cdot \mid 1,w) \qquad \text{ and } \qquad F_{Y_0 \mid X,W}(\cdot \mid 0,w) = F_{Y \mid X,W}(\cdot \mid 0,w)
\end{align}
for all $w\in\supp(W)$.

We now define the identified set for the distribution of $(Y_1,Y_0) \mid X,W$ from the observable distribution $F_{Y,X,W}$ under $c$-dependence. This set consists of all conditional cdfs and copulas that imply a distribution for $(Y_1,Y_0) \mid X,W$ that is both compatible with the data distribution $F_{Y,X,W}$ and with a $c$-dependence condition.

\begin{definition}[Identified Set]
For a given distribution of the observables $F_{Y,X,W}$ and bound functions $c \coloneqq (\lc(w,\eta),\uc(w,\eta))$, the identified set for $(F_{Y_1 \mid X,W}, F_{Y_0 \mid X,W},C_{1,0 \mid X,W})$ under marginal $c$-dependence is given by
\begin{align*}
	\mathcal{I}^{\text{marg}}(F_{Y,X,W};c) &\coloneqq \{(F_{Y_1 \mid X,W}, F_{Y_0 \mid X,W},C_{1,0 \mid X,W}): F_{Y_1,Y_0 \mid X,W} = C_{1,0 \mid X,W}(F_{Y_1 \mid X,W},F_{Y_0 \mid X,W}) \\
	&\qquad\qquad \text{ and } p_{1|(\cdot)} \text{ satisfy equation } \eqref{eq:dist_compatibility} \text{ and  Assumption \ref{assn:marginal_cdep}}\}.
\end{align*}
The identified set under joint $c$-dependence is given by
\begin{align*}
	\mathcal{I}^{\text{joint}}(F_{Y,X,W};c) &\coloneqq \{(F_{Y_1 \mid X,W}, F_{Y_0 \mid X,W},C_{1,0 \mid X,W}): F_{Y_1,Y_0 \mid X,W} = C_{1,0 \mid X,W}(F_{Y_1 \mid X,W},F_{Y_0 \mid X,W}) \\
	&\qquad\qquad \text{ and } p_{1|(\cdot)} \text{ satisfy equation } \eqref{eq:dist_compatibility} \text{ and  Assumption \ref{assn:joint_cdep}}\}.
\end{align*}
\end{definition}

In our later derivations, we sometimes refer to the identified set for $(F_{Y_1 \mid W},F_{Y_0 \mid W},C_{1,0 \mid X,W})$ instead, which we denote by $\mathcal{I}_0^{i}(F_{Y,X,W};c)$ for $i \in \{\text{marg}, \text{joint} \}$. Via equation \eqref{eq:relationship_F_YXW_FYW}, this set can be viewed as an affine transformation of the identified set for $(F_{Y_1 \mid X,W}, F_{Y_0 \mid X,W},C_{1,0 \mid X,W})$.

We now show some key properties of the cdfs and copulas in these identified sets. We begin with marginal $c$-dependence.
\begin{theorem}\label{thm:cdf_sharp_margcdep}
Let Assumptions \ref{assn:overlap} and \ref{assn:marginal_cdep} hold. For all $(\varepsilon,\gamma) \in [0,1]^2$ and for any $C_{1,0 \mid X,W} \in \mathcal{C}_{1,0 \mid X,W}$, 
\begin{align*}
	\left(\varepsilon \underline{F}_{Y_1 \mid W} + (1-\varepsilon) \overline{F}_{Y_1 \mid W},\gamma \underline{F}_{Y_0 \mid W} + (1-\gamma) \overline{F}_{Y_0 \mid W}, C_{1,0 \mid X,W}\right) \in \mathcal{I}_0^{\text{marg}}(F_{Y,X,W};c).
\end{align*}	
\end{theorem}

This theorem shows that the four pairs of cdfs $(\underline{F}_{Y_1 \mid W},\underline{F}_{Y_0 \mid W})$, $(\overline{F}_{Y_1 \mid W},\underline{F}_{Y_0 \mid W})$, $(\underline{F}_{Y_1 \mid W},\overline{F}_{Y_0 \mid W})$, and $(\overline{F}_{Y_1 \mid W},\overline{F}_{Y_0 \mid W})$ are part of the identified set. This is obtained by varying $(\varepsilon,\gamma)$ over $\{(1,1),(0,1),(1,0),(0,0)\}$.  We show this by explicitly constructing latent propensity scores $p_1(Y_1,w)$ and $p_0(Y_0,w)$ that lie in $[\lc(w,\eta),\uc(w,\eta)]$ almost surely under the implied distribution of $(Y_1,Y_0) \mid W = w$. These propensity scores have a switching structure where they equal the lower/upper bound for low values of $Y_x$ and the upper/lower bound for large values of $Y_x$. For example, the propensity score $p_1(Y_1,w)$ associated with cdf upper bound $\overline{F}_{Y_1 \mid W}$ equals
\begin{align*}
	p_1(Y_1,w) &\coloneqq \begin{cases}\lc(w,\eta) &\text{ if } Y_1 < \overline{Q}_1\\
	\overline{A}_1 &\text{ if } Y_1 = \overline{Q}_1\\
	\uc(w,\eta) &\text{ if } Y_1 > \overline{Q}_1
	\end{cases}
\end{align*}
where
\begin{align*}
	\overline{Q}_1 &\coloneqq Q_{Y \mid X,W}\left(\frac{(\uc(w,\eta) - p_{1|w})\lc(w,\eta)}{(\uc(w,\eta) - \lc(w,\eta))p_{1|w}} \mid X=1,W = w\right)
\end{align*}
and
\begin{align*}
	\overline{A}_1 &\coloneqq \frac{\P(Y = \overline{Q}_1, X=1 \mid W=w)}{\overline{F}_{Y_1 \mid W}(\overline{Q}_1 \mid w) - \overline{F}_{Y_1 \mid W}(\overline{Q}_1- \mid w)}.
\end{align*}
Note that $\overline{A}_1 \in [\lc(w,\eta),\uc(w,\eta)]$. We denote $\lim_{q \nearrow \overline{Q}_1} \overline{F}_{Y_1 \mid W}(q \mid w)$ by $\overline{F}_{Y_1 \mid W}(\overline{Q}_1- \mid w)$. The propensity scores associated with the cdf bounds $\underline{F}_{Y_1 \mid W}$, $\overline{F}_{Y_0 \mid W}$, and $\underline{F}_{Y_0 \mid W}$ can all be found in Appendix \ref{appendix:notation}.

This switching structure of the latent propensity score was observed for conditional $c$-dependence by \citet[pages 335--339]{MastenPoirier2018}, and for the MSM in Proposition 2 of \cite{DornGuo2023}. Our sharpness proof implies that latent propensity score $p_1(Y_1,w)$ satisfies an integral constraint, namely that $\E[p_1(Y_1,W) \mid W = w]= \P(X=1 \mid W = w)$, in order to ensure it is compatible with the observed propensity score.

Theorem \ref{thm:cdf_sharp_margcdep} also shows that any convex linear combinations of these four cdf pairs lies in the identified set. As a result, the identified set for $F_{Y_x \mid W}(y \mid w)$ is the entire closed interval $[\underline{F}_{Y_x \mid W}(y \mid w),\overline{F}_{Y_x \mid W}(y \mid w)]$. Moreover, the identified set for the pair $(F_{Y_1 \mid W}(y \mid w),F_{Y_0 \mid W}(y \mid w))$ is the Cartesian product of their individual identified sets, meaning that fixing or knowing the conditional distribution of one potential outcome does not affect the identified set of the distribution of the other potential outcome.

Finally, Theorem \ref{thm:cdf_sharp_margcdep} proves that no conditional copulas are ruled out by marginal $c$-dependence. For example, marginal $c$-dependence allows $Y_1$ and $Y_0$ to be independent, comonotonic\footnote{This is also referred to as \textit{rank invariance}. For example, see the discussion in \cite{HeckmanSmithClements1997}.}, or countermonotonic given $X$ and $W$. 

A similar result is obtained under joint $c$-dependence.

\begin{theorem}\label{thm:cdf_sharp_jointcdep}
	Let Assumptions \ref{assn:overlap} and \ref{assn:joint_cdep} hold. For all $(\varepsilon,\gamma) \in [0,1]^2$ there exists $C_{1,0 \mid X,W} \in \mathcal{C}_{1,0 \mid X,W}$ such that
\begin{align*}
	(\varepsilon \underline{F}_{Y_1 \mid W} + (1-\varepsilon) \overline{F}_{Y_1 \mid W},\gamma \underline{F}_{Y_0 \mid W} + (1-\gamma) \overline{F}_{Y_0 \mid W}, C_{1,0 \mid X,W}) \in \mathcal{I}_0^{\text{joint}}(F_{Y,X,W};c).
\end{align*}	
\end{theorem}

The only difference between the two theorems concerns the dependence structures between $Y_1$ and $Y_0$. Theorem \ref{thm:cdf_sharp_margcdep} shows that all copulas are compatible with marginal $c$-dependence, while our proof of Theorem \ref{thm:cdf_sharp_jointcdep} only exhibits one copula for each pair of conditional distributions $(\varepsilon \underline{F}_{Y_1 \mid W} + (1-\varepsilon) \overline{F}_{Y_1 \mid W},\gamma \underline{F}_{Y_0 \mid W} + (1-\gamma) \overline{F}_{Y_0 \mid W})$.

\subsection{Bounds on Monotonic Parameters} \label{subsec:monotonic_parameter_bounds}

The sharp bounds provided in theorems \ref{thm:cdf_sharp_margcdep} and \ref{thm:cdf_sharp_jointcdep} can be used to deliver analytical expressions for sharp bounds on a large class of treatment effect parameters. We first define the identified set for a parameter $\theta$ defined in \eqref{eq:TE_param_general}.
\begin{definition}[Parameter Identified Sets]\label{def:param_ID_sets}
	Let $\theta(F_{Y_1 \mid X,W},F_{Y_0 \mid X,W},C_{1,0 \mid X,W},F_{Y,X,W})$ be a parameter. Its identified set under marginal $c$-dependence with bounds $c \coloneqq (\lc(w,\eta),\uc(w,\eta))$ is given by
	\begin{align*}
		\mathcal{I}^\text{marg}_{\theta}(F_{Y,X,W};c) &\coloneqq \left\{\theta(F_1,F_0,C,F_{Y,X,W}):(F_1, F_0,C) \in \mathcal{I}^{\text{marg}}(F_{Y,X,W};c)\right\}.
	\end{align*}
	For a parameter $\theta(F_{Y_1 \mid X,W},F_{Y_0 \mid X,W},F_{Y,X,W})$ that does not depend on the copula $C_{1,0 \mid X,W}$, its identified set under joint $c$-dependence is given by
	\begin{align*}
		\mathcal{I}^j_{\theta}(F_{Y,X,W};c) &\coloneqq \left\{\theta(F_1,F_0,F_{Y,X,W}):(F_1, F_0,C) \in \mathcal{I}^{\text{joint}}(F_{Y,X,W};c) \text{ for some } C \in \mathcal{C}_{1,0 \mid X,W}\right\}.
	\end{align*}
\end{definition}

These sets are the parameter values consistent with the known distribution of observables $F_{Y,X,W}$ and with a $c$-dependence condition. Without restrictions on how $\theta$ depends on the distribution of potential outcomes, these sets may take various shapes.

We focus on a class of scalar estimands that can be ordered with respect to first order stochastic dominance.

\begin{definition}[First-Order Stochastic Dominance]
Let $\mathcal{F}$ be the set of all univariate cdfs and let $F, G \in \mathcal{F}$. Say that $F$ first-order stochastically dominates $G$, denoted by $F \succeq G$, if $F(u) \leq G(u)$ for all $u \in \R$.
\end{definition}

Next we define our target class of parameters.

\begin{definition}[Monotonic Parameters]
Let $\theta: \mathcal{F} \rightarrow \R$ be a parameter. Say that $\theta$ is \textit{increasing} if $F \succeq G$ implies $\theta(F) \geq \theta(G)$. Say that $\theta$ is \textit{decreasing} if $-\theta$ is increasing, and say $\theta$ is \textit{monotonic} if it is either increasing or decreasing.
\end{definition}

Following \cite{Manski1997a}, monotonic parameters are also called $D$-parameters, or $D_1$-parameters. Also see \cite{Manski2003} or \cite{Stoye2010} who consider parameters that are increasing with respect to second-order stochastic dominance. 

As an example, consider a parameter $\theta(F_{Y_1 \mid W})$ that is increasing in $F_{Y_1 \mid W}$ and suppose $c$-dependence holds. Then
\begin{align*}
	\theta(F_{Y_1 \mid W}) \in \left[\theta(\overline{F}_{Y_1 \mid W}),\theta(\underline{F}_{Y_1 \mid W})\right]
\end{align*}
since $\overline{F}_{Y_1 \mid W} \preceq  F_{Y_1 \mid W} \preceq \underline{F}_{Y_1 \mid W}$, which holds by Lemma \ref{lemma:cdf_bounds_margcdep}. This interval cannot be made narrower since the cdf bounds $[\underline{F}_{Y_1 \mid W},\overline{F}_{Y_1 \mid W}]$ are sharp by theorems \ref{thm:cdf_sharp_margcdep} and \ref{thm:cdf_sharp_jointcdep}. Therefore, the identified set for $\theta(F_{Y_1 \mid W})$ is a subset of this closed interval that always contains its two endpoints. The interior of this interval is also part of the identified set if the functional $\theta$ is continuous in the sense that $\varepsilon \mapsto \theta(\varepsilon \underline{F}_{Y_1 \mid W} + (1-\varepsilon) \overline{F}_{Y_1 \mid W})$ is continuous. This type of continuity is implied by the continuity of the mapping $F \mapsto \theta(F)$ under the sup-distance metric. 

Assuming monotonicity of the parameter will help derive properties of its identified set. Monotonicity is a substantive restriction, but all eleven parameters from Section \ref{subsec:setup} satisfy it. This is formally established in Lemma \ref{lemma:monotonic parameters} below. We begin by considering monotonic parameters that do not depend on copulas.

\begin{theorem}\label{prop:IDset_monotonic_param}
Let $\theta(F_{Y_1 \mid W},F_{Y_0 \mid W},F_{Y,X,W})$ be increasing in $F_{Y_1 \mid W}(\cdot \mid w)$ and decreasing in $F_{Y_0 \mid W}(\cdot \mid w)$ for each $w\in\supp(W)$. Let Assumption \ref{assn:overlap} hold, and either Assumption \ref{assn:marginal_cdep} or \ref{assn:joint_cdep} hold. Then, for $i \in \{ \text{marg}, \text{joint} \}$ the convex hull of the identified set for $\theta(F_{Y_1 \mid W},F_{Y_0 \mid W},F_{Y,X,W})$ is the closed interval 
	\begin{align*}
		\mathcal{I}^i_{\theta}(F_{Y,X,W};c)
		&= \left[\inf_{(F_1,F_0,C) \in \mathcal{I}_0^{i}(F_{Y,X,W};c)}\theta(F_1,F_0,F_{Y,X,W}), \sup_{(F_1,F_0,C) \in \mathcal{I}^{i}_0(F_{Y,X,W};c)}\theta(F_1,F_0,F_{Y,X,W})\right]\\ 
		 &= \left[\theta(\overline{F}_{Y_1 \mid W},\underline{F}_{Y_0 \mid W},F_{Y,X,W}),\theta(\underline{F}_{Y_1 \mid W},\overline{F}_{Y_0 \mid W},F_{Y,X,W})\right].
	\end{align*}
	If $(\varepsilon,\gamma) \mapsto \theta(\varepsilon \underline{F}_{Y_1 \mid W} + (1-\varepsilon) \overline{F}_{Y_1 \mid W},\gamma \underline{F}_{Y_0 \mid W} + (1-\gamma) \overline{F}_{Y_0 \mid W},F_{Y,X,W})$ is continuous over $(\varepsilon,\gamma) \in [0,1]^2$, this interval equals the identified set.
\end{theorem}

This theorem shows that substituting the upper/lower cdf bounds delivers sharp bounds for any parameter that is monotonic in the first-order stochastic dominance sense. The result is derived under an assumption that the parameter is increasing in $F_{Y_1 \mid W}$ and decreasing in $F_{Y_0 \mid W}$, but it immediately generalizes to parameters that are increasing or decreasing in either or both conditional cdfs. For example, the cdf pair $(\overline{F}_{Y_1 \mid W},\underline{F}_{Y_0 \mid W})$ will maximize a parameter that is decreasing in $F_{Y_1 \mid W}$ and increasing in $F_{Y_0 \mid W}$, and the cdf pair $(\overline{F}_{Y_1 \mid W},\overline{F}_{Y_0 \mid W})$ will maximize (minimize) a parameter that is decreasing (increasing) in both $F_{Y_1 \mid W}$ and $F_{Y_0 \mid W}$. The identified set for these parameters always contains endpoints $\theta(\overline{F}_{Y_1 \mid W},\underline{F}_{Y_0 \mid W},F_{Y,X,W})$ and $\theta(\underline{F}_{Y_1 \mid W},\overline{F}_{Y_0 \mid W},F_{Y,X,W})$. It also contains all the values between these endpoints whenever the mapping $\theta$ is continuous in the appropriate sense.

We document the monotonicity of various building blocks for parameters of interest in the following technical lemma. We omit covariates $W$ for simplicity here, except in part 4 on QCATE because that parameter requires covariates to be nontrivial.

\begin{lemma}\label{lemma:monotonic parameters}
	Let Assumption \ref{assn:overlap} hold. Then, for $x \in \{0,1\}$ and $\tau \in (0,1)$,
	\begin{enumerate}
		\item $\theta_{\E}(F_{Y_x}) \coloneqq \int y \, dF_{Y_x}(y)$ is increasing and continuous in the sense that $\varepsilon \mapsto \theta_{\E}(\varepsilon F_{Y_x} + (1-\varepsilon)F_{Y_x}')$ is continuous for any $(F_{Y_x},F_{Y_x}')$ over $\varepsilon \in [0,1]$.
		
		\item $\theta_{Q}(F_{Y_x};\tau) \coloneqq F_{Y_x}^{-1}(\tau)$ is increasing.
		
		\item $\theta_{CQ}(F_{Y_x};\tau) \coloneqq F_{Y_x \mid X}^{-1}(\tau \mid 1-x)$ is increasing.
		
		\item $\theta_\text{QCATE}(F_{Y_1 \mid W},F_{Y_0 \mid W},F_W;\tau)$ (see Example \ref{ex:CATE_dist}) is decreasing in $F_{Y_1 \mid W}$ and increasing in $F_{Y_0 \mid W}$.
		
		\item $\theta_{\text{CDF}}(F_{Y_1},F_{Y_0},C;y_1,y_0) \coloneqq C(F_{Y_1}(y_0),F_{Y_0}(y_0))$ is decreasing in $F_{Y_1}$ and $F_{Y_0}$ for all $(y_1,y_0) \in \R^2$ and copulas $C$.
		
		\item $\theta_{\text{DTE}}(F_{Y_1},F_{Y_0},C;z) \coloneqq \int_{\{y_1 - y_0 \leq z\}} \, dC(F_{Y_1}(y_1),F_{Y_0}(y_0))$ is decreasing in $F_{Y_1}$ and increasing in $F_{Y_0}$ for all $z \in \R$ and copulas $C$.
	\end{enumerate}
\end{lemma}

Using this lemma, all eight parameters that are invariant to copulas are bounded by substituting the upper or lower cdf bounds from Theorem \ref{lemma:cdf_bounds_margcdep}. This allows us to compute analytical bounds for these parameters. 

\section{Analytical Bounds on Treatment Effect Parameters}\label{sec:mainTreatmentEffectBounds}

We explore these analytical bounds by focusing on five of our examples to illustrate these expressions. The first three parameters are independent from the copula, while the last two are copula-dependent.

\subsection{Average Treatment Effects (Example \ref{ex:ATE})}

From Lemma \ref{lemma:monotonic parameters}.1, we have that the ATE satisfies
\begin{align*}
	\E[Y_1 - Y_0] &= \theta_{\text{ATE}}(F_{Y_1 \mid W},F_{Y_0 \mid W},F_W) \in [\theta_{\text{ATE}}(\overline{F}_{Y_1 \mid W},\underline{F}_{Y_0 \mid W},F_W) , \theta_{\text{ATE}}(\underline{F}_{Y_1 \mid W},\overline{F}_{Y_0 \mid W},F_W)].
\end{align*}
This interval equals the identified set by the monotonicity and continuity of the expectation functional which was established in Lemma \ref{lemma:monotonic parameters}.1. The lower and upper bounds can be obtained by calculating $\int y \, d\overline{F}_{Y_x \mid W}(y \mid w)$ and $\int y \, d\underline{F}_{Y_x \mid W}(y \mid w)$ for $x \in \{0,1\}$. Via the quantile transformation, these bounds can also be written as integrals of $\underline{Q}_{Y_x \mid W}(u \mid w)$ and $\overline{Q}_{Y_x \mid W}(u \mid w)$ over $u \in (0,1)$. Thus the ATE bounds can be written as integrals of quantiles. Via Lemma \ref{lemma:expectation bounds} in Appendix \ref{appendix:bounds_exp_outcomes}, we show that these quantile integrals can be converted into conditional expectations of outcomes given that they exceed or fall short of a fixed conditional quantile. These are equivalent to Conditional Value at Risk (CVaR) measures that appear in \cite{DornGuoKallus2024}. These bounds are stated explicitly in equations \eqref{eq:Y1ubExp}--\eqref{eq:Y0lbExp} in Appendix \ref{appendix:bounds_exp_outcomes} in the general case. When $Y \mid X,W$ is continuously distributed, we obtain simpler expressions for these bounds that we give here:
\begin{align*}
	\theta_{\text{ATE}}(\overline{F}_{Y_1 \mid W},\underline{F}_{Y_0 \mid W},F_W) &= \E\left[\left(\E[Y \mid Y \leq \overline{Q}_1, X=1,W] - \E[Y \mid Y \leq \underline{Q}_0, X=0, W] \right)\frac{\uc - p_{1 \mid W}}{\uc - \lc}\right]\\
	& +   \E\left[\left(\E[Y \mid Y > \overline{Q}_1, X=1, W] - \E[Y \mid Y > \underline{Q}_0, X=0,W ]\right) \frac{p_{1 \mid W} - \lc}{\uc - \lc}\right]
\end{align*}
and
\begin{align*}
	\theta_{\text{ATE}}(\underline{F}_{Y_1 \mid W},\overline{F}_{Y_0 \mid W},F_W) &= \E\left[\left(\E[Y \mid Y \leq \underline{Q}_1, X=1,W] - \E[Y \mid Y \leq \overline{Q}_0, X=0, W] \right)\frac{p_{1 \mid W} - \lc}{\uc - \lc}\right]\\
	& +   \E\left[\left(\E[Y \mid Y > \underline{Q}_1, X=1, W] - \E[Y \mid Y > \overline{Q}_0, X=0,W ]\right) \frac{\uc - p_{1 \mid W}}{\uc - \lc}\right].
\end{align*}
Note that the dependence of $(\lc,\uc)$ on $(W,\eta)$ was suppressed for convenience.

\subsection{Quantile Treatment Effects (Example \ref{ex:QTE})}

We now consider bounds on the quantile treatment effect for a fixed quantile $\tau \in (0,1)$. By Lemma \ref{lemma:monotonic parameters}.2, the functional $\theta_\text{QTE}$ is increasing in $F_{Y_1 \mid W}$ and decreasing in $F_{Y_0 \mid W}$.  Therefore, by Theorem \ref{prop:IDset_monotonic_param}, $\text{QTE}(\tau)$ has the following sharp bounds:
\begin{align*}
	\text{QTE}(\tau) \in \left[\underline{Q}_{Y_1}(\tau) - \overline{Q}_{Y_0}(\tau), \overline{Q}_{Y_1}(\tau) - \underline{Q}_{Y_0}(\tau)\right]
\end{align*}
where $\overline{Q}_{Y_x}$ is the left-inverse of cdf $\underline{F}_{Y_x}(\cdot) \coloneqq \E[\underline{F}_{Y_x \mid W}(\cdot \mid W)]$ for $x \in \{0,1\}$. Analogously, $\underline{Q}_{Y_x}$ is the left-inverse of cdf $\overline{F}_{Y_x}(\cdot) \coloneqq \E[\overline{F}_{Y_x \mid W}(\cdot \mid W)]$. Analytical expressions for the unconditional cdf bounds for the treated potential outcome are given by
\begin{align*}
	\overline{F}_{Y_1}(y) &= \E\left[\min\left\{F_{Y \mid X,W}(y \mid 1,W)\frac{p_{1 \mid W}}{\lc}, \frac{\uc - p_{1 \mid W}}{\uc} + F_{Y \mid X,W}(y \mid 1,W)\frac{p_{1 \mid W}}{\uc}\right\}\right]\\
	\underline{F}_{Y_1}(y) &= \E\left[\max\left\{F_{Y \mid X,W}(y \mid 1,W)\frac{p_{1 \mid W}}{\uc}, \frac{\lc - p_{1 \mid W}}{\lc} + F_{Y \mid X,W}(y \mid 1,W)\frac{p_{1 \mid W}}{\lc}\right\}\right]
\end{align*}
and similar expressions can be obtained for $Y_0$. The left-inverses of the previous expressions yield bounds on quantiles of $Y_1$ and $Y_0$, and which can be used to compute the QTE bounds.

\subsection{Average Weighted Welfare (Example \ref{ex:AWW})}

Consider a policy $\omega: \supp(W) \to [0,1]$ that treats units with covariate value $w$ with probability $\omega(w)$. The average welfare in a population under such policy is given by
\begin{align*}
	\text{AWW}(\omega) &= \theta_\text{AWW}(F_{Y_1 \mid W},F_{Y_0 \mid W},F_W,\omega) = \E[\omega(W) \E[Y_1 \mid W] + (1 - \omega(W)) \E[Y_0 \mid W]].
\end{align*}
By adapting Lemma \ref{lemma:monotonic parameters}.1, this functional is increasing in $F_{Y_1 \mid W}$, increasing in $F_{Y_0 \mid W}$, and continuous in the sense defined in the lemma. Therefore, by Theorem \ref{prop:IDset_monotonic_param}, its identified set is the closed interval given by
\begin{align*}
	\left[\theta_\text{AWW}(\overline{F}_{Y_1 \mid W},\overline{F}_{Y_0 \mid W},F_W,\omega),\theta_\text{AWW}(\underline{F}_{Y_1 \mid W},\underline{F}_{Y_0 \mid W},F_W,\omega)\right].
\end{align*}
An analytical expression for these bounds can be obtained by substituting in the expressions for the cdf bounds in the previous functionals. When $Y \mid X,W$ is continuously distributed, the bounds are given by
\begin{align*}
	&\theta_\text{AWW}(\overline{F}_{Y_1 \mid W},\overline{F}_{Y_0 \mid W},F_W,\omega)\\
	&= \E\left[\omega(W)\left(\E[Y \mid Y \leq \overline{Q}_1, X=1,W]\frac{\uc - p_{1 \mid W}}{\uc - \lc} + \E[Y \mid Y > \overline{Q}_1, X=1, W]\frac{p_{1 \mid W} - \lc}{\uc - \lc}\right)\right]\\
	&+ \E\left[(1-\omega(W))\left(\E[Y \mid Y \leq \overline{Q}_0, X=0, W]\frac{p_{1 \mid W} - \lc}{\uc - \lc} + \E[Y \mid Y > \overline{Q}_0, X=0,W] \frac{\uc - p_{1 \mid W}}{\uc - \lc}\right)\right]
\end{align*}
and
\begin{align*}
	&\theta_\text{AWW}(\underline{F}_{Y_1 \mid W},\underline{F}_{Y_0 \mid W},F_W,\omega)\\
	&= \E\left[\omega(W)\left(\E[Y \mid Y \leq \underline{Q}_1, X=1,W]\frac{p_{1 \mid W} - \lc}{\uc - \lc} + \E[Y \mid Y > \underline{Q}_1, X=1, W]\frac{\uc - p_{1 \mid W}}{\uc - \lc}\right)\right]\\
	&+ \E\left[(1-\omega(W))\left(\E[Y \mid Y \leq \underline{Q}_0, X=0, W]\frac{\uc - p_{1 \mid W}}{\uc - \lc} + \E[Y \mid Y > \underline{Q}_0, X=0,W] \frac{p_{1 \mid W} - \lc}{\uc - \lc}\right)\right].
\end{align*}

\subsection{Copula-Dependent Parameters}

We now consider identification of the parameters in examples \ref{ex:Joint_dist}--\ref{ex:QDTE} which all depend on the copulas $C_{1,0 \mid X,W}$. Even under unconfoundedness these parameters are not point-identified. Relaxing unconfoundedness will yield larger identified sets for these parameters when compared to the unconfoundedness baseline. We will focus on marginal $c$-dependence since it does not restrict the dependence structure between the potential outcomes.

\subsubsection*{The Joint Distribution Function}

Consider identification of the joint cdf $F_{Y_1,Y_0}(y_1,y_0)$ under marginal $c$-dependence. Consider the functional
\begin{align*}
  &\theta_\text{CDF}(F_{Y_1 \mid X,W},F_{Y_0 \mid X,W},C_{1,0 \mid X,W},F_{Y,X,W};y_1,y_0) \\
  &\coloneqq \int \left(C_{1,0 \mid X,W}(F_{Y \mid X,W}(y_1 \mid 1,w),F_{Y_0 \mid X,W}(y_0 \mid 1,w) \mid 1,w)p_{1|w} \right.\\
  &\qquad\qquad + \left. C_{1,0 \mid X,W}(F_{Y_1 \mid X,W}(y_1 \mid 0,w),F_{Y \mid X,W}(y_0|0,w)|0,w)p_{0|w}\right) dF_W(w).
\end{align*}
Fix the conditional copula function $C_{1,0 \mid X,W}(\cdot, \cdot \mid \cdot, \cdot)$. Then by Lemma \ref{lemma:monotonic parameters}.4, this functional is decreasing in $F_{Y_0 \mid X,W}(y_0 \mid 1,w)$ and $F_{Y_1 \mid X,W}(y_1 \mid 0,w)$. Thus it is bounded below by
\[
	\theta_\text{CDF}(\underline{F}_{Y_1 \mid X,W},\underline{F}_{Y_0 \mid X,W},C_{1,0 \mid X,W},F_{Y,X,W};y_1,y_0)
\]
and above by
\[
	\theta_\text{CDF}(\overline{F}_{Y_1 \mid X,W},\overline{F}_{Y_0 \mid X,W},C_{1,0 \mid X,W},F_{Y,X,W};y_1,y_0).
\]
Moreover, by Theorem \ref{prop:IDset_monotonic_param}, these bounds are sharp.

Since $C_{1,0 \mid X,W}$ is unknown, we then compute the maximum and minimum of these bounds over the set of copulas that are consistent with marginal $c$-dependence; this is simply the set of all copulas. The Fr\'echet-Hoeffding bounds  show that all copulas $C$ satisfy
\begin{align*}
	C(u,v) \in [\max\{u + v - 1,0\},\min\{u,v\}] \eqqcolon [\underline{C}(u,v),\overline{C}(u,v)]
\end{align*}
for all $(u,v) \in [0,1]^2$. The copula bounds $\underline{C}$ and $\overline{C}$ are themselves copulas. Combining these facts, we obtain the following analytical bounds on the joint cdf of potential outcomes.
\begin{proposition}[Identified set for joint cdf] \label{prop:joint_cdf_bounds}
Let assumptions \ref{assn:overlap} and \ref{assn:marginal_cdep} hold. Then, for any $(y_1,y_0) \in \R^2$, the identified set for $F_{Y_1,Y_0}(y_1,y_0)$ is given by the closed interval
\begin{equation}
\label{eq:IDset_jointcdf}
\begin{aligned}
	\mathcal{I}^\text{marg}_{\theta_\text{CDF}}(F_{Y,X,W};c)&= \left[\E\left(\max\{\underline{F}_{Y_1 \mid X,W}(y_1 \mid X,W)+\underline{F}_{Y_0 \mid X,W}(y_0 \mid X,W) - 1,0\}\right), \right. \\
	&\qquad\qquad\qquad\qquad \left.\E\left(\min\{\overline{F}_{Y_1 \mid X,W}(y_1 \mid X,W),\overline{F}_{Y_0 \mid X,W}(y_0 \mid X,W)\}\right)\right].\end{aligned}
\end{equation}
\end{proposition}

\medskip

The bounds in \eqref{eq:IDset_jointcdf} are themselves cdfs, so these bounds can be attained simultaneously for all $(y_1,y_0) \in \R^2$. The bounds for $F_{Y_1,Y_0}$ under unconfoundedness are obtained as a special case when $\lc = p_{1 \mid W} = \uc$, which implies that $\overline{F}_{Y_x \mid X,W} = \underline{F}_{Y_x \mid X,W} = F_{Y \mid X,W}(\cdot \mid x,\cdot)$. Making this substitution in equation \eqref{eq:IDset_jointcdf} yields these bounds under unconfoundedness.

\subsubsection*{The Distribution of Treatment Effects}

Identification of this parameter under unconfoundedness was studied in \cite{FanPark2010}, by applying results first shown in \cite{Makarov1982} and later studied in \cite{WilliamsonDowns1990}. \cite{MastenPoirier2019BF} also studied this parameter under conditional $c$-dependence and under a range of assumptions on copulas for $(Y_1,Y_0)$. By Lemma 2.1 in \cite{FanPark2010}, the cdf of $Y_1 - Y_0$  given $(X,W) = (x,w)$ satisfies
\begin{align*}
	F_{Y_1 - Y_0 \mid X,W}(z \mid x,w) &\in \left[\max\left\{\sup_{y \in \R} \left(F_{Y_1 \mid X,W}(y \mid x,w) - F_{Y_0 \mid X,W}(y - z \mid x,w)\right),0\right\},\right.\\
	& \qquad \left. 1 + \min\left\{\inf_{y \in \R} \left(F_{Y_1 \mid X,W}(y \mid x,w) - F_{Y_0 \mid X,W}(y - z \mid x,w)\right),0\right\}\right]
\end{align*}
and these bounds are sharp for any pair of cdfs $(F_{Y_1 \mid X,W},F_{Y_0 \mid X,W})$. These bounds are decreasing in $F_{Y_1 \mid X,W}$ and increasing in $F_{Y_0 \mid X,W}$, therefore substituting the upper/lower cdf bounds for $F_{Y_x \mid X,W}$ results in sharp bounds for $F_{Y_1 - Y_0 \mid X,W}$ under $c$-dependence. This was established in Lemma \ref{lemma:monotonic parameters}.6. Integrating these bounds over the marginal distribution of $(X,W)$ yields sharp bounds for the unconditional cdf of $Y_1 - Y_0$. This result is summarized in the following proposition.
\begin{proposition}[Identified set for DTE] \label{prop:DTE_bounds}
	Let assumptions \ref{assn:overlap} and \ref{assn:marginal_cdep} hold. For any $z \in \R$, the convex hull of the identified set for $F_{Y_1 - Y_0}(z)$ is given by
	\begin{align*}
		\mathcal{I}^\text{marg}_{\theta_\text{DTE}}(F_{Y,X,W};c)&= \left[\E\left(\max\left\{\sup_{y \in \R} \left(\underline{F}_{Y_1 \mid X,W}(y \mid X,W) - \overline{F}_{Y_0 \mid X,W}(y - z \mid X,W)\right),0\right\}\right),\right.\\
		& \qquad \left. 1 + \E\left(\min\left\{\inf_{y \in \R} \left(\overline{F}_{Y_1 \mid X,W}(y \mid X,W) - \underline{F}_{Y_0 \mid X,W}(y - z \mid X,W)\right),0\right\}\right)\right].
	\end{align*}
\end{proposition}

This expression involves two one-dimensional optimization problems, but the objective functions are known, closed-form functionals of the distribution of the observables. Bounds on the QDTE can be obtained as a corollary by taking the left-inverse of the cdf bounds.

\section{Conclusion}\label{sec:conclusion}

In this paper we proposed a general class of relaxations of unconfoundedness, and showed how it includes several previous approaches as special cases. We then derived closed form identification results for many different target parameters under this general class of relaxations. There are at least three natural next steps. First, in this paper we focused on population level identification results. Corresponding estimation and inference results can likely be derived by using standard sample analog estimators and arguments, but we leave the details to future work. Second, it would be interesting to explore whether our bounds have either the double-sharpness or double-validity properties defined in \cite{DornGuoKallus2024}, and if not, whether alternative bounds that had these properties could be derived. Third, it would be interesting to extend our results to independence assumptions beyond unconfoundedness, such as IV exogeneity (e.g., section 4 of \citealt{MastenPoirier2018a}).

\bibliographystyle{econometrica}
\bibliography{JELBibliography}

\begin{thebibliography}{40}
\newcommand{\enquote}[1]{``#1''}
\expandafter\ifx\csname natexlab\endcsname\relax\def\natexlab#1{#1}\fi

\bibitem[\protect\citeauthoryear{Ash and Dol{\'e}ans-Dade}{Ash and
  Dol{\'e}ans-Dade}{2000}]{Ash2000}
\textsc{Ash, R.~B. and C.~A. Dol{\'e}ans-Dade} (2000): \emph{Probability and
  Measure Theory}, Academic Press.

\bibitem[\protect\citeauthoryear{Billingsley}{Billingsley}{1995}]{Billingsley1995}
\textsc{Billingsley, P.} (1995): \emph{Probability and Measure}, John Wiley \&
  Sons, 3rd ed.

\bibitem[\protect\citeauthoryear{Bonvini and Kennedy}{Bonvini and
  Kennedy}{2022}]{BonviniKennedy2022}
\textsc{Bonvini, M. and E.~H. Kennedy} (2022): \enquote{Sensitivity analysis
  via the proportion of unmeasured confounding,} \emph{Journal of the American
  Statistical Association}, 117, 1540--1550.

\bibitem[\protect\citeauthoryear{Chernozhukov, Fern{\'a}ndez-Val, and
  Luo}{Chernozhukov et~al.}{2018}]{ChernozhukovFernandez-ValLuo2018}
\textsc{Chernozhukov, V., I.~Fern{\'a}ndez-Val, and Y.~Luo} (2018):
  \enquote{The sorted effects method: Discovering heterogeneous effects beyond
  their averages,} \emph{Econometrica}, 86, 1911--1938.

\bibitem[\protect\citeauthoryear{Ding and VanderWeele}{Ding and
  VanderWeele}{2016}]{DingVanderWeele2016}
\textsc{Ding, P. and T.~J. VanderWeele} (2016): \enquote{Sensitivity analysis
  without assumptions,} \emph{Epidemiology}, 27, 368.

\bibitem[\protect\citeauthoryear{Dorn and Guo}{Dorn and
  Guo}{2023}]{DornGuo2023}
\textsc{Dorn, J. and K.~Guo} (2023): \enquote{Sharp sensitivity analysis for
  inverse propensity weighting via quantile balancing,} \emph{Journal of the
  American Statistical Association}, 118, 2645--2657.

\bibitem[\protect\citeauthoryear{Dorn, Guo, and Kallus}{Dorn
  et~al.}{2024}]{DornGuoKallus2024}
\textsc{Dorn, J., K.~Guo, and N.~Kallus} (2024):
  \enquote{Doubly-valid/doubly-sharp sensitivity analysis for causal inference
  with unmeasured confounding,} \emph{Journal of the American Statistical
  Association}, 1--12.

\bibitem[\protect\citeauthoryear{Dorn and Yap}{Dorn and
  Yap}{2024}]{DornYap2024}
\textsc{Dorn, J. and L.~Yap} (2024): \enquote{Sensitivity analysis for linear
  estimands,} \emph{arXiv preprint arXiv:2309.06305}.

\bibitem[\protect\citeauthoryear{Duarte}{Duarte}{2024}]{Duarte2024}
\textsc{Duarte, G.} (2024): \enquote{A unified approach for assessing
  sensitivity to violations of causal assumptions,} \emph{Working paper}.

\bibitem[\protect\citeauthoryear{Fan and Park}{Fan and
  Park}{2010}]{FanPark2010}
\textsc{Fan, Y. and S.~S. Park} (2010): \enquote{Sharp bounds on the
  distribution of treatment effects and their statistical inference,}
  \emph{Econometric Theory}, 26, 931--951.

\bibitem[\protect\citeauthoryear{Gu, Russell, and Stringham}{Gu
  et~al.}{2024}]{GuRussellStringham2024}
\textsc{Gu, J., T.~Russell, and T.~Stringham} (2024): \enquote{Counterfactual
  identification and latent space enumeration in discrete outcome models,}
  \emph{Review of Economic Studies (forthcoming)}.

\bibitem[\protect\citeauthoryear{Heckman, Smith, and Clements}{Heckman
  et~al.}{1997}]{HeckmanSmithClements1997}
\textsc{Heckman, J.~J., J.~Smith, and N.~Clements} (1997): \enquote{Making the
  most out of programme evaluations and social experiments: Accounting for
  heterogeneity in programme impacts,} \emph{The Review of Economic Studies},
  64, 487--535.

\bibitem[\protect\citeauthoryear{Huang and Pimentel}{Huang and
  Pimentel}{2025}]{HuangPimentel2025}
\textsc{Huang, M. and S.~D. Pimentel} (2025): \enquote{Variance-based
  sensitivity analysis for weighting estimators results in more informative
  bounds,} \emph{Biometrika}, 112.

\bibitem[\protect\citeauthoryear{Imbens and Rubin}{Imbens and
  Rubin}{2015}]{ImbensRubin2015}
\textsc{Imbens, G.~W. and D.~B. Rubin} (2015): \emph{Causal Inference for
  Statistics, Social, and Biomedical Sciences}, Cambridge University Press.

\bibitem[\protect\citeauthoryear{Kallus and Zhou}{Kallus and
  Zhou}{2018}]{KallusZhou2018}
\textsc{Kallus, N. and A.~Zhou} (2018): \enquote{Confounding-robust policy
  improvement,} in \emph{Advances in Neural Information Processing Systems},
  vol.~31.

\bibitem[\protect\citeauthoryear{Makarov}{Makarov}{1982}]{Makarov1982}
\textsc{Makarov, G.} (1982): \enquote{Estimates for the distribution function
  of a sum of two random variables when the marginal distributions are fixed,}
  \emph{Theory of Probability \& its Applications}, 26, 803--806.

\bibitem[\protect\citeauthoryear{Manski}{Manski}{1997}]{Manski1997a}
\textsc{Manski, C.~F.} (1997): \enquote{Monotone treatment response,}
  \emph{Econometrica}, 1311--1334.

\bibitem[\protect\citeauthoryear{Manski}{Manski}{2003}]{Manski2003}
---\hspace{-.1pt}---\hspace{-.1pt}--- (2003): \emph{Partial Identification of
  Probability Distributions}, Springer.

\bibitem[\protect\citeauthoryear{Masten and Poirier}{Masten and
  Poirier}{2016}]{MastenPoirier2016}
\textsc{Masten, M.~A. and A.~Poirier} (2016): \enquote{Partial independence in
  nonseparable models,} \emph{cemmap working paper, CWP26/16}.

\bibitem[\protect\citeauthoryear{Masten and Poirier}{Masten and
  Poirier}{2018{\natexlab{a}}}]{MastenPoirier2018}
---\hspace{-.1pt}---\hspace{-.1pt}--- (2018{\natexlab{a}}):
  \enquote{Identification of treatment effects under conditional partial
  independence,} \emph{Econometrica}, 86, 317--351.

\bibitem[\protect\citeauthoryear{Masten and Poirier}{Masten and
  Poirier}{2018{\natexlab{b}}}]{MastenPoirier2018a}
---\hspace{-.1pt}---\hspace{-.1pt}--- (2018{\natexlab{b}}): \enquote{Salvaging
  falsified instrumental variable models,} \emph{arXiv preprint
  arXiv:1812.11598v1}.

\bibitem[\protect\citeauthoryear{Masten and Poirier}{Masten and
  Poirier}{2020}]{MastenPoirier2019BF}
---\hspace{-.1pt}---\hspace{-.1pt}--- (2020): \enquote{Inference on breakdown
  frontiers,} \emph{Quantitative Economics}, 11, 41--111.

\bibitem[\protect\citeauthoryear{Masten and Poirier}{Masten and
  Poirier}{2023}]{MastenPoirier2023EJ}
---\hspace{-.1pt}---\hspace{-.1pt}--- (2023): \enquote{Choosing exogeneity
  assumptions in potential outcome models,} \emph{The Econometrics Journal},
  26, 327--349.

\bibitem[\protect\citeauthoryear{Masten, Poirier, and Zhang}{Masten
  et~al.}{2024}]{MastenPoirierZhang2020}
\textsc{Masten, M.~A., A.~Poirier, and L.~Zhang} (2024): \enquote{Assessing
  sensitivity to unconfoundedness: Estimation and inference,} \emph{Journal of
  Business \& Economic Statistics}, 42, 1--13.

\bibitem[\protect\citeauthoryear{Nelsen}{Nelsen}{2006}]{Nelsen2006}
\textsc{Nelsen, R.~B.} (2006): \emph{An Introduction to Copulas}, Springer,
  second ed.

\bibitem[\protect\citeauthoryear{Rambachan, Coston, and Kennedy}{Rambachan
  et~al.}{2023}]{RambachanCostonKennedy2023}
\textsc{Rambachan, A., A.~Coston, and E.~Kennedy} (2023): \enquote{Robust
  design and evaluation of predictive algorithms under unobserved confounding,}
  \emph{arXiv preprint arXiv:2212.09844}.

\bibitem[\protect\citeauthoryear{Rosenbaum}{Rosenbaum}{1995}]{Rosenbaum1995}
\textsc{Rosenbaum, P.~R.} (1995): \emph{Observational Studies}, Springer.

\bibitem[\protect\citeauthoryear{Rosenbaum}{Rosenbaum}{2002}]{Rosenbaum2002}
---\hspace{-.1pt}---\hspace{-.1pt}--- (2002): \emph{Observational Studies},
  Springer, second ed.

\bibitem[\protect\citeauthoryear{Rosenbaum}{Rosenbaum}{2017}]{Rosenbaum2017}
---\hspace{-.1pt}---\hspace{-.1pt}--- (2017): \emph{Observation and Experiment:
  An Introduction to Causal Inference}, Harvard University Press.

\bibitem[\protect\citeauthoryear{Royden and Fitzpatrick}{Royden and
  Fitzpatrick}{2010}]{Royden2010}
\textsc{Royden, H. and P.~M. Fitzpatrick} (2010): \emph{Real Analysis},
  Pearson, 4th ed.

\bibitem[\protect\citeauthoryear{Sj{\"o}lander}{Sj{\"o}lander}{2024}]{Sjolander2024}
\textsc{Sj{\"o}lander, A.} (2024): \enquote{Sharp bounds for causal effects
  based on Ding and VanderWeele's sensitivity parameters,} \emph{Journal of
  Causal Inference}, 12, 20230019.

\bibitem[\protect\citeauthoryear{Sklar}{Sklar}{1959}]{Sklar1959}
\textsc{Sklar, M.} (1959): \emph{Fonctions de r{\'e}partition {\`a} n
  dimensions et leurs marges}, Universit{\'e} Paris 8.

\bibitem[\protect\citeauthoryear{Stoye}{Stoye}{2010}]{Stoye2010}
\textsc{Stoye, J.} (2010): \enquote{Partial identification of spread
  parameters,} \emph{Quantitative Economics}, 1, 323--357.

\bibitem[\protect\citeauthoryear{Tan}{Tan}{2006}]{Tan2006}
\textsc{Tan, Z.} (2006): \enquote{A distributional approach for causal
  inference using propensity scores,} \emph{Journal of the American Statistical
  Association}, 101, 1619--1637.

\bibitem[\protect\citeauthoryear{Tan}{Tan}{2024}]{Tan2024}
---\hspace{-.1pt}---\hspace{-.1pt}--- (2024): \enquote{Model-assisted
  sensitivity analysis for treatment effects under unmeasured confounding via
  regularized calibrated estimation,} \emph{Journal of the Royal Statistical
  Society Series B: Statistical Methodology}, 86, 1339--1363.

\bibitem[\protect\citeauthoryear{Torgovitsky}{Torgovitsky}{2019}]{Torgovitsky2018}
\textsc{Torgovitsky, A.} (2019): \enquote{Partial identification by extending
  subdistributions,} \emph{Quantitative Economics}, 10, 105--144.

\bibitem[\protect\citeauthoryear{van~der Vaart}{van~der
  Vaart}{2000}]{Vaart2000}
\textsc{van~der Vaart, A.~W.} (2000): \emph{Asymptotic Statistics}, Cambridge
  University Press.

\bibitem[\protect\citeauthoryear{VanderWeele and Ding}{VanderWeele and
  Ding}{2017}]{VanderWeeleDing2017}
\textsc{VanderWeele, T.~J. and P.~Ding} (2017): \enquote{Sensitivity analysis
  in observational research: Introducing the E-value,} \emph{Annals of Internal
  Medicine}, 167, 268--274.

\bibitem[\protect\citeauthoryear{Williamson and Downs}{Williamson and
  Downs}{1990}]{WilliamsonDowns1990}
\textsc{Williamson, R.~C. and T.~Downs} (1990): \enquote{Probabilistic
  arithmetic. I. Numerical methods for calculating convolutions and dependency
  bounds,} \emph{International Journal of Approximate Reasoning}, 4, 89--158.

\bibitem[\protect\citeauthoryear{Zhao, Small, and Bhattacharya}{Zhao
  et~al.}{2019}]{ZhaoSmallBhattacharya2019}
\textsc{Zhao, Q., D.~S. Small, and B.~B. Bhattacharya} (2019):
  \enquote{Sensitivity analysis for inverse probability weighting estimators
  via the percentile bootstrap,} \emph{Journal of the Royal Statistical Society
  Series B: Statistical Methodology}, 81, 735--761.

\end{thebibliography}

\newpage
\appendix
\singlespacing
\footnotesize

\titleformat{\section}
  {\large\bfseries}        
  {\thesection}{1em}{}     

\titleformat{\subsection}
  {\normalsize\bfseries}   
  {\thesubsection}{1em}{} 

\titleformat{\subsubsection}
  {\small\bfseries}       
  {\thesubsubsection}{1em}{} 

\section{Bound Expressions}
\label{appendix:notation}
In this section we summarize all the bounds that will be used in the proofs of lemmas and main results. These expressions are given for an arbitrary values of $(y,w) \in \R \times \supp(W)$. Denote $\lc \coloneqq \lc(w,\eta)$ and $\uc \coloneqq \uc(w,\eta)$, where the dependence on $w$ and $\eta$ is implicitly understood.

\subsection{Lower bound on $Y_1$}
Let
\begin{align*}
	\underline{\tau}_{1} &= %
	 \frac{(p_{1|w} - \lc)\uc}{p_{1|w}(\uc - \lc)} \qquad \text{ and } \qquad \underline{Q}_1 = Q_{Y \mid X,W}(\underline{\tau}_1 \mid 1,w)
\end{align*}
and
\begin{align*}
	&\underline{F}_{Y_1 \mid W}(y \mid w) 
	= \max\left\{F_{Y \mid X,W}(y \mid 1,w)\frac{p_{1|w}}{\uc}, \frac{\lc - p_{1|w}}{\lc} + F_{Y \mid X,W}(y \mid 1,w)\frac{p_{1|w}}{\lc}\right\}\\
	&= F_{Y \mid X,W}(y \mid 1,w)\frac{p_{1|w}}{\uc} \1\left(y < \underline{Q}_1\right) + \left(\frac{\lc - p_{1|w}}{\lc} + F_{Y \mid X,W}(y \mid 1,w)\frac{p_{1|w}}{\lc}\right)\1\left(y \geq \underline{Q}_1\right)\\[1em]
	&\underline{F}_{Y_1 \mid X,W}(y \mid 0,w) 
	= \max\left\{F_{Y \mid X,W}(y \mid 1,w) \frac{p_{1|w}(1 - \uc)}{p_{0|w} \uc},\frac{\lc - p_{1|w}}{\lc p_{0|w}} + F_{Y \mid X,W}(y \mid 1,w)\frac{p_{1|w}(1-\lc)}{p_{0|w} \lc}\right\}\\
	&= F_{Y \mid X,W}(y \mid 1,w)\frac{p_{1|w}(1 - \uc)}{p_{0|w} \uc} \1\left(y < \underline{Q}_1\right) + \left(\frac{\lc - p_{1|w}}{\lc p_{0|w}} + F_{Y \mid X,W}(y \mid 1,w)\frac{p_{1|w}(1-\lc)}{p_{0|w} \lc}\right)\1\left(y \geq \underline{Q}_1\right)\\[1em]
	&\underline{F}_{Y_1 \mid X,W}(y \mid 1,w) 
	= F_{Y \mid X,W}(y \mid 1,w).
\end{align*}
Define
\[
	\underline{p}_{1}(y,w) = \uc \1(y < \underline{Q}_1) + \underline{A}_1 \1(y = \underline{Q}_1) + \lc \1(y > \underline{Q}_1)
\]
where
\[
	\underline{A}_1 = \frac{ \P(Y = \underline{Q}_1,X=1 \mid W=w)}{\left(\frac{\lc - p_{1|w}}{\lc} + F_{Y \mid X,W}(\underline{Q}_1 \mid 1,w)\frac{p_{1|w}}{\lc}\right) - \frac{p_{1|w}}{\uc}\P(Y < \underline{Q}_1 \mid X=1,W=w)}.
\]
If the denominator is 0, set $\underline{A}_1 = p_{1|w}$. We can also derive the associated quantiles function bounds for $\tau \in (0,1)$:
\begin{align*}
	\overline{Q}_{Y_1 \mid W}(\tau \mid w) &\coloneqq \underline{F}_{Y_1 \mid W}^{-1}(\tau \mid w) = Q_{Y \mid X,W}\left(\min\left\{\frac{\uc}{p_{1|w}}\tau, \frac{p_{1|w} - \lc}{p_{1|w}} + \frac{\lc}{p_{1|w}}\tau \right\} \mid 1, w\right)\\[0.5em]
	\overline{Q}_{Y_1 \mid X,W}(\tau \mid 0,w) &\coloneqq \underline{F}_{Y_1 \mid X}^{-1}(\tau \mid 0,w) = Q_{Y \mid X,W}\left(\min\left\{\frac{\uc p_{0|w}}{p_{1|w}(1-\uc)}\tau, \frac{p_{1|w} - \lc}{p_{1|w}(1-\lc)} + \frac{\lc p_{0|w}}{p_{1|w}(1-\lc)}\tau \right\} \mid 1, w \right).
\end{align*}

\subsection{Upper bound on $Y_1$}
Let
\[
	\overline{\tau}_1 = %
	 1- \underline{\tau}_1 \qquad \text{ and } \qquad \overline{Q}_1 = Q_{Y \mid X,W}(\overline{\tau}_1 \mid 1,w)
\]
and
\begin{align*}
	&\overline{F}_{Y_1 \mid W}(y \mid w)
	= \min\left\{F_{Y \mid X,W}(y \mid 1,w)\frac{p_{1|w}}{\lc}, \frac{\uc - p_{1|w}}{\uc} + F_{Y \mid X,W}(y \mid 1,w)\frac{p_{1|w}}{\uc}\right\}\\
	&= F_{Y \mid X,W}(y \mid 1,w)\frac{p_{1|w}}{\lc} \1\left(y < \overline{Q}_1\right) + \left(\frac{\uc - p_{1|w}}{\uc} + F_{Y \mid X,W}(y \mid 1,w)\frac{p_{1|w}}{\uc}\right)\1\left(y \geq \overline{Q}_1\right)\\[1em]
	&\overline{F}_{Y_1 \mid X,W}(y \mid 0,w)
	= \min\left\{F_{Y \mid X,W}(y \mid 1,w)\frac{p_{1|w} (1- \lc)}{p_{0|w} \lc}, \frac{\uc - p_{1|w}}{\uc p_{0|w}} + F_{Y \mid X,W}(y \mid 1,w)\frac{p_{1|w}(1-\uc)}{p_{0|w} \uc}\right\}\\
	&= F_{Y \mid X,W}(y \mid 1,w)\frac{p_{1|w} (1- \lc)}{(1-p_{1|w})\lc} \1\left(y < \overline{Q}_1\right) + \left(\frac{\uc - p_{1|w}}{\uc p_{0|w}} + F_{Y \mid X}(y \mid 1)\frac{p_{1|w}(1-\uc)}{p_{0|w} \uc}\right)\1\left(y \geq \overline{Q}_1\right)\\[1em]
	&\overline{F}_{Y_1 \mid X,W}(y \mid 1,w) = F_{Y \mid X,W}(y \mid 1,w).
\end{align*}
Define
\begin{align*}
	\overline{p}_{1}(y,w) &= \lc \1(y < \overline{Q}_1) + \overline{A}_1 \1(y = \overline{Q}_1) + \uc \1(y > \overline{Q}_1)
\end{align*}
where
\begin{align*}
	\overline{A}_1 &= \frac{\P(Y = \overline{Q}_1,X=1 \mid W=w)}{\left(\frac{\uc - p_{1|w}}{\uc} + F_{Y \mid X,W}(\overline{Q}_1 \mid 1,w)\frac{p_{1|w}}{\uc}\right) - \frac{p_{1|w}}{\lc}\P(Y < \overline{Q}_1 \mid X=1,W=w)}.
\end{align*}
If the denominator is 0, set $\overline{A}_1 = p_{1|w}$. We can also derive the associated quantiles function bounds for $\tau \in (0,1)$:
\begin{align*}
	\underline{Q}_{Y_1 \mid W}(\tau \mid w) &\coloneqq \overline{F}_{Y_1 \mid W}^{-1}(\tau \mid w) = Q_{Y \mid X,W}\left(\max\left\{\frac{\lc}{p_{1|w}}\tau, \frac{p_{1|w} - \uc}{p_{1|w}} + \frac{\uc}{p_{1|w}}\tau \right\} \mid 1,w \right)\\[0.5em]
	\underline{Q}_{Y_1 \mid X,W}(\tau \mid 0,w) &\coloneqq \overline{F}_{Y_1 \mid X,W}^{-1}(\tau \mid 0,w) = Q_{Y \mid X,W}\left(\max\left\{\frac{\lc p_{0|w}}{p_{1|w}(1-\lc)}\tau, \frac{p_{1|w} - \uc}{p_{1|w}(1-\uc)} + \frac{\uc p_{0|w}}{p_{1|w}(1-\uc)}\tau \right\} \mid 1,w \right).
\end{align*}

\subsection{Lower Bound on $Y_0$}

Under $c$-dependence, we have that $\P(X=0 \mid Y_0,W=w) = 1 - \P(X=1 \mid Y_0,W=w) \in [1-\uc(w,\eta),1-\lc(w,\eta)]$. So we can take the bound expressions for $F_{Y_1 \mid W}$, swap $p_{1|w}$ and $p_{0|w}$, and swap $(\lc,\uc)$ and $(1-\uc,1-\lc)$ and get the correct expressions for the bounds for $F_{Y_0 \mid W}$. Here are all the expressions.

Let 
\begin{align*}
	\underline{\tau}_0 &= %
	 \frac{(\uc - p_{1|w})(1-\lc)}{p_{0|w}(\uc - \lc)} \qquad \text{ and } \qquad \underline{Q}_0 = Q_{Y \mid X,W}(\underline{\tau}_0|0,w)
\end{align*}
and
\begin{align*}
	&\underline{F}_{Y_0 \mid W}(y \mid w) 
	= \max\left\{F_{Y \mid X,W}(y \mid 0,w)\frac{p_{0|w}}{1-\lc}, \frac{p_{1|w} - \uc}{1-\uc} + F_{Y \mid X,W}(y \mid 0,w)\frac{p_{0|w}}{1-\uc}\right\}\\
	&= F_{Y \mid X,W}(y \mid 0,w)\frac{p_{0|w}}{1-\lc} \1\left(y < \underline{Q}_0 \right) + \left(\frac{p_{1|w} - \uc}{1-\uc} + F_{Y \mid X,W}(y \mid 0,w)\frac{p_{0|w}}{1-\uc}\right)\1\left(y \geq \underline{Q}_0\right)\\[1em]
	&\underline{F}_{Y_0 \mid X,W}(y \mid 1,w)
	= \max\left\{F_{Y \mid X,W}(y \mid 0,w)\frac{p_{0|w} \lc}{p_{1|w}(1-\lc)}, \frac{p_{1|w} - \uc}{(1-\uc)p_{1|w}} + F_{Y \mid X,W}(y \mid 0,w)\frac{p_{0|w} \uc}{p_{1|w}(1-\uc)} \right\}\\
	&=  F_{Y \mid X,W}(y \mid 0,w)\frac{p_{0|w} \lc}{p_{1|w}(1-\lc)} \1\left(y < \underline{Q}_0\right) + \left(\frac{p_{1|w} - \uc}{(1-\uc)p_{1|w}} + F_{Y \mid X,W}(y \mid 0,w)\frac{p_{0|w} \uc}{p_{1|w}(1-\uc)}\right)\1\left(y \geq \underline{Q}_0 \right)\\[1em]
	&\underline{F}_{Y_0 \mid X,W}(y \mid 0,w) = F_{Y \mid X}(y \mid 0,w).
\end{align*}
Define 
\[
    \underline{p}_{0}(y,w) = \lc \1(y < \underline{Q}_0) + (1 - \underline{A}_0) \1(y = \underline{Q}_0) + \uc \1(y > \underline{Q}_0)
\]
where
\begin{align*}
    \underline{A}_0 
    &= \frac{p_{0|w} (F_{Y \mid X,W}(\underline{Q}_0|0,w) - F_{Y \mid X,W}(\underline{Q}_0-|0,w))}{\underline{F}_{Y_0 \mid W}(\underline{Q}_0|w) - \underline{F}_{Y_0 \mid W}(\underline{Q}_0-|w)}.
\end{align*}
If the denominator is 0, set $\underline{A}_0 = p_{0|w}$. We can also derive the associated quantiles function bounds for $\tau \in (0,1)$:
\begin{align*}
	\overline{Q}_{Y_0 \mid W}(\tau \mid w) &\coloneqq \underline{F}_{Y_0 \mid W}^{-1}(\tau) = Q_{Y \mid X,W}\left(\min\left\{\frac{1-\lc}{p_{0|w}}\tau, \frac{p_{0|w} - (1-\uc)}{p_{0|w}} + \frac{1-\uc}{p_{0|w}}\tau \right\} \mid 0,w \right)\\
	\overline{Q}_{Y_0 \mid X,W}(\tau \mid 1,w) &\coloneqq \underline{F}_{Y_0 \mid X}^{-1}(\tau \mid 1) = Q_{Y \mid X,W}\left(\min\left\{\frac{(1-\lc) p_{1|w}}{p_{0|w} \lc}\tau, \frac{p_{0|w} - (1-\uc)}{p_{0|w}\uc} + \frac{p_{1|w}(1-\uc)}{p_{0|w}\uc}\tau \right\} \mid 0,w \right).
\end{align*}

\subsection{Upper bound for $Y_0$}
Let
\begin{align*}
	\overline{\tau}_0 &= %
	 1 - \underline{\tau}_0 \qquad \text{ and } \qquad \overline{Q}_0 = Q_{Y \mid X,W}(\overline{\tau}_0|0,w)
\end{align*}
and
\begin{align*}
	&\overline{F}_{Y_0 \mid W}(y \mid w)
	= \min\left\{F_{Y \mid X,W}(y \mid 0,w)\frac{p_{0|w}}{1-\uc}, \frac{p_{1|w} - \lc}{1-\lc} + F_{Y \mid X,W}(y \mid 0,w)\frac{p_{0|w}}{1-\lc}\right\}\\
	&= F_{Y \mid X,W}(y \mid 0,w)\frac{p_{0|w}}{1-\uc} \1\left(y < \overline{Q}_0\right) + \left(\frac{p_{1|w} - \lc}{1-\lc} + F_{Y \mid X,W}(y \mid 0,w)\frac{p_{0|w}}{1-\lc}\right)\1\left(y \geq \overline{Q}_0\right)\\[1em]
	&\overline{F}_{Y_0 \mid X,W}(y \mid 1,w)
	= \min\left\{F_{Y \mid X,W}(y \mid 0,w)\frac{p_{0|w} \uc}{p_{1|w}(1-\uc)}, \frac{p_{1|w} - \lc}{(1-\lc)p_{1|w}} + F_{Y \mid X,W}(y \mid 0,w)\frac{p_{0|w}\lc}{p_{1|w}(1-\lc)} \right\}\\
	&=  F_{Y \mid X,W}(y \mid 0,w)\frac{p_{0|w} \uc}{p_{1|w}(1-\uc)} \1\left(y < \overline{Q}_0\right) + \left(\frac{p_{1|w} - \lc}{(1-\lc)p_{1|w}} + F_{Y \mid X,W}(y \mid 0,w)\frac{p_{0|w}\lc}{p_{1|w}(1-\lc)}\right)\1\left(y \geq \overline{Q}_0\right).
\end{align*}
Define
\begin{align*}
	\overline{p}_{0}(y,w) &= \uc \1(y < \overline{Q}_0) + (1-\overline{A}_0)  \1(y = \overline{Q}_0) + \lc \1(y > \overline{Q}_0)
\end{align*}
where
\begin{align*}
	\overline{A}_0 &= \frac{ \P(Y = \overline{Q}_0,X=0 \mid W=w)}{\left(\frac{p_{1|w} - \lc}{1-\lc} + F_{Y \mid X,W}(\overline{Q}_0|0,w)\frac{p_{0|w}}{1-\lc}\right) - \frac{p_{0|w}}{1-\uc}\P(Y < \overline{Q}_0 \mid X=0,W=w)}.
\end{align*}
If the denominator is 0, set $\overline{A}_0 = p_{0|w}$.
We can also derive the associated quantiles function bounds for $\tau \in (0,1)$:
\begin{align*}
	\underline{Q}_{Y_0 \mid W}(\tau \mid w) &\coloneqq \overline{F}_{Y_0 \mid W}^{-1}(\tau \mid w) = Q_{Y \mid X,W}\left(\max\left\{\frac{1-\uc}{p_{0|w}}\tau, \frac{p_{0|w} - (1-\lc)}{p_{0|w}} + \frac{1-\lc}{p_{0|w}}\tau \right\} \mid 0,w \right)\\[0.5em]
	\underline{Q}_{Y_0 \mid X,W}(\tau \mid 1,w) &\coloneqq \overline{F}_{Y_0 \mid X,W}^{-1}(\tau \mid 1,w) = Q_{Y \mid X,W}\left(\max\left\{\frac{(1-\uc) p_{1|w}}{p_{0|w}\uc}\tau, \frac{p_{0|w} - (1-\lc)}{p_{0|w} \lc} + \frac{(1-\lc) p_{1|w}}{p_{0|w}\lc}\tau \right\} \mid 0,w \right).
\end{align*}

\section{Proofs for Section \ref{subsec:generalClassOfRelax}} 

\begin{proof}[Proof of Lemma \ref{lemma:joint_implies_marginal}]
	Let $x \in \{0,1\}$ and fix $w \in \supp(W)$. By the law of iterated expectations, $\E[X \mid Y_x,W=w] = \E[\E[X \mid Y_1,Y_0,W=w] \mid Y_x,W=w]$. Since $\E[X \mid Y_1,Y_0,W=w] \in [\lc(w,\eta),\uc(w,\eta)]$ almost surely, we then have that $\E[X \mid Y_x,W=w] \in [\lc(w,\eta),\uc(w,\eta)]$ almost surely as well.
\end{proof}

\begin{proof}[Proof of Proposition \ref{prop:c-dep_GMSM_equivalence}]

\textbf{Part 1:} Suppose marginal $c$-dependence holds with bound functions $[\lc(w,\eta),\uc(w,\eta)]$. Fix $(x,w) \in \{0,1\} \times \supp(W)$. We have that
\begin{align}\label{eq:proof_cdep_GMSM_equivalence_1}
	R_x(Y_x,w) &= \frac{p_x(Y_x,w)}{1 - p_x(Y_x,w)} \Big/ \frac{p_{1|w}}{1-p_{1|w}} \in \left[\frac{\lc(w,\eta)}{1-\lc(w,\eta)}\Big/ \frac{p_{1|w}}{1-p_{1|w}},\frac{\uc(w,\eta)}{1-\uc(w,\eta)}\Big/ \frac{p_{1|w}}{1-p_{1|w}}\right],
\end{align}
where the inclusion holds from the mapping $a \mapsto a/(1-a)$ being strictly increasing over $a \in (0,1)$ and from $p_x(Y_x,w) \in [\lc(w,\eta),\uc(w,\eta)] \subset (0,1)$ almost surely. We note that
\[
	\frac{\uc(w,\eta)}{1-\uc(w,\eta)}\Big/ \frac{p_{1|w}}{1-p_{1|w}} \in \left[\frac{p_{1|w}}{1-p_{1|w}}\Big/ \frac{p_{1|w}}{1-p_{1|w}},+\infty\right) = [1, +\infty)
\]
where the inclusion holds from $\uc(w,\eta) \in [p_{1|w}, 1)$. Similarly,
\[
	\frac{\lc(w,\eta)}{1-\lc(w,\eta)}\Big/ \frac{p_{1|w}}{1-p_{1|w}} \in \left(\frac{0}{1-0}\Big/ \frac{p_{1|w}}{1-p_{1|w}}, \frac{p_{1|w}}{1-p_{1|w}}\Big/ \frac{p_{1|w}}{1-p_{1|w}}\right] = (0,1]
\]
from $\lc(w,\eta) \in (0,p_{1|w}]$. We conclude that the GMSM holds with the bound functions from equation \eqref{eq:c-dep_GMSM_equivalence}. Replacing marginal $c$-dependence with joint $c$-dependence delivers the same bounds for the GJSM.

\textbf{Part 2:} Suppose the GMSM holds with bound functions $\left[\lL(w,\eta),\uL(w,\eta)\right]$. Fix $(x,w) \in \{0,1\} \times \supp(W)$. We have that
\begin{align}\label{eq:proof_cdep_GMSM_equivalence_2}
	p_x(Y_x,w) &= \frac{p_{1|w} R_x(Y_x,w)}{p_{0|w} + p_{1|w} R_x(Y_x,w)} \in \left[\frac{p_{1|w}\lL(w,\eta)}{p_{0|w} + p_{1|w}\lL(w,\eta)},\frac{p_{1|w}\uL(w,\eta)}{p_{0|w} + p_{1|w}\uL(w,\eta)}\right].
\end{align}
The equality holds from inverting the equation $R_x(Y_x,w) = \frac{p_x(Y_x,w)}{1 - p_x(Y_x,w)} \Big/ \frac{p_{1|w}}{1-p_{1|w}}$ in $p_x(Y_x,w)$. The inclusion holds from the mapping $a \mapsto a/(1+a)$ being strictly increasing for $a \in [0,+\infty)$ and from $R_x(Y_x,w) \in \left[\lL(w,\eta), \uL(w,\eta)\right] \subset (0,+\infty)$ almost surely. We note that
\[
	\frac{p_{1|w}\lL(w,\eta)}{p_{0|w} + p_{1|w}\lL(w,\eta)} \in \left(\frac{p_{1|w} \cdot 0}{p_{0|w} + p_{1|w} \cdot 0},\frac{p_{1|w}\cdot 1}{p_{0|w} + p_{1|w}\cdot 1}\right] = (0,p_{1|w}]
\]
by $\lL(w,\eta) \in (0,1]$. Similarly,
\[
	\frac{p_{1|w}\uL(w,\eta)}{p_{0|w} + p_{1|w}\uL(w,\eta)} \in \left[\frac{p_{1|w} \cdot 1}{p_{0|w} + p_{1|w} \cdot 1},1\right) = [p_{1|w},1)
\]
by $\uL(w,\eta) \in [1,+\infty]$. We conclude that marginal $c$-dependence holds with the bound functions from equation \eqref{eq:c-dep_GMSM_equivalence2}. Replacing the GMSM with the GJSM delivers the same bounds for joint $c$-dependence.
\end{proof}

\section{Proofs for Section \ref{sec:generalIdentificationResults}}

\begin{proof}[Proof of Lemma \ref{lemma:cdf_bounds_margcdep}]
By Lemma \ref{lemma:joint_implies_marginal}, it suffices to show the desired results under Assumption \ref{assn:marginal_cdep}. Let $y \in \R$ and $w\in\supp(W)$ be fixed. Note that
\begin{align*}
	\E\left[\frac{\1(Y \leq y)}{p_1(Y,w)} \mid X=1, W=w\right] p_{1|w} 
    &= \E\left[\frac{\1(Y_1 \leq y)X}{p_1(Y_1,w)}|W=w\right] \\
    &= \E\left[\frac{\1(Y_1 \leq y)\E[X \mid Y_1,W=w]}{p_1(Y_1, w)}|W=w\right] \\
    &= \E\left[\1(Y_1 \leq y)| W = w\right] \\
    &= F_{Y_1 \mid W}(y \mid w),
\end{align*}
where the second equality follows from the law of iterated expectations and the third from $p_1(Y_1,w) \geq \lc(w,\eta) > 0$ almost surely by Assumption \ref{assn:marginal_cdep}. Likewise, we have 
\[
    \Exp\left[\frac{\1(Y > y)}{p_1(Y,w)}\mid X=1, W=w\right]p_{1|w} = 1-F_{Y_1 \mid W}(y \mid w).
\]
Therefore,
\begin{align*}
	F_{Y_1 \mid W}(y \mid w) 
        &= \E\left[\frac{\1(Y \leq y)}{p_1(Y,w)} \mid X=1,W=w\right] p_{1|w} \\
        &\leq \E\left[\frac{\1(Y \leq y)}{\lc(w,\eta)} \mid X=1,W=w\right] p_{1|w} \\
        &= F_{Y \mid X,W}(y \mid 1,w)\frac{p_{1|w}}{\lc(w,\eta)}
\end{align*}
and
\begin{align*}
	F_{Y_1 \mid W}(y \mid w) 
        &= 1 - \Prob(Y_1 > y \mid W=w) \\
        &= 1 - \E\left[\frac{\1(Y > y)}{p_1(Y,w)} \mid X=1,W=w\right] p_{1|w} \\
        &\leq 1 - \E\left[\frac{\1(Y > y)}{\uc(w,\eta)} \mid X=1,W=w\right] p_{1|w}  \\
        &= \frac{\uc(w,\eta) - p_{1|w}}{\uc(w,\eta)} + F_{Y \mid X,W}(y \mid 1,w)\frac{p_{1|w}}{\uc(w,\eta)}.
\end{align*}
The inequalities follow from $p_1(Y_1,w)^{-1} \in [\uc(w,\eta)^{-1},\lc(w,\eta)^{-1}]$ almost surely. By these two inequalities, 
\[
    F_{Y_1 \mid W}(y \mid w) \leq \min\left\{F_{Y \mid X,W}(y \mid 1,w)\frac{p_{1|w}}{\lc(w,\eta)}, \frac{\uc(w,\eta) - p_{1|w}}{\uc(w,\eta)} + F_{Y \mid X,W}(y \mid 1,w)\frac{p_{1|w}}{\uc(w,\eta)}\right\} = \overline{F}_{Y_1 \mid W}(y \mid w).
\]
Similarly,
\begin{align*}
    F_{Y_1 \mid W}(y \mid w) 
    &= \E\left[\frac{\1(Y \leq y)}{p_1(Y,w)} \mid X=1,W=w\right] p_{1|w}  \\
    &\geq \E\left[\frac{\1(Y \leq y)}{\uc(w,\eta)} \mid X=1,W=w\right] p_{1|w} \\
    &= F_{Y \mid X,W}(y \mid 1,w)\frac{p_{1|w}}{\uc(w,\eta)}
\end{align*}
and
\begin{align*}
	F_{Y_1 \mid W}(y \mid w) 
        &= 1 - \Prob(Y_1 > y \mid W=w) \\
        &= 1 - \E\left[\frac{\1(Y > y)}{p_1(Y,w)} \mid X=1,W\right] p_{1|w} \\
        &\geq 1 - \E\left[\frac{\1(Y > y)}{\lc(w,\eta)} \mid X=1,W=w\right] p_{1|w} \\
        &= \frac{\lc(w,\eta) - p_{1|w}}{\lc(w,\eta)} + F_{Y \mid X,W}(y \mid 1,w)\frac{p_{1|w}}{\lc(w,\eta)}.
\end{align*}
By these two inequalities, 
\[
    F_{Y_1 \mid W}(y \mid w) \geq \max\left\{F_{Y \mid X,W}(y \mid 1,w)\frac{p_{1 \mid W}}{\uc(w,\eta)}, \frac{\lc(w,\eta) - p_{1|w}}{\lc(w,\eta)} + F_{Y \mid X,W}(y \mid 1,w)\frac{p_{1|w}}{\lc(w,\eta)}\right\} = \underline{F}_{Y_1 \mid W}(y \mid w).
\]
Therefore we have established $F_{Y_1 \mid W}(y \mid w) \in [\underline{F}_{Y_1 \mid W}(y \mid w),\overline{F}_{Y_1 \mid W}(y \mid w)]$, as desired.

Next we establish the bounds for $F_{Y_0 \mid W}(y \mid w)$. We also have that 
\[
	\E\left[\frac{\1(Y \leq y)}{1-p_{0}(Y,w)} \mid X=0,W=w\right] p_{0|w}= F_{Y_0 \mid W}(y \mid w).
\] 
Furthermore, note that $\E[(1-X) \mid Y_0, W=w] = \Prob(X=0 \mid Y_0, W=w) \in [1-\uc(w,\eta),1-\lc(w,\eta)]$. Changing $X = 1$ to $X=0$, and $(\lc(w,\eta),\uc(w,\eta))$ to $(1-\uc(w,\eta),1-\lc(w,\eta))$  yields
\begin{align*}
	F_{Y_0 \mid W}(y \mid w) &\leq  F_{Y \mid X,W}(y \mid 0,w)\frac{p_{0|w}}{1-\uc(w,\eta)}\\
	F_{Y_0 \mid W}(y \mid w) &\leq \frac{p_{1|w}-\lc(w,\eta)}{1-\lc(w,\eta)} + F_{Y \mid X,W}(y \mid 0,w)\frac{p_{0|w}}{1-\lc(w,\eta)}\\
	F_{Y_0 \mid W}(y \mid w) &\geq  F_{Y \mid X,W}(y \mid 0,w)\frac{p_{0|w}}{1-\lc(w,\eta)}\\
	F_{Y_0 \mid W}(y \mid w) &\geq \frac{p_{1|w}-\uc(w,\eta)}{1-\uc(w,\eta)} + F_{Y \mid X}(y \mid 0)\frac{p_{0|w}}{1-\uc(w,\eta)}.
\end{align*}
almost surely. 
Therefore, $F_{Y_0 \mid W}(y \mid w) \in [\underline{F}_{Y_0 \mid W}(y \mid w),\overline{F}_{Y_0 \mid W}(y \mid w)]$.
\end{proof}

\subsection{Proof of Theorem \ref{thm:cdf_sharp_margcdep}} \label{appendix_thm2_proof}

We provide and show a number of preliminary lemmas that are used to prove Theorem \ref{thm:cdf_sharp_margcdep}.  This first lemma establishes some properties of cdf bounds for $Y_x$ given $(X,W)$.

\begin{lemma}[Bounds of CDFs] \label{lemma:cdf_properties}
Let assumptions \ref{assn:overlap} and \ref{assn:marginal_cdep} hold. Then, for $x\in\{0,1\}$ and $w\in \supp(W)$, 
\begin{enumerate}
	\item The functions $\underline{F}_{Y_x \mid X,w}(\cdot \mid 1-x,w)$ and $\overline{F}_{Y_x \mid X,W}(\cdot \mid 1-x,w)$, which are defined in Appendix \ref{appendix:notation}, are cdfs;
	\item For all $y \in \R$, 
        \begin{align*}
            &\underline{F}_{Y_x \mid X,W}(y \mid 1-x,w)p_{1-x|w} + F_{Y \mid X,W}(Y \mid X,w)p_{x|w} = \underline{F}_{Y_x \mid W}(y \mid w) \\
            &\overline{F}_{Y_x \mid X,W}(y \mid 1-x,w)p_{1-x|w} + F_{Y \mid X,W}(Y \mid X,w)p_{x|w} = \overline{F}_{Y_x \mid W}(y \mid w).
        \end{align*}
\end{enumerate}
\end{lemma}

\begin{proof}[Proof of Lemma \ref{lemma:cdf_properties}]

\noindent \textbf{Proof of Part 1:} We show that $\underline{F}_{Y_1 \mid X,W}(y \mid 0,w)$ is a cdf by showing it is nondecreasing, has limits $(0,1)$ when $y$ approaches $(-\infty,+\infty)$, and is right-continuous. The same arguments can be used to deduce that $\overline{F}_{Y_1 \mid X,W}(y \mid 0,w)$, $\underline{F}_{Y_0 \mid X,W}(y \mid 1,w)$, and $\overline{F}_{Y_0 \mid X,W}(y \mid 1,w)$ are also cdfs.

The function
\begin{align*}
	\underline{F}_{Y_1 \mid X,W}(y \mid 0,w) &= \max\left\{F_{Y \mid X,W}(y \mid 1,w)\frac{p_{1|w}(1 - \uc)}{p_{0|w} \uc},\frac{\lc - p_{1|w}}{\lc p_{0|w}} + F_{Y \mid X,W}(y \mid 1,w)\frac{p_{1|w}(1-\lc)}{\lc p_{0|w}}\right\}
\end{align*}
is nondecreasing in $y$ since each of its two arguments is nondecreasing in $y$, due to $F_{Y \mid X,W}(\cdot \mid 1,w)$ being a cdf. Then note that 
\[
    \lim_{y \to \infty} \underline{F}_{Y_1 \mid X,W}(y \mid 0,w) = \max\left\{\frac{p_{1|w}(1-\uc)}{p_{0|w} \uc},1\right\} = \max\left\{\frac{p_{1|w}- p_{1|w} \uc}{\uc - p_{1|w} \uc},1\right\} = 1,
\]
where the last equality follows by $p_{1|w} \leq \uc$. Also note that 
\[
   \lim_{y \to -\infty} \underline{F}_{Y_1 \mid X,W}(y \mid 0,w) = \max\left\{0,\frac{\lc - p_{1|w}}{\lc p_{0|w}}\right\} = 0,
\]
where the last equality follows by $\lc \leq p_{1|w}$. Finally, we can see that $\underline{F}_{Y_1 \mid X,W}(y \mid 0,w)$ is right-continuous with respect to $y$ since $F_{Y \mid X}(y \mid 1,w)$ is right-continuous and by the continuity of affine transformations and of the maximum function. Therefore, $\underline{F}_{Y_1 \mid X,W}(y \mid 0,w)$ is a cdf.

\bigskip

\noindent \textbf{Proof of Part 2:}
We show the first equality with $x = 1$, and the same arguments can be used to establish the equality for other cases. For $y \in \R$, the desired result follows by the following derivations:
\begin{align*}
	&\underline{F}_{Y_1 \mid X,W}(y \mid 0,w)p_{0|w} + F_{Y \mid X,W}(y \mid 1,w)p_{1|w}  \\
    &= \max\left\{F_{Y \mid X,W}(y \mid 1,w) \frac{p_{1|w}(1 - \uc)}{p_{0|w} \uc},\frac{\lc - p_{1|w}}{\lc p_{0|w}} + F_{Y \mid X,W}(y \mid 1,w)\frac{p_{1|w}(1-\lc)}{p_{0|w} \lc}\right\}p_{0|w} + F_{Y \mid X,W}(y \mid 1,w)p_{1|w} \\
	&= \max\left\{F_{Y \mid X,W}(y \mid 1,w) p_{1|w} \left(\frac{1 - \uc}{\uc} +1\right),\frac{\lc - p_{1|w}}{\lc} + F_{Y \mid X,W}(y \mid 1,w)p_{1|w}\left(\frac{1-\lc}{\lc} + 1\right)\right\}\\
	&= \max\left\{\frac{F_{Y \mid X,W}(y \mid 1,w)p_{1|w}}{\uc},\frac{\lc - p_{1|w}}{\lc} + \frac{F_{Y \mid X,W}(y \mid 1,w)p_{1|w}}{\lc}\right\}\\
	&= \underline{F}_{Y_1 \mid W}(y \mid w).
\end{align*}
Thus the proof is complete.
\end{proof}

\begin{lemma}\label{lemma:propscore_equal_condexp}
	 Let $x \in \{0,1\}$ and $w \in \supp(W)$. Suppose $m(\cdot)$ is a Borel measurable function and $\Prob(m(Y_x) \geq \delta|W=w) = 1$ for some $\delta > 0$. The following statements are equivalent:
	\begin{enumerate}
	\item Conditional on $W=w$, the following statement holds almost surely:
	\begin{equation}
	\label{eq:prop_equality_1_marg}
		m(Y_x) = \P(X=x \mid Y_x,W=w).
	\end{equation}
	\item For all $y \in \R$, the following equality holds:
		\begin{equation}
		\label{eq:prop_equality_2_marg}
			\E\left[\frac{\1(Y_x \leq y) \1(X = x)}{m(Y_x)}|W=w\right] = \P(Y_x \leq y \mid W=w).
		\end{equation}
	\end{enumerate}
\end{lemma}

\begin{proof}[Proof of Lemma \ref{lemma:propscore_equal_condexp}]

We first prove that \eqref{eq:prop_equality_1_marg} implies \eqref{eq:prop_equality_2_marg}. This follows from the law of iterated expectations:
\begin{align*}
	\E\left[\frac{\1(Y_x \leq y) \1(X = x)}{m(Y_x)} \mid W=w\right] 
	&= \E\left[\frac{\1(Y_x \leq y) \E[\1(X=x) \mid Y_x, W=w]}{m(Y_x)} \mid W=w\right] \\
	&=  \E\left[\frac{\1(Y_x \leq y) \Prob(X=x \mid Y_x, W=w)}{m(Y_x)} \mid W=w\right] \\
	&= \E[\1(Y_x \leq y) \mid W=w] \\
	& = \P(Y_x \leq y \mid W=w),
\end{align*}
where we use \eqref{eq:prop_equality_1_marg} and the assumption $m(Y_x) \geq \delta > 0$ for the third equality.

Next, we prove that \eqref{eq:prop_equality_2_marg} implies \eqref{eq:prop_equality_1_marg}. To show this result, we first establish a few key facts:
	\begin{enumerate}

	\item By the law of iterated expectations, \eqref{eq:prop_equality_2_marg} implies
	\begin{align*}
		\E\left[\1(Y_x \leq y) \frac{\P(X=x \mid Y_x,W=w) - m(Y_x)}{m(Y_x)}|W=w\right] = 0
	\end{align*}
	for all $y \in \R^2$. 

	\item For $y \in \R$, define the preimage from a half-space on $\R^2$ as 
	\[
		I_{y} = \{\omega \in \Omega: Y_x(\omega) \leq y\},
	\]
	where $\Omega$ denotes $Y_x$'s sample space. Let $\mathcal{A} = \{I_{y}: y\in\R\}$. We then note that $\mathcal{A}$ is closed under intersection since 
	\[
		I_{y}\cap I_{y'} = I_{\min\{y,y'\}} \in \mathcal{A}\quad \text{ for any } y, y' \in \R.
	\]
	This, combined with the non-emptyness of $\mathcal{A}$, implies that $\mathcal{A}$ is a $\pi$-system.

	\item The sample space can be written as a countable union of sets in $\mathcal{A}$ since 
	\[
		\Omega = \{\omega\in\Omega: Y_x(\omega) < \infty\} = \bigcup_{n=1}^\infty I_{n}.
	\]

	\item The random variable $[\P(X=x \mid Y_x,W=w) - m(Y_x)]/m(Y_x)$ is measurable with respect to the $\sigma$-algebra generated by $Y_x$ due to the Borel measurability of $m(\cdot)$, and it is integrable since 
        \[
            \left|\frac{\P(X=x \mid Y_x,W=w) - m(Y_x)}{m(Y_x)}\right| \leq \frac{\P(X=x \mid Y_x,W=w)}{m(Y_x)} + 1 \leq \frac{1}{\delta} + 1 < +\infty,
        \]
        where the first inequality follows by trangle inequality, and the second inequality follows by the assumption that $m(Y_x) \geq \delta > 0$ almost surely.
	\end{enumerate}

	Given the above facts, it follows by \citet[Theorem 34.1]{Billingsley1995} 
        that 
	\[
		\frac{\Prob(X=x| Y_x,W=w) - m(Y_x)}{m(Y_x)} = 0 \quad \text{with probability one conditional on } W=w.
	\]
	From this equality we conclude that $\Prob(X=x \mid Y_x,W=w) = m(Y_x)$ with probability one conditional on $W=w$ since $m(Y_x) \geq \delta > 0$ almost surely.  So the desired result has been established.
\end{proof}

The following lemma is a subset of Lemma 21.1 in \cite{Vaart2000}, so we omit its proof.

\begin{lemma}[Properties of CDFs and Quantiles]\label{lemma:quantile_props}
	Let $p \in (0,1)$ and $x \in \R$. Let $F$ be a cdf and $Q(p) = \inf\{z \in \R: F(z) \geq p\}$ be its quantile function. Then,
	\begin{enumerate}
		\item $Q(p) \leq x$ if and only if $p \leq F(x)$;
		\item $F(Q(p)) \geq p$ where equality can fail only if $F$ is discontinuous at $Q(p)$;
		\item $F(Q(p)-) \leq p$.
	\end{enumerate}
\end{lemma}

This next lemma is a compendium of properties of the cdf bounds. Its results are used throughout our proofs for the main theorems.

\begin{lemma}[Preliminary Results] \label{lemma:preliminary_margcdep}
Let $w\in\supp(W)$ and suppose assumptions \ref{assn:overlap} and \ref{assn:marginal_cdep} hold with $\lc(w,\eta) < p_{1|w} < \uc(w,\eta)$. Then, for $x \in \{0,1\}$,
\begin{enumerate}
	\item $\underline{\tau}_x$, $\overline{\tau}_x \in (0,1)$;
	\item $\overline{F}_{Y_x \mid W}(y \mid w)$ is continuous at $y = \overline{Q}_x$ if and only if $\P(Y = \overline{Q}_x \mid X=x, W=w) = 0$, and $\underline{F}_{Y_x \mid W}(y \mid w)$ is continuous at $y = \underline{Q}_x$ if and only if $\P(Y = \underline{Q}_x \mid X=x,W=w) = 0$;
	\item $\overline{A}_1, \underline{A}_1, 1-\overline{A}_0, 1-\underline{A}_0 \in [\lc,\uc]$;
	\item $\overline{p}_{x}(Y_x,w), \underline{p}_{x}(Y_x,w) \in [\lc,\uc]$ almost surely;
	\item For all $y \in \R$,
	\begin{align*}
	\E\left[\frac{\1(Y \leq y)X}{\underline{p}_{1}(Y,w)}|W = w\right] = \underline{F}_{Y_1 \mid W}(y \mid w)
	\quad &\text{ and } \quad 
	\E\left[\frac{\1(Y \leq y)X}{\overline{p}_{1}(Y,w)}|W = w\right] = \overline{F}_{Y_1 \mid W}(y \mid w), \\
	\E\left[\frac{\1(Y \leq y)(1-X)}{1-\underline{p}_{0}(Y,w)}|W=w\right] = \underline{F}_{Y_0 \mid W}(y \mid w) 
	\quad &\text{ and } \quad 
	\E\left[\frac{\1(Y \leq y)(1-X)}{1-\overline{p}_{0}(Y,w)}|W=w\right] = \overline{F}_{Y_0 \mid W}(y \mid w).
	\end{align*}
 
	\item For all $(y_1,y_0) \in \R^2$, the following inequalities are equivalent:
	\begin{enumerate}
	\item $\overline{F}_{Y_x \mid W}(y_x|w) \leq \underline{F}_{Y_{1-x}|W}(y_{1-x}|w)$;
	\item $F_{Y \mid X,W}(y_x|x,w) \leq \underline{F}_{Y_{1-x} \mid X,W}(y_{1-x}|x,w)$;
	\item $\overline{F}_{Y_x \mid X,W}(y_x \mid 1-x,w) \leq F_{Y \mid X,W}(y_{1-x} \mid 1-x,w)$.
	\end{enumerate} 
	Also, the following inequalities are equivalent:
	\begin{enumerate}
	\item[(d)] $\overline{F}_{Y_x \mid W}(y_x|w) \geq \underline{F}_{Y_{1-x}|W}(y_{1-x}|w)$;
	\item[(e)] $F_{Y \mid X,W}(y_x|x,w) \geq \underline{F}_{Y_{1-x} \mid X,W}(y_{1-x}|x,w)$;
	\item[(f)] $\overline{F}_{Y_x \mid X,W}(y_x \mid 1-x,w) \geq F_{Y \mid X,W}(y_{1-x} \mid 1-x,w)$.
	\end{enumerate}
 
	\item The following inequalities are equivalent:
	\begin{enumerate}
	\item $F_{Y \mid X,W}(\overline{Q}_1- \mid 1,w) \leq \underline{F}_{Y_0 \mid X,W}(\underline{Q}_0- \mid 1,w)$;
	\item $\overline{F}_{Y_1 \mid X,W}(\overline{Q}_1-|0,w) \leq F_{Y \mid X,W}(\underline{Q}_0-|0,w)$;
	\item $\overline{F}_{Y_1 \mid W}(\overline{Q}_1- \mid w) \leq \underline{F}_{Y_0 \mid W}(\underline{Q}_0-|w)$.
	\end{enumerate}
	Also, the following inequalities are equivalent:
	\begin{enumerate}
	\item[(d)] $F_{Y \mid X,W}(\underline{Q}_1- \mid 1,w) \leq \overline{F}_{Y_0 \mid X,W}(\overline{Q}_0- \mid 1,w)$;
	\item[(e)] $\underline{F}_{Y_1 \mid X,W}(\underline{Q}_1-|0,w) \leq F_{Y \mid X,W}(\overline{Q}_0-|0,w)$;
	\item[(f)] $\underline{F}_{Y_1 \mid W}(\underline{Q}_1- \mid w) \leq \overline{F}_{Y_0 \mid W}(\overline{Q}_0-|w)$.
	\end{enumerate}

	\item The following inequalities hold:
	\[
	\max\left\{\overline{F}_{Y_1 \mid W}(\overline{Q}_1- \mid w), \underline{F}_{Y_0 \mid W}(\underline{Q}_0-|w)\right\} \leq \frac{\uc - p_{1|w}}{\uc - \lc} \leq \min\left\{\overline{F}_{Y_1 \mid W}(\overline{Q}_1 \mid w), \underline{F}_{Y_0 \mid W}(\underline{Q}_0|w)\right\}
	\]
	and
	\[
	\max\left\{\overline{F}_{Y_0 \mid W}(\overline{Q}_0-|w), \underline{F}_{Y_1 \mid W}(\underline{Q}_1- \mid w)\right\} \leq \frac{p_{1|w} - \lc}{\uc - \lc} \leq \min\left\{\overline{F}_{Y_0 \mid W}(\overline{Q}_0|w), \underline{F}_{Y_1 \mid W}(\underline{Q}_1 \mid w)\right\}.
	\]
 
	\item The following inequalities hold:
	\[
	\max\left\{{F}_{Y \mid X,W}(\overline{Q}_1- \mid 1,w), \underline{F}_{Y_0 \mid X,W}(\underline{Q}_0- \mid 1,w)\right\}
	 \leq \overline{\tau}_1 \leq 
	\min\left\{{F}_{Y \mid X,W}(\overline{Q}_1 \mid 1,w), \underline{F}_{Y_0 \mid X,W}(\underline{Q}_0 \mid 1,w)\right\}
	\]
	 and  
	\[
	 \max\left\{\overline{F}_{Y_0 \mid X,W}(\overline{Q}_0- \mid 1,w), {F}_{Y \mid X,W}(\underline{Q}_1- \mid 1,w)\right\}
	 \leq \underline{\tau}_1 \leq 
	 \min\left\{\overline{F}_{Y_0 \mid X,W}(\overline{Q}_0 \mid 1,w), {F}_{Y \mid X,W}(\underline{Q}_1 \mid 1,w)\right\}.
	\]
 
    \item For all $(y_1,y_0) \in \R^2$, the following inequalities are equivalent:
    \begin{enumerate}
    	\item $\overline{F}_{Y_x \mid W}(y_x|w) +  \overline{F}_{Y_{1-x}|W}(y_{1-x}|w) \geq 1$;
    	\item $F_{Y \mid X,W}(y_x|x,w) + \overline{F}_{Y_{1-x} \mid X,W}(y_{1-x}|x,w) \geq 1$;
    	\item $\overline{F}_{Y_x \mid X,W}(y_x \mid 1-x,w) + F_{Y \mid X,W}(y_{1-x} \mid 1-x,w) \geq 1$. 
    \end{enumerate}
    Also, the following inequalities are equivalent:
    \begin{enumerate}
    	\item[(d)] $\underline{F}_{Y_x \mid W}(y_x|w) +  \underline{F}_{Y_{1-x}|W}(y_{1-x}|w) \geq 1$;
    	\item[(e)] $F_{Y \mid X,W}(y_x|x,w) + \underline{F}_{Y_{1-x} \mid X,W}(y_{1-x}|x,w) \geq 1$;
    	\item[(f)] $\underline{F}_{Y_x \mid X,W}(y_x \mid 1-x,w) + F_{Y \mid X,W}(y_{1-x} \mid 1-x,w) \geq 1$.
    \end{enumerate}

    \item The following inequalities are equivalent:
    \begin{enumerate}
    \item $\overline{F}_{Y_1 \mid W}(\overline{Q}_1 \mid w) + \overline{F}_{Y_0 \mid W}(\overline{Q}_0-|w) \geq 1$;
    \item $F_{Y \mid X,W}(\overline{Q}_1 \mid 1,w) + \overline{F}_{Y_0 \mid X,W}(\overline{Q}_0- \mid 1,w) \geq 1$;
    \item $\overline{F}_{Y_1 \mid X,W}(\overline{Q}_1|0,w) + F_{Y \mid X,W}(\overline{Q}_0-|0,w) \geq 1$. 
    \end{enumerate}
    Also, the following inequalities are equivalent:
    \begin{enumerate}
    \item[(d)] $\overline{F}_{Y_1 \mid W}(\overline{Q}_1- \mid w) + \overline{F}_{Y_0 \mid W}(\overline{Q}_0|w) \geq 1$;
    \item[(e)] $F_{Y \mid X,W}(\overline{Q}_1- \mid 1,w) + \overline{F}_{Y_0 \mid X,W}(\overline{Q}_0 \mid 1,w) \geq 1$;
    \item[(f)] $\overline{F}_{Y_1 \mid X,W}(\overline{Q}_1-|0,w) + F_{Y \mid X,W}(\overline{Q}_0|0,w) \geq 1$.
    \end{enumerate}
\end{enumerate}
\end{lemma}

\begin{proof}[Proof of Lemma \ref{lemma:preliminary_margcdep}]

For brevity, we omit covariates $w \in \supp(W)$ and drop notation referring on conditional probability $\cdot \mid w$ and $\cdot \mid W=w$ from this proof. However, note that our arguments hold when conditioning on $W = w$ throughout. 

\bigskip

\noindent \textbf{Proof of Part 1:}
First consider $\underline{\tau}_1$:
\begin{align*}
	\underline{\tau}_1 &= \frac{p_1 - \lc}{\uc - \lc} \frac{\uc}{p_1} = \frac{p_1 \uc - \lc \uc}{p_1 \uc - \lc p_1} < \frac{p_1 \uc - \lc \uc}{p_1 \uc - \lc \uc} = 1,
\end{align*}
where the inequality is strict because $p_1 < \uc$ and $\lc > 0$. Similarly, 
\begin{align*}
	\underline{\tau}_1 &= \frac{p_1 - \lc}{\uc - \lc} \frac{\uc}{p_1} > \frac{\lc - \lc}{\uc - \lc} \frac{\uc}{p_1} = 0
\end{align*}
where the inequality is strict because $p_1 > \lc$ and $\uc > 0$. Thus, $\underline{\tau}_1 \in (0,1)$. Since $\overline{\tau}_1 = 1- \underline{\tau}_1$, $\overline{\tau}_1 \in (0,1)$ as well. The proofs for $\overline{\tau}_0$ and $\underline{\tau}_0$ are similar.

\bigskip

\noindent \textbf{Proof of Part 2:} First consider the statement involving $\overline{Q}_x$ with $x = 1$. We show the following inequality
\begin{equation}
\label{eq:Y1cdf_ub_cts}
    \frac{p_1}{\uc}\P(Y=\overline{Q}_1 \mid X=1) \leq \overline{F}_{Y_1}(\overline{Q}_1) - \overline{F}_{Y_1}(\overline{Q}_1-) \leq  \frac{p_1}{\lc}\P(Y=\overline{Q}_1 \mid X=1).
\end{equation}
From this inequality and Assumption \ref{assn:marginal_cdep} that $0<\lc\leq p_1\leq \uc$, we will conclude that $\overline{F}_{Y_1}(y)$ is continuous at $y = \overline{Q}_1$ if and only if $\P(Y = \overline{Q}_1\mid X=1) = 0$.

To show the lower bound inequality in \eqref{eq:Y1cdf_ub_cts}, note that 
\begin{align*}
	\overline{F}_{Y_1}(\overline{Q}_1) - \overline{F}_{Y_1}(\overline{Q}_1-) &= \frac{\uc - p_1}{\uc} + F_{Y \mid X}(\overline{Q}_1 \mid 1)\frac{p_1}{\uc} - \frac{p_1}{\lc}F_{Y \mid X}(\overline{Q}_1- \mid 1)\\
	&= \frac{\uc - p_1}{\uc} - p_1 F_{Y \mid X}(Q_{Y \mid X}(\overline{\tau}_1 \mid 1)- \mid 1)\frac{\uc - \lc}{\lc \uc} + \frac{p_1}{\uc}\P(Y = \overline{Q}_1 \mid X=1)\\
	&\geq \frac{\uc - p_1}{\uc} - p_1 \overline{\tau}_1 \frac{\uc - \lc}{\lc \uc} + \frac{p_1}{\uc}\P(Y = \overline{Q}_1 \mid X=1)\\
	&= \frac{p_1}{\uc}\P(Y = \overline{Q}_1 \mid X=1), 
\end{align*} 
The first line holds by the definition of $\overline{F}_{Y_1}$. The third line holds by Lemma \ref{lemma:quantile_props}.3. The last line holds by the definition of $\overline{\tau}_1$. 

Likewise, we also have the following derivation:
\begin{align*}
	\overline{F}_{Y_1}(\overline{Q}_1) - \overline{F}_{Y_1}(\overline{Q}_1-) 
        &= \frac{\uc - p_1}{\uc} + F_{Y \mid X}(\overline{Q}_1 \mid 1)\frac{p_1}{\uc} - \frac{p_1}{\lc}F_{Y \mid X}(\overline{Q}_1- \mid 1)\\
	&= \frac{\uc - p_1}{\uc} - p_1 F_{Y \mid X}(Q_{Y \mid X}(\overline{\tau}_1 \mid 1) \mid 1)\frac{\uc - \lc}{\lc \uc} + \frac{p_1}{\lc}\P(Y = \overline{Q}_1 \mid X=1)\\
	&\leq \frac{\uc - p_1}{\uc} - p_1 \overline{\tau}_1 \frac{\uc - \lc}{\lc \uc} + \frac{p_1}{\lc}\P(Y = \overline{Q}_1 \mid X=1)\\
	&= \frac{p_1}{\lc}\P(Y = \overline{Q}_1 \mid X=1), 
\end{align*} 
where we use Lemma \ref{lemma:quantile_props}.2 in the third line. This establishes the upper bound inequality in \eqref{eq:Y1cdf_ub_cts}. So the desired result follows. The proofs for the statements involving $\underline{Q}_1$, $\overline{Q}_0$, and $\underline{Q}_0$ are similar by establishing the following bounds:
\begin{equation}
\label{eq:cdf_discts_bd}
\begin{aligned}
    & \underline{F}_{Y_1}(\underline{Q}_{1}) - \underline{F}_{Y_1}(\underline{Q}_1-) \in \left[\frac{p_1}{\uc} \P(Y=\underline{Q}_1 \mid X=1), \quad \frac{p_1}{\lc} \P(Y=\underline{Q}_1 \mid X=1)\right] \\
    &\overline{F}_{Y_0}(\overline{Q}_{0}) - \overline{F}_{Y_0}(\overline{Q}_0-) \in \left[\frac{p_0}{1-\lc}\P(Y = \overline{Q}_0\mid X=0), \quad \frac{p_0}{1-\uc} \P(Y = \overline{Q}_0\mid X=0)\right]\\
    & \underline{F}_{Y_0}(\underline{Q}_{0}) - \underline{F}_{Y_0}(\underline{Q}_0-) \in \left[\frac{p_0}{1-\lc}\P(Y=\underline{Q}_0\mid X=0), \quad \frac{p_0}{1-\uc}\P(Y=\underline{Q}_0\mid X=0)\right],
\end{aligned}
\end{equation}
which can be derived by similar steps to those above.

\bigskip

\noindent \textbf{Proof of Part 3:}
First consider $\overline{A}_1$ if the denominator is positive, as $\overline{A}_1 = p_1$ is trivially bounded in $[\lc, \uc]$ by Assumption \ref{assn:marginal_cdep} if the denominator becomes zero. By its definition, we have
\begin{align*}
    \overline{A}_1
    = \frac{p_1 \P(Y = \overline{Q}_1 \mid X=1)}{\left(\frac{\uc - p_1}{\uc} + F_{Y \mid X}(\overline{Q}_1 \mid 1)\frac{p_1}{\uc}\right) - \frac{p_1}{\lc}\P(Y < \overline{Q}_1 \mid X=1)} 
    = \frac{p_1 \P(Y = \overline{Q}_1 \mid X=1)}{\overline{F}_{Y_1}(\overline{Q}_1) - \overline{F}_{Y_1}(\overline{Q}_1-)}.
\end{align*}
From the inequality in equation \eqref{eq:Y1cdf_ub_cts}, we deduce that $ \P(Y = \overline{Q}_1 \mid X=1) > 0$, and 
\[
    \frac{p_1 \P(Y = \overline{Q}_1 \mid X=1)}{\overline{F}_{Y_1}(\overline{Q}_1) - \overline{F}_{Y_1}(\overline{Q}_1-)} \in [\lc, \uc].
\]
So this concludes that $\overline{A}_1 \in [\lc, \uc]$. Similarly, the results for $\underline{A}_1$, $1-\overline{A}_0$, and $1-\underline{A}_0$ can be deduced by inequalities \eqref{eq:cdf_discts_bd}.

\bigskip

\noindent \textbf{Proof of Part 4:} These propensity scores can only take the values $\lc, \uc, \overline{A}_1, \underline{A}_1, 1-\overline{A}_0$, and $1-\underline{A}_0$. By Part 3, these values all lie in $[\lc,\uc]$.

\bigskip

\noindent \textbf{Proof of Part 5:} First we show that  $\E\left[\1(Y \leq y)X/\underline{p}_1(Y)\right] = \underline{F}_{Y_1}(y)$ for all $y \in \R$. The proof for  $\overline{F}_{Y_1}$ is similar by interchanging $\lc$ with $\uc$ and thus omitted.

To prove the desired identity, we split the analysis in three cases depending on the value of $y \in \R$. For $y < \underline{Q}_1$, we have $\1(Y\leq y)/\underline{p}_1(Y) = \1(Y\leq y)/\uc$ and thus
\begin{align*}
	\E\left[\frac{\1(Y \leq y)X}{\underline{p}_1(Y)}\right] &= \E\left[\frac{\1(Y \leq y)X}{\uc}\right] = \frac{F_{Y \mid X}(y \mid 1)p_1}{\uc} = \underline{F}_{Y_1}(y).
\end{align*}

When $y = \underline{Q}_1$, we can write
\begin{align*}
	\E\left[\frac{\1(Y \leq \underline{Q}_1)X}{\underline{p}_1(Y)}\right] 
        &= \E\left[\frac{\1(Y <\underline{Q}_1)X}{\uc}\right] +\E\left[\frac{\1(Y = \underline{Q}_1)X}{\underline{A}_1}\right]\\
	&= \frac{\P(Y < \underline{Q}_1 \mid X=1)p_1}{\uc} + \P(Y=\underline{Q}_1 \mid X=1)p_1\left(\frac{\P(Y=\underline{Q}_1 \mid X=1)p_1}{\underline{F}_{Y_1}(\underline{Q}_1) - \underline{F}_{Y_1}(\underline{Q}_1-)}\right)^{-1}\\
	&= \frac{\P(Y < \underline{Q}_1 \mid X=1)p_1}{\uc} + \underline{F}_{Y_1}(\underline{Q}_1) - \underline{F}_{Y_1}(\underline{Q}_1-)\\
	&= \underline{F}_{Y_1}(\underline{Q}_1-) + \underline{F}_{Y_1}(\underline{Q}_1) - \underline{F}_{Y_1}(\underline{Q}_1-)\\
	&= \underline{F}_{Y_1}(\underline{Q}_1).
\end{align*}

Finally, when $y > \underline{Q}_1$, we can write
\begin{align*}
	\E\left[\frac{\1(Y \leq y)X}{\underline{p}_1(Y)}\right] 
        &= \E\left[\frac{\1(Y \leq \underline{Q}_1)X}{\underline{p}_1(Y)}\right] +\E\left[\frac{\1(\underline{Q}_1< Y \leq y)X}{\lc}\right]\\
	&= \underline{F}_{Y_1}(\underline{Q}_1) + \frac{\left(F_{Y \mid X}(y \mid 1) - F_{Y \mid X}(\underline{Q}_1 \mid 1)\right)p_1}{\lc}\\
	&= \frac{\lc - p_1}{\lc} + F_{Y \mid X}(\underline{Q}_1 \mid 1)\frac{p_1}{\lc} + \frac{\left(F_{Y \mid X}(y \mid 1) - F_{Y \mid X}(\underline{Q}_1 \mid 1)\right)p_1}{\lc}\\
	&= \frac{\lc - p_1}{\lc} + F_{Y \mid X}(y \mid 1)\frac{p_1}{\lc}\\
	&= \underline{F}_{Y_1}(y).
\end{align*}
Thus the desired identity holds for all $y \in \R$.

Next we show the identity $\E\left[\1(Y\leq y)(1-X)/(1-\underline{p}_0(Y))\right] = \underline{F}_{Y_0}(y)$ for all $y \in \R$. The proof for $\overline{F}_{Y_0}$ is similar by interchanging $\lc$ with $\uc$ and thus omitted. Similar to above, we split the analysis in three cases depending on the value of $y \in \R$. For $y < \underline{Q}_0$, we have $ \indicator(Y \leq y)/(1-\underline{p}_0(Y)) = \indicator(Y \leq y)/(1-\lc)$ and thus 
\begin{align*}
    \E\left[\frac{1(Y\leq y) (1-X)}{1-\underline{p}_0(Y)}\right]
    = \E\left[\frac{1(Y\leq y) (1-X)}{1-\lc}\right] 
    = \frac{F_{Y \mid X}(y \mid 0)p_0}{1-\lc}
    = \underline{F}_{Y_0}(y).
\end{align*}
When $y = \underline{Q}_0$, we can write
\begin{align*}
     \E\left[\frac{1(Y\leq y) (1-X)}{1-\underline{p}_0(Y)}\right]
     &= \Exp\left[\frac{\1(Y < \underline{Q}_0)(1-X)}{1-\lc}\right] + \E\left[\frac{\1(Y = \underline{Q}_0) (1-X)}{\underline{A}_0}\right] \\
     &= \frac{\P(Y< \underline{Q}_0 \mid X=0)p_0}{1-\lc} + \P(Y = \underline{Q}_0 \mid X=0)p_0 \left(\frac{\P(Y = \underline{Q}_0 \mid X=0)p_0}{\underline{F}_{Y_0}(\underline{Q}_0) - \underline{F}_{Y_0}(\underline{Q}_0-)}\right)^{-1} \\
     &= \frac{\P(Y< \underline{Q}_0 \mid X=0)p_0}{1-\lc} + \underline{F}_{Y_0}(\underline{Q}_0) - \underline{F}_{Y_0}(\underline{Q}_0-) \\
     &= \underline{F}_{Y_0}(\underline{Q}_0-) + \underline{F}_{Y_0}(\underline{Q}_0) - \underline{F}_{Y_0}(\underline{Q}_0-) \\
     &= \underline{F}_{Y_0}(\underline{Q}_0).
\end{align*}
When $y > \underline{Q}_0$, we can write
\begin{align*}
    \E\left[\frac{1(Y\leq y) (1-X)}{1-p_0(Y)}\right]
    &= \E\left[\frac{1(Y\leq \underline{Q}_0)(1-X)}{1- \underline{p}_0(Y)} \right] + \E\left[\frac{\1(\underline{Q}_0 < Y \leq y)(1-X)}{1-\uc}\right] \\
    &= \underline{F}_{Y_0}(\underline{Q}_0) + \frac{\left(F_{Y \mid X}(y \mid 0) - F_{Y \mid X}(\underline{Q}_0|0)\right)p_0}{1-\uc} \\
    &= \frac{p_1 - \uc}{1 - \uc} + F_{Y \mid X}(\underline{Q}_0|0)\frac{p_0}{1-\uc} + \frac{\left(F_{Y \mid X}(y \mid 0) - F_{Y \mid X}(\underline{Q}_0|0)\right)p_0}{1-\uc} \\
    &= \frac{p_1 - \uc}{1 - \uc} + F_{Y \mid X}(y \mid 0)\frac{p_0}{1-\uc} \\
    &= \underline{F}_{Y_0}(y).
\end{align*}
Thus the desired identity has been established for all $y \in \R$.

\bigskip

\noindent \textbf{Proof of Part 6:} We begin by considering the first sequence of equivalences between (a), (b), and (c) for $x = 1$. By Lemma \ref{lemma:quantile_props}.1.
\begin{align*}
	\overline{F}_{Y_1}(y_1) \leq \underline{F}_{Y_0}(y_0) &\Longleftrightarrow \overline{R}_1(y_1) \coloneqq \overline{Q}_{Y_0}(\overline{F}_{Y_1}(y_1)) \leq y_0\\
	F_{Y \mid X}(y_1 \mid 1) \leq \underline{F}_{Y_0 \mid X}(y_0 \mid 1) &\Longleftrightarrow \overline{R}_2(y_1) \coloneqq \overline{Q}_{Y_0 \mid X}(F_{Y \mid X}(y_1 \mid 1) \mid 1) \leq y_0\\
	\overline{F}_{Y_1 \mid X}(y_1 \mid 0) \leq {F}_{Y_0 \mid X}(y_0|0) &\Longleftrightarrow \overline{R}_3(y_1) \coloneqq Q_{Y \mid X}(\overline{F}_{Y_1 \mid X}(y_1 \mid 0)|0) \leq y_0.
\end{align*}
The equivalence relationship for the statements on the left hand side holds if $\overline{R}_1(y_1) = \overline{R}_2(y_1) = \overline{R}_3(y_1)$ for all $y_1 \in \R$. By direct calculation, we can see that
\begin{align*}
	\overline{R}_2(y_1) = \overline{R}_3(y_1) = Q_{Y \mid X}\left(\min\left\{\frac{(1-\lc) p_1}{p_0 \lc}F_{Y \mid X}(y_1 \mid 1), \frac{p_0 - (1-\uc)}{p_0\uc} + \frac{p_1(1-\uc)}{p_0\uc}F_{Y \mid X}(y_1 \mid 1) \right\} \mid 0 \right)
\end{align*}
and thus it remains to show that $\overline{R}_1(y_1) = \overline{R}_2(y_1)$. Recall that
\begin{align*}
    \overline{R}_1(y_1) 
    &= \overline{Q}_{Y_0}(\overline{F}_{Y_1}(y_1))\\
    &= Q_{Y \mid X}\left(\min\left\{\frac{1-\lc}{p_0}\overline{F}_{Y_1}(y_1), \frac{p_0 - (1-\uc)}{p_0} + \frac{1-\uc}{p_0}\overline{F}_{Y_1}(y_1) \right\} \mid 0 \right).
\end{align*}
We split the proof into two cases. First consider $y_1 < \overline{Q}_1 = Q_{Y \mid X}(\overline{\tau}_1 \mid 1)$. In such case we have 
\begin{equation}
\label{eq:cdf_bound_1}
    F_{Y \mid X}(y_1 \mid 1) < \overline{\tau}_1 = \frac{(\uc - p_1)\lc}{(\uc - \lc)p_1}
\end{equation}
by Lemma \ref{lemma:quantile_props}.1, and 
\begin{equation}
\label{eq:cdf_bound_2}
    \overline{F}_{Y_1}(y_1) = F_{Y \mid X}(y \mid 1) \frac{p_1}{\lc}.
\end{equation}
Using equations \eqref{eq:cdf_bound_1} and \eqref{eq:cdf_bound_2}, it can be verified that
\[
    \frac{1-\lc}{p_0}\overline{F}_{Y_1}(y_1) < \frac{p_0 - (1-\uc)}{p_0} + \frac{1-\uc}{p_0}\overline{F}_{Y_1}(y_1).
\]
This implies 
\begin{equation}
\label{eq:R1_left}
    \overline{R}_1(y_1) = Q_{Y \mid X}\left(\frac{1-\lc}{p_0}\overline{F}_{Y_1}(y_1)|0\right) =  Q_{Y \mid X}\left(\frac{(1-\lc)p_1}{p_0 \lc}F_{Y \mid X}(y_1 \mid 1)|0\right)\qquad \text{if } y_1 < \overline{Q}_1.
\end{equation}
Next consider $y \geq \overline{Q}_1$. Then we have $F_{Y \mid X}(y_1 \mid 1) \geq \overline{\tau}_1$ by Lemma \ref{lemma:quantile_props}.1, and 
\[
    \overline{F}_{Y_1}(y_1) = \frac{\uc - p_1}{\uc} + F_{Y \mid X}(y \mid 1)\frac{p_1}{\uc}.
\]
These two implications lead to the following inequality
\[
    \frac{1-\lc}{p_0}\overline{F}_{Y_1}(y_1) \geq \frac{p_0 - (1-\uc)}{p_0} + \frac{1-\uc}{p_0}\overline{F}_{Y_1}(y_1),
\]
which further implies
\begin{equation}
\label{eq:R1_right}
\begin{aligned}
    \overline{R}_1(y_1) 
    &= Q_{Y \mid X}\left(\frac{p_0 - (1-\uc)}{p_0} + \frac{1-\uc}{p_0}\overline{F}_{Y_1}(y_1)|0\right) \\
    &= Q_{Y \mid X}\left(\frac{p_0 - (1-\uc)}{p_0\uc} + \frac{p_1(1-\uc)}{p_0\uc}F_{Y \mid X}(y_1 \mid 1)| 0 \right) \qquad \text{if } y_1 \geq \overline{Q}_1.
\end{aligned}
\end{equation}
By Lemma \ref{lemma:quantile_props}.1, $y_1 < \overline{Q}_1$ is equivalent to $F_{Y \mid X}(y_1 \mid 1) < \overline{\tau}_1$, and it is further equivalent to
\[
    \frac{(1-\lc) p_1}{p_0 \lc}F_{Y \mid X}(y_1 \mid 1) < \frac{p_0 - (1-\uc)}{p_0\uc} + \frac{p_1(1-\uc)}{p_0\uc}F_{Y \mid X}(y_1 \mid 1).
\]
From this, \eqref{eq:R1_left}, and \eqref{eq:R1_right}, we deduce that
\begin{align*}
    \overline{R}_2(y_1) 
    &= Q_{Y \mid X}\left(\min\left\{\frac{(1-\lc) p_1}{p_0 \lc}F_{Y \mid X}(y_1 \mid 1), \frac{p_0 - (1-\uc)}{p_0\uc} + \frac{p_1(1-\uc)}{p_0\uc}F_{Y \mid X}(y_1 \mid 1) \right\} \mid 0 \right) \\
    &= Q_{Y \mid X}\left(\frac{(1-\lc) p_1}{p_0 \lc}F_{Y \mid X}(y_1 \mid 1)|0\right) \1(y < \overline{Q}_1) + Q_{Y \mid X}\left(\frac{\uc - p_1}{p_0\uc} + \frac{p_1(1-\uc)}{p_0\uc}F_{Y \mid X}(y_1 \mid 1) | 0\right) \1(y \geq \overline{Q}_1) \\
    &= \overline{R}_1(y_1).
\end{align*}
Therefore, the desired conclusion for $x=1$ has been established. Similar arguments can be used to show that the same conclusion also holds for $x=0$, and to show the second set of equivalences between (d), (e), and (f).

\bigskip

\noindent \textbf{Proof of Part 7:} 
We first consider the equivalence of the statement (a), (b), and (c). We can write
\[
		\Delta_1 \coloneqq F_{Y \mid X}(\overline{Q}_1- \mid 1) - \underline{F}_{Y_0 \mid X}(\underline{Q}_0- \mid 1) = F_{Y \mid X}(\overline{Q}_1- \mid 1) - \frac{p_0 \lc F_{Y \mid X}(\underline{Q}_0-|0)}{p_1(1-\lc)}
\]
From this, we note that
	\begin{align*}
		\Delta_2 &\coloneqq \overline{F}_{Y_1}(\overline{Q}_1-) - \underline{F}_{Y_0}(\underline{Q}_0-) = \frac{p_1 F_{Y \mid X}(\overline{Q}_1- \mid 1)}{\lc} - \frac{p_0 F_{Y \mid X}(\underline{Q}_0-|0)}{1-\lc} = \frac{p_1}{\lc}\Delta_1,
        \end{align*}
and
        \begin{align*}
		\Delta_3 &\coloneqq \overline{F}_{Y_1 \mid X}(\overline{Q}_1-|0) - F_{Y \mid X}(\underline{Q}_0-|0) = \frac{p_1(1-\lc)}{p_0 \lc} F_{Y \mid X}(\overline{Q}_1- \mid 1) - F_{Y \mid X}(\underline{Q}_0-|0) = \frac{p_1(1-\lc)}{p_0 \lc}\Delta_1.
	\end{align*}
The desired result follows by noting that $\Delta_1$, $\Delta_2$, and $\Delta_3$ all have the same sign. The proof of the equivalence of statements (d), (e), and (f) is similar and thus omitted.

\bigskip

\noindent \textbf{Proof of Part 8:} We have that
\begin{align*}
	\overline{F}_{Y_1}(\overline{Q}_1-) &= \frac{F_{Y \mid X}(Q_{Y \mid X}(\overline{\tau}_1 \mid 1)- \mid 1)p_1}{\lc} \leq \overline{\tau}_1 \frac{p_1}{\lc} = \frac{\uc - p_1}{\uc - \lc}
\end{align*}
by Lemma \ref{lemma:quantile_props}.3. Similarly,
\begin{align*}
	\overline{F}_{Y_1}(\overline{Q}_1) &= \frac{\uc - p_1}{\uc} + \frac{F_{Y \mid X}(Q_{Y \mid X}(\overline{\tau}_1 \mid 1) \mid 1)p_1}{\uc} \geq \frac{\uc - p_1}{\uc} + \overline{\tau}_1 \frac{p_1}{\uc} = \frac{\uc - p_1}{\uc - \lc}.
\end{align*}
by Lemma \ref{lemma:quantile_props}.2. 
The other inequalities can be shown in a similar manner. Their derivations are thus omitted.

\bigskip

\noindent \textbf{Proof of Part 9:} We have that 
\[
    F_{Y \mid X}(\overline{Q}_1 -  \mid 1) = F_{Y \mid X}\left(Q_{Y \mid X} (\overline{\tau}_1  \mid 1) -  \mid 1\right) \leq \overline{\tau}_1
\]
by Lemma \ref{lemma:quantile_props}.3. Similarly, 
\[
    F_{Y \mid X}(\overline{Q}_1 \mid 1) = F_{Y \mid X}\left(Q_{Y \mid X} (\overline{\tau}_1 \mid 1) \mid 1\right) \geq \overline{\tau}_1
\]
by Lemma \ref{lemma:quantile_props}.2. The same arguments can be applied to $\underline{F}_{Y_0 \mid X}(\underline{Q}_0- | 1)$ and $\underline{F}_{Y_0 \mid X}(\underline{Q}_0 | 1)$. So we have
\[
    \underline{F}_{Y_0 \mid X}(\underline{Q}_0- | 1) = F_{Y \mid X}(Q_{Y \mid X}(\underline{\tau_0}|0)-|0) \frac{p_0\lc}{p_1(1-\lc)} \leq \frac{\underline{\tau_0}p_0\lc}{p_1(1-\lc)} = \frac{\lc(\uc - p_1)}{p_1(\uc - \lc)} = \overline{\tau}_1.
\]
via Lemma \ref{lemma:quantile_props}.3, and 
\[
    \underline{F}_{Y_0 \mid X}(\underline{Q}_0 | 1) = \frac{p_1 - \uc}{(1-\uc)p_1} + F_{Y \mid X}(Q_{Y \mid X}(\underline{\tau_0}|0)|0) \frac{p_0\uc}{p_1(1-\uc)} \geq  \frac{p_1 - \uc}{(1-\uc)p_1} + \frac{\underline{\tau}_0p_0\uc}{p_1(1-\uc)}  = \overline{\tau}_1.
\]
via Lemma \ref{lemma:quantile_props}.2. The proofs for the other inequalities are similar and thus omitted.

\bigskip

\noindent \textbf{Proof of Part 10:}
We begin by considering the first set of equivalences between (a), (b), and (c) when $x = 1$, and the equivalences for $x=0$ are identical. By Lemma \ref{lemma:quantile_props}.1, we have the following equivalence relationships:
\begin{align*}
	\overline{F}_{Y_1}(y_1) + \overline{F}_{Y_0}(y_0) &\geq 1 \Longleftrightarrow y_1 \geq \underline{R}_1(y_0) \coloneqq \underline{Q}_{Y_1}(1 - \overline{F}_{Y_0}(y_0))\\
	F_{Y \mid X}(y_1 \mid 1) + \overline{F}_{Y_0 \mid X}(y_0 \mid 1) &\geq 1 \Longleftrightarrow y_1 \geq \underline{R}_2(y_0) \coloneqq Q_{Y \mid X}(1 - \overline{F}_{Y_0 \mid X}(y_0 \mid 1) \mid 1)\\
	\overline{F}_{Y_1 \mid X}(y_1 \mid 0) + F_{Y \mid X}(y_0|0) &\geq 1 \Longleftrightarrow y_1 \geq \underline{R}_3(y_0) \coloneqq \underline{Q}_{Y_1 \mid X}(1 - F_{Y \mid X}(y_0|0) \mid 1).
\end{align*}
To prove the equivalence of statements on the left, it suffices to show that $\underline{R}_1(y_0) = \underline{R}_2(y_0) = \underline{R}_3(y_0)$ for all $y_0 \in \R$. First, we directly compute $\underline{R}_2(y_0)$ and $\underline{R}_3(y_0)$:
\begin{align*}
	\underline{R}_2(y_0) &= Q_{Y \mid X}(1 - \overline{F}_{Y_0 \mid X}(y_0 \mid 1) \mid 1)\\
	&= Q_{Y \mid X}\left(1 - \min\left\{\frac{p_0 \uc F_{Y \mid X}(y_0|0)}{p_1(1-\uc)}, \frac{p_1 - \lc + p_0 \lc F_{Y \mid X}(y_0|0)}{p_1(1-\lc)}\right\} \mid 1\right)\\
	&= Q_{Y \mid X}\left(\max\left\{1 - \frac{p_0 \uc F_{Y \mid X}(y_0|0)}{p_1(1-\uc)},1 - \frac{p_1 - \lc + p_0 \lc F_{Y \mid X}(y_0|0)}{p_1(1-\lc)}\right\} \mid 1\right)\\
	&= Q_{Y \mid X}\left(\max\left\{1 - \frac{p_0 \uc F_{Y \mid X}(y_0|0)}{p_1(1-\uc)}, \frac{p_1(1-\lc) - p_1 + \lc - p_0 \lc F_{Y \mid X}(y_0|0)}{p_1(1-\lc)}\right\} \mid 1\right)\\
	&= Q_{Y \mid X}\left(\max\left\{1 - \frac{p_0 \uc F_{Y \mid X}(y_0|0)}{p_1(1-\uc)}, \frac{\lc p_0(1- F_{Y \mid X}(y_0|0))}{p_1(1-\lc)}\right\} \mid 1\right),
\end{align*}
and
\begin{align*}
	\underline{R}_3(y_0) &= \underline{Q}_{Y_1 \mid X}(1 - F_{Y \mid X}(y_0|0) \mid 1)\\
	&= Q_{Y \mid X}\left(\max\left\{\frac{\lc p_0(1 - F_{Y \mid X}(y_0|0))}{p_1(1-\lc)},\frac{p_1 - \uc + \uc p_0(1-F_{Y \mid X}(y_0|0))}{p_1(1-\uc)}\right\} \mid 1\right) \\
	&= Q_{Y \mid X}\left(\max\left\{\frac{\lc p_0(1 - F_{Y \mid X}(y_0|0))}{p_1(1-\lc)},1 - \frac{ \uc p_0 F_{Y \mid X}(y_0|0))}{p_1(1-\uc)}\right\} \mid 1\right).
\end{align*}
Since $y_0$ was arbitrary, we conclude that $\underline{R}_2(y_0) = \underline{R}_3(y_0)$ for all $y_0 \in \R$.

We next establish that $\underline{R}_1(y_0) = \underline{R}_2(y_0)$ for all $y_0 \in \R$. Note that 
\begin{align*}
	\underline{R}_1(y_0) 
        &= \underline{Q}_{Y_1}(1 - \overline{F}_{Y_0}(y_0))\\
	&= Q_{Y \mid X}\left(\max\left\{\frac{\lc}{p_1}(1 - \overline{F}_{Y_0}(y_0)), \frac{p_1 - \uc}{p_1} + \frac{\uc}{p_1}(1 - \overline{F}_{Y_0}(y_0))\right\} \mid 1\right)\\
        &= Q_{Y \mid X}\left(\max\left\{\frac{\lc}{p_1}(1 - \overline{F}_{Y_0}(y_0)), 1 - \frac{\uc}{p_1}\overline{F}_{Y_0}(y_0)\right\} \mid 1\right).
 \end{align*}
If $y_0 < \overline{Q}_0$, Lemma \ref{lemma:quantile_props}.1 implies
\begin{align*}
    \overline{F}_{Y_0}(y_0) = \frac{p_0}{1-\uc} F_{Y \mid X}(y \mid 0) 
    \quad \text{and} \quad 
    F_{Y \mid X}(y_0|0) < \overline{\tau}_0 = \frac{(p_1 - \lc)(1-\uc)}{(\uc-\lc) p_0}.
\end{align*}
These two (in)equalities imply that
\[
    \frac{\lc}{p_1}(1 - \overline{F}_{Y_0}(y_0)) < 1 - \frac{\uc}{p_1}\overline{F}_{Y_0}(y_0).
\]
Then it follows that
\begin{equation}
\label{eq:R1_left_lem10}
    \underline{R}_1(y_0) =  Q_{Y \mid X}\left(1 - \frac{\uc}{p_1}\overline{F}_{Y_0}(y_0) \mid 1\right) = Q_{Y \mid X}\left(1 - \frac{\uc p_0F_{Y \mid X}(y_0|0)}{p_1(1-\uc)}| 1\right)\quad \text{if } y_0 < \overline{Q}_0.
\end{equation}
Similarly, it can be verified that 
\begin{equation}
\label{eq:R1_right_lem10}
    \underline{R}_1(y_0) = Q_{Y \mid X}\left(\frac{\lc}{p_1}(1 - \overline{F}_{Y_0}(y_0)) \mid 1\right) = Q_{Y \mid X}\left(\frac{\lc p_0(1 - F_{Y \mid X}(y_0|0))}{p_1(1-\lc)} \mid 1\right) \quad \text{if } y_0 \geq \overline{Q}_0.
\end{equation}

By Lemma \ref{lemma:quantile_props}.1, $y_0 < \overline{Q}_0$ is equivalent to $F_{Y \mid X}(y_0|0) < \overline{\tau}_0$, and it is further equivalent to 
\[
    1 - \frac{\uc p_0F_{Y \mid X}(y_0|0)}{p_1(1-\uc)} > \frac{\lc p_0(1 - F_{Y \mid X}(y_0|0))}{p_1(1-\lc)}.
\]
Therefore, we can write 
\begin{equation}
\label{eq:R2_lem10}
    \underline{R}_2(y_0) = Q_{Y \mid X}\left(1 - \frac{\uc p_0F_{Y \mid X}(y_0|0)}{p_1(1-\uc)} \mid 1\right)\1(y_0 < \overline{Q}_0) + Q_{Y \mid X}\left(\frac{\lc p_0(1 - F_{Y \mid X}(y_0|0))}{p_1(1-\lc)} \mid 1\right)\1(y_0 \geq \overline{Q}_0).
\end{equation}
By combining \eqref{eq:R1_left_lem10}, \eqref{eq:R1_right_lem10}, and \eqref{eq:R2_lem10}, we note that $\underline{R}_2(y_0) = \underline{R}_1(y_0)$ for all $y_0 \in \R$, as desired. The proof for the second set of equivalences between (d), (e), and (f) is similar and thus omitted.

\medskip

\textbf{Proof of Part 11}: We show the first set of equivalences between (a), (b), and (c), and the proof for the second set of equivalences between (d), (e), and (f) follows similar arguments and thus omitted. First, we expand
\[
    \Delta_1' 
    \coloneqq F_{Y \mid X}(\overline{Q}_1 \mid 1) + \overline{F}_{Y_0 \mid X}(\overline{Q}_0- \mid 1) - 1
    = F_{Y \mid X}(\overline{Q}_1 \mid 1) + F_{Y \mid X}(\overline{Q}_0-|0) \frac{p_0\uc}{p_1(1-\uc)} - 1.
\]
Next, note that 
\begin{align*}
    \Delta_2' 
    &\coloneqq \overline{F}_{Y_1}(\overline{Q}_1) + \overline{F}_{Y_0}(\overline{Q}_0-) - 1 \\
    &= \frac{\uc-p_1}{\uc} + F_{Y \mid X}(\overline{Q}_1 \mid 1) \frac{p_1}{\uc} + F_{Y \mid X}(\overline{Q}_0-|0) \frac{p_0}{1-\uc} - 1 \\
    &= \frac{p_1}{\uc}\left[F_{Y \mid X}(\overline{Q}_1 \mid 1) + F_{Y \mid X}(\overline{Q}_0-|0)\frac{p_0\uc}{p_1(1-\uc)} - 1\right] \\
    &= \frac{p_1}{\uc} \Delta_1',
\end{align*}
and 
\begin{align*}
    \Delta_3'
    &\coloneqq \overline{F}_{Y_1 \mid X}(\overline{Q}_1|0) + F_{Y \mid X}(\overline{Q}_0-|0) - 1 \\
    &= \frac{\uc-p_1}{\uc p_0} + F_{Y \mid X}(\overline{Q}_1 \mid 1) \frac{p_1(1-\uc)}{p_0\uc} + F_{Y \mid X}(\overline{Q}_0-|0) - 1 \\
    &= \frac{p_1(1-\uc)}{p_0\uc} \left[F_{Y \mid X}(\overline{Q}_1 \mid 1) + F_{Y \mid X}(\overline{Q}_0-|0)\frac{p_0\uc}{p_1(1-\uc)}  + \frac{\uc-p_1}{p_1(1-\uc)} - \frac{(1-p_1)\uc}{p_1(1-\uc)}\right] \\
    &= \frac{p_1(1-\uc)}{p_0\uc} \left[F_{Y \mid X}(\overline{Q}_1 \mid 1) + F_{Y \mid X}(\overline{Q}_0-|0)\frac{p_0\uc}{p_1(1-\uc)} - 1 \right] \\
    &= \frac{p_1(1-\uc)}{p_0\uc} \Delta_1'.
\end{align*}
Therefore, the desired result follows by noting that $\Delta_1'$, $\Delta_2'$, $\Delta_3'$ all have the same sign.
\end{proof}

\begin{proof}[Proof of Theorem \ref{thm:cdf_sharp_margcdep}]
Fix a $w\in\supp(W)$ and $(\varepsilon,\gamma, C_{1,0 \mid 1,w}, C_{1,0|0,w}) \in [0,1]^2\times \mathcal{C}^2$. We prove this theorem by constructing a probability distribution $\widetilde{\P}$ for $(Y_1,Y_0,X)$ conditional on $W=w$ such that for all $y \in \R$, $x \in \{0,1\}$, and $(y_1, y_0) \in \R^2$, the following conditions hold:
\begin{enumerate}
	\item $\widetilde{\P}(Y_1 \leq y \mid W = w) = \varepsilon \underline{F}_{Y_1 \mid W}(y \mid w) + (1-\varepsilon) \overline{F}_{Y_1 \mid W}(y \mid w)$ and \\ $\widetilde{\P}(Y_0 \leq y \mid W = w) = \gamma \underline{F}_{Y_0 \mid W}(y \mid w) + (1-\gamma) \overline{F}_{Y_0 \mid W}(y \mid w)$ ;
	\item $\widetilde{\P}(X=x \mid W = w) = p_{x|w}$;
	\item $\widetilde{\P}(Y_x \leq y \mid X=x,W=w) = F_{Y \mid X,W}(Y \mid X,w)$;
        \item $\widetilde{\P}(Y_1\leq y_1, Y_0\leq y_0 \mid X=x,W=w) = C_{1,0|x,w}(\widetilde{\Prob}(Y_1\leq y_1 \mid X=x,W=w), \widetilde{\Prob}(Y_1\leq y_0 \mid X=x,W=w))$;
	\item $\widetilde{\P}(X=1 \mid Y_x,W=w) \in [\lc(w,\eta),\uc(w,\eta)]$ $\widetilde{\P}$-almost surely.
\end{enumerate}

Condition 1 requires that an arbitrary convex combination of marginal cdf bounds stated in Theorem \ref{thm:cdf_sharp_margcdep} can be achieved by the constructed measure. Condition 4 then states that any bivariate copula $C_{1,0|x,w}$ is also achievable. Conditions 2 and 3 require the constructed measure generate the same distribution of $(Y,X)$ as the observed data conditional on $W=w$. Finally, Condition 5 requires the marginal $c$-dependence Assumption \ref{assn:marginal_cdep} to be satisfied for the constructed measure when conditioning on $W=w$. As a result, the constructed measure $\widetilde{\P}$ generates the marginal cdfs and copulas in Theorem \ref{thm:cdf_sharp_margcdep} and satisfies all the requirements in the definition of identified set $\mathcal{I}_0^\text{marg}(F_{Y,X,W})$.

For the conciseness of the proof, we write $C_{1,0|x,w}$ as $C_{x,w}$ for $x \in \{0,1\}$ so that subscripts of copulas denote the conditioning variables. 

Let
\begin{equation}
\label{eq:sharpness_margcdep}
\begin{aligned} 
    \widetilde{\P}(Y_1 \leq y_1, Y_0 \leq y_0, X = x \mid W = w) 
    &= x C_{1,w}(F_{Y \mid X,W}(y_1 \mid 1,w),F_0(y_0 \mid 1,w;\gamma)) p_{1|w} \\
    &\quad + (1-x) C_{0,w}(F_1(y_1 \mid 0,w;\varepsilon),F_{Y \mid X,W}(y_0|0,w)) p_{0|w}. 
\end{aligned}
\end{equation}
where
\begin{align*}
	F_0(y_0 \mid 1,w;\gamma) &= \gamma \underline{F}_{Y_0 \mid X,W}(y_0 \mid 1,w) + (1-\gamma) \overline{F}_{Y_0 \mid X,W}(y_0 \mid 1,w)\\
	F_1(y_1 \mid 0,w;\varepsilon) &= \varepsilon \underline{F}_{Y_1 \mid X,W}(y_1 \mid 0,w) + (1-\varepsilon) \overline{F}_{Y_1 \mid X,W}(y_1 \mid 0,w).
\end{align*}
Since convex combinations of cdfs are cdfs, and by Lemma \ref{lemma:cdf_properties}.1, both $F_0(\cdot \mid 1,w;\gamma)$ and $F_1(\cdot \mid 0,w;\varepsilon)$ are cdfs. By Sklar's Theorem \citep[Theorem 2.3.3]{Nelsen2006}, the expression in \eqref{eq:sharpness_margcdep} is a joint distribution function for $(Y_1,Y_0,X)$ conditional on $W=w$.

For the rest of the proof, we very conditions 1-5 for the constructed measure $\widetilde{\P}$.

\medskip

\noindent \textbf{Verifying Condition 1:}
For $y \in \R$, we can see that
\begin{align*}
	\widetilde{\P}(Y_1 \leq y \mid W=w) &= \sum_{x \in \{0,1\}}\lim_{y_0 \to +\infty} \widetilde{\P}(Y_1 \leq y, Y_0 \leq y_0, X=x \mid W=w) \\
	&= \lim_{y_0 \to +\infty} C_{1,w}(F_{Y \mid X,W}(y_1 \mid 1,w),F_0(y_0 \mid 1,w;\gamma)) p_{1|w} \\
	&\quad + \lim_{y_0 \to +\infty} C_{0,w}(F_1(y_1 \mid 0,w;\varepsilon),F_{Y \mid X,W}(y_0|0,w)) p_{0|w} \\
	&= C_{1,w}(F_{Y \mid X,W}(y_1 \mid 1,w),1) p_{1|w} + C_{0,w}(F_1(y_1 \mid 0,w;\varepsilon),1) p_{0|w} \\
	&= F_{Y \mid X,W}(y_1 \mid 1,w) p_{1|w} + F_1(y_1 \mid 0,w;\varepsilon)p_{0|w} \\
	&= \varepsilon (F_{Y \mid X,W}(y_1 \mid 1,w) p_{1|w} + \underline{F}_{Y_1 \mid X,W}(y_1 \mid 0,w)p_{0|w}) \\
	&\quad + (1-\varepsilon)(F_{Y \mid X,W}(y_1 \mid 1,w) p_{1|w} + \overline{F}_{Y_1 \mid X,W}(y_1 \mid 0,w)p_{0|w})\\
	&= \varepsilon \underline{F}_{Y_1 \mid W}(y \mid w) + (1-\varepsilon) \overline{F}_{Y_1 \mid W}(y \mid w).
\end{align*}
The third line holds since $C_{x,w}(1,u) = C_{x,w}(u,1) = u$ for $x \in \{0,1\}$ and $u \in [0,1]$. The last line holds by Lemma \ref{lemma:cdf_properties}.2. 

Likewise, 
\begin{align*}
	\widetilde{\P}(Y_0 \leq y \mid W=w) &= \sum_{x \in \{0,1\}} \lim_{y_1 \to +\infty} \widetilde{\P}(Y_1 \leq y_1, Y_0 \leq y, X=x \mid W=w)\\
	&= \gamma \underline{F}_{Y_0 \mid W}(y \mid w) + (1-\gamma) \overline{F}_{Y_0 \mid W}(y \mid w).
\end{align*}

\medskip

\noindent \textbf{Verifying Condition 2:}
For $x \in \{0,1\}$, we have that
\begin{align*}
	\widetilde{\P}(X=x \mid W=w) 
	&= \lim_{y_1,y_0 \to \infty} \widetilde{\P}(Y_1 \leq y_1, Y_0 \leq y_0, X = x \mid W=w) \\
	&= x C_{1,w}(1,1) p_{1|w} + (1-x) C_{0,w}(1,1) p_{0|w} \\
	&= x p_{1|w} + (1-x)p_{0|w} \\
	&= p_{x|w}.
\end{align*}
The third equality uses the fact that $C_{x,w}(1,1) = 1$ for $x\in\{0,1\}$.

\medskip 

\noindent \textbf{Verifying Condition 3:}
For $x \in \{0,1\}$ and $y \in \R$, we have that
\begin{align*}
	\widetilde{\P}(Y_x \leq y \mid X=x, W=w) &= \lim_{y' \to +\infty} \frac{\widetilde{\P}(Y_x \leq y, Y_{1-x} \leq y',X=x \mid W=w)}{\widetilde{\P}(X=x \mid W=w)}\\
	&= \frac{x C_{1,w}(F_{Y \mid X,W}(y \mid 1,w),1)p_{1|w} + (1-x) C_{0,w}(1,F_{Y \mid X,W}(y \mid 0,w))p_{0|w}}{p_{x|w}}\\
	&= \frac{x F_{Y \mid X,W}(y \mid 1,w)p_{1|w} + (1-x) F_{Y \mid X,W}(y \mid 0)p_{0|w}}{p_{x|w}}\\
	&= \frac{F_{Y \mid X,W}(Y \mid X,w)p_{x|w}}{p_{x|w}}\\
	&= F_{Y \mid X,W}(Y \mid X,w).
\end{align*}
The third line holds again by $C_x(1,u) = C_x(u,1) = u$ for $x \in \{0,1\}$ and $u \in [0,1]$. The last line follows by Assumption \ref{assn:overlap} that $p_{x|w} > 0$ for $x \in \{0,1\}$.

\medskip

\noindent \textbf{Verifying Condition 4:} First, following similar steps for verifying condition 3, we have
\begin{equation}
\label{eq:cond_margprob}
    \begin{aligned}
    \widetilde{\P}(Y_x \leq y \mid X=1-x,W=w)
    &= \frac{\lim_{y' \to +\infty}\widetilde{\P}(Y_x \leq y, Y_{1-x}\leq y', X=1-x \mid W=w)}{\widetilde{\P}(X=1-x \mid W=w)} \\
    &= \frac{(1-x)C_{1,w}(1,F_0(y_0 \mid 1,w;\gamma))p_{1|w} + xC_{0,w}(F_1(y_1 \mid 0,w;\epsilon), 1)p_{0|w}}{p_{1-x|w}} \\
    &= \frac{(1-x)p_{1|w} F_0(y_0 \mid 1,w;\gamma) + xp_{0|w} F_1(y_1 \mid 0,w;\epsilon)}{p_{1-x|w}} \\
    &= (1-x)F_0(y_0 \mid 1,w;\gamma) + xF_1(y_1 \mid 0,w;\epsilon).
\end{aligned}
\end{equation}
Then for $(y_1,y_0) \in \R^2$, it follows that
\begin{align*}
    &\widetilde{\P}(Y_1 \leq y_1, Y_0 \leq y_0 | X = x) \\
    &= xC_{1,w}(F_{Y \mid X,W}(y_1 \mid 1,w), F_0(y_0 \mid 1,w;\gamma)) + (1-x)C_{0,w}(F_1(y_1 \mid 0,w;\epsilon), F_{Y \mid X,W}(y_0|0,w)) \\
    &= xC_{1,w}(\widetilde{\P}(Y_1\leq y_1 \mid X=1,W=w), \widetilde{\P}(Y_0\leq y_0 \mid X=1,W=w)) \\
    &\quad + (1-x)C_{0,w}(\widetilde{\P}(Y_1\leq y_1 \mid X=0,W=w), \widetilde{\P}(Y_0\leq y_0 \mid X=0,W=w)) \\
    &= C_{x,w}(\widetilde{\P}(Y_1\leq y_1 \mid X=x,W=w), \widetilde{\P}(Y_0\leq y_0 \mid X=x,W=w)).
\end{align*}
The second line holds by Condition 3 and equation \eqref{eq:cond_margprob}.

\medskip

\noindent \textbf{Verifying Condition 5:} In this part, we establish an explicit formula of the propensity score function under $\widetilde{\P}$ and show that it is contained in $[\lc(w,\eta), \uc(w,\eta)]$ almost surely. To achieve this goal, we divide the analysis into two cases. 

\textbf{Case 1:} Consider the case where $p_{1|w} = \lc(w,\eta)$. By direct calculation,
\[
    \underline{F}_{Y_1 \mid W}(y \mid w) = \overline{F}_{Y_1 \mid W}(y \mid w) = F_{Y \mid X,W}(y \mid 1,w)
    \quad\text{and}\quad
    \underline{F}_{Y_0 \mid W}(y \mid w) = \overline{F}_{Y_0 \mid W}(y \mid w) = F_{Y \mid X,W}(y \mid 0,w).
\]
Based on condition 1 we verified above, we have
\[
    \widetilde{\P}(Y_1 \leq y_1 \mid W=w) = F_{Y \mid X,W}(y_1 \mid 1,w)
    \quad\text{and}\quad
    \widetilde{\P}(Y_0 \leq y_0 \mid W=w) = F_{Y \mid X,W}(y_1 \mid 0,w).
\]
Since $p_{1|w} = \lc(w,\eta)$, by Assumption \ref{assn:marginal_cdep}, it is straightforwardly to see that $\P(X=1 \mid Y_1,W=w) = \P(X = 1 \mid Y_0,W=w) = \lc(w,\eta)$ almost surely, which further implies
\begin{align*}
    \widetilde\E\left[\frac{\indicator[Y \leq y_1]X}{\lc(w,\eta)}|W=w\right] 
    &= \E\left[\frac{\indicator[Y \leq y_1]X}{\lc(w,\eta)}|W=w\right] \\
    &= F_{Y \mid X,W}(y_1 \mid 1,w) \\
    &= \widetilde{\P}(Y_1 \leq y_1 \mid W=w)
\end{align*}
and
\begin{align*}
    \widetilde\E\left[\frac{\indicator[Y \leq y_0](1-X)}{1-\lc(w,\eta)}|W=w\right] 
    &= \E\left[\frac{\indicator[Y \leq y_0](1-X)}{1-\lc(w,\eta)}|W=w\right] \\
    &= F_{Y \mid X,W}(y_0|0,w) \\
    &= \widetilde{\P}(Y_0 \leq y_0 \mid W=w).
\end{align*}
Following Lemma \ref{lemma:propscore_equal_condexp}, this implies $\widetilde{\E}(X \mid Y_1,W=w) = \widetilde{\E}(X \mid Y_0,W=w) = \lc(w,\eta)$ almost surely under $\widetilde{\P}$, which is naturally bounded within $[\lc(w,\eta), \uc(w,\eta)]$, as desired. The proof for the case where $p_{1|w} = \uc(w,\eta)$ follows the same argument by interchanging $\lc(w,\eta)$ with $\uc(w,\eta)$ and thus omitted.

\textbf{Case 2:} Consider the case where $\lc(w,\eta) < p_{1|w} < \uc(w,\eta)$. Define
\begin{align*}
	p_1(y,w;\varepsilon) &= \frac{1}{\varepsilon \underline{p}_1(y,w)^{-1} + (1-\varepsilon)\overline{p}_1(y,w)^{-1}},
\end{align*}
where $\overline{p}_1(y,w)$ and $\underline{p}_1(y,w)$ are defined in Appendix \ref{appendix:notation}.

By Lemma \ref{lemma:preliminary_margcdep}.4, $\underline{p}_1(Y_1,w), \overline{p}_1(Y_1,w) \in [\lc(w,\eta),\uc(w,\eta)]$ almost surely. Therefore,
\begin{align*}
	p_1(Y_1,w;\varepsilon) &= \frac{1}{\varepsilon \underline{p}_1(Y_1,w)^{-1} + (1-\varepsilon)\overline{p}_1(Y_1,w)^{-1}} \leq  \frac{1}{\varepsilon \uc(w,\eta)^{-1} + (1-\varepsilon)\uc(w,\eta)^{-1}} = \uc(w,\eta)
\end{align*}
and 
\begin{align*}
	p_1(Y_1,w;\varepsilon) &= \frac{1}{\varepsilon \underline{p}_1(Y_1,w)^{-1} + (1-\varepsilon)\overline{p}_1(Y_1,w)^{-1}} \geq  \frac{1}{\varepsilon \lc(w,\eta)^{-1} + (1-\varepsilon)\lc(w,\eta)^{-1}} = \lc(w,\eta).
\end{align*}
Therefore $p_1(Y_1,w;\varepsilon) \in [\lc(w,\eta),\uc(w,\eta)]$ almost surely. 

Next we will show that $\widetilde{\E}[X \mid Y_1,W=w] = p_1(Y_1,w;\varepsilon)$ via Lemma \ref{lemma:propscore_equal_condexp} by verifying that 
\[
\widetilde{\E}\left[\frac{\1(Y_1 \leq y)X}{p_1(Y_1,w;\varepsilon)}|W=w\right] = \widetilde{\P}(Y_1 \leq y \mid W=w) = \varepsilon \underline{F}_{Y_1 \mid W}(y \mid w) + (1-\varepsilon) \overline{F}_{Y_1 \mid W}(y \mid w), \text{ for all } y \in \R.
\] 
To show this, we have the following derivations:
\begin{align*}
	\widetilde{\E}\left[\frac{\1(Y_1 \leq y)X}{p_1(Y_1,w;\varepsilon)}|W=w\right] 
	&= \E\left[\frac{\1(Y \leq y)X}{p_1(Y,w;\varepsilon)}|W=w\right]\\
	&= \E\left[\1(Y \leq y)X\left(\frac{\varepsilon}{\underline{p}_1(Y,w)} + \frac{1-\varepsilon}{\overline{p}_1(Y,w)}\right) \mid W=w\right]\\
	&= \varepsilon \E\left[\frac{\1(Y \leq y)X}{\underline{p}_1(Y,w)}|W=w\right] + (1-\varepsilon) \E\left[\frac{\1(Y \leq y)X}{\overline{p}_1(Y,w)}|W=w\right]\\
	&= \varepsilon \underline{F}_{Y_1 \mid W}(y \mid w) + (1-\varepsilon) \overline{F}_{Y_1 \mid W}(y \mid w).
\end{align*}
The first equality holds by noting that the distribution of $Y_1$ conditional $X=1$ and $W=w$ under $\widetilde{\P}$ is the same as the one under the population $\P$ as verified by condition 2. The last equality follows by applying Lemma \ref{lemma:preliminary_margcdep}.5. 

For the cdf of $Y_0$, define
\begin{align*}
	p_0(Y_0,w;\gamma) &= 1 - \frac{1}{\gamma (1-\underline{p}_0(Y_0,w))^{-1} + (1-\gamma)(1-\overline{p}_0(Y_0,w))^{-1}}.
\end{align*}
Since $1 - \underline{p}_0(Y_0,w), 1-\overline{p}_0(Y_0,w) \in [1-\uc(w,\eta),1-\lc(w,\eta)]$ almost surely, we have that
\begin{align*}
	p_0(Y_0,w;\gamma) 
	&= 1 - \frac{1}{\gamma (1-\underline{p}_0(Y_0,w))^{-1} + (1-\gamma)(1-\overline{p}_0(Y_0,w))^{-1}} \\
	&\leq 1 - \frac{1}{\gamma (1-\uc(w,\eta))^{-1} + (1-\gamma)(1-\uc(w,\eta))^{-1}} \\
	&= \uc(w,\eta),
\end{align*}
and
\begin{align*}
	p_0(Y_0,w;\gamma) 
	&= 1 - \frac{1}{\gamma (1-\underline{p}_0(Y_0,w))^{-1} + (1-\gamma)(1-\overline{p}_0(Y_0,w))^{-1}} \\
	&\geq 1 - \frac{1}{\gamma (1-\lc(w,\eta))^{-1} + (1-\gamma)(1-\lc(w,\eta))^{-1}} \\
	&= \lc(w,\eta).
\end{align*}
Therefore, $p_0(Y_0,w;\eta) \in [\lc(w,\eta),\uc(w,\eta)]$ almost surely. We can also see that
\begin{align*}
	\widetilde{\E}\left[\frac{\1(Y_0 \leq y)(1-X)}{1 - p_0(Y_0,w;\gamma)}|W=w\right] 
	&= \E\left[\frac{\1(Y \leq y)(1-X)}{1 - p_0(Y,w;\eta)}|W=w\right] \\
	&= \E\left[\1(Y \leq y)(1-X)\left(\frac{\gamma}{1 - \underline{p}_0(Y_0,w)} + \frac{1-\gamma}{1-\overline{p}_0(Y_0,w)}\right) \mid W=w\right]\\
	&= \gamma \E\left[\frac{\1(Y \leq y)(1-X)}{1 - \underline{p}_0(Y_0,w)}|W=w\right] + (1-\gamma)\E\left[\frac{\1(Y \leq y)(1-X)}{1 - \overline{p}_0(Y_0,w)}|W=w\right]\\
	&= \gamma \underline{F}_{Y_0 \mid W}(y \mid w) + (1-\gamma) \overline{F}_{Y_0 \mid W}(y \mid w),
\end{align*}
where the last equality follows by Lemma \ref{lemma:preliminary_margcdep}.5. Therefore, by Lemma \ref{lemma:propscore_equal_condexp}, $\widetilde{\P}(X=1 \mid Y_0,W=w) = p_0(Y_0,w;\eta) \in [\lc(w,\eta),\uc(w,\eta)]$ almost surely, which concludes the proof.
\end{proof}

\subsection{Proof of Theorem \ref{thm:cdf_sharp_jointcdep}}

This appendix provides a proof of Theorem \ref{thm:cdf_sharp_jointcdep} and all of its auxiliary lemmas.  
First, we define four latent propensity score functions. For $w\in\supp(W)$, let
\begin{align}\label{eq:UL_propscore}
	p^{ul}(y_1,y_0,w;B) &= \begin{cases}
							\lc &\text{ if } y_1 \leq \overline{Q}_1, y_0 \leq \underline{Q}_0, (y_1,y_0) \neq (\overline{Q}_1,\underline{Q}_0)\\
							B &\text{ if } (y_1,y_0) = (\overline{Q}_1,\underline{Q}_0)\\
							\uc &\text{ if } y_1 \geq \overline{Q}_1, y_0 \geq  \underline{Q}_0, (y_1,y_0) \neq (\overline{Q}_1,\underline{Q}_0)\\
							p_{1|w} &\text{ otherwise,}
						\end{cases}
\end{align}

\begin{align}\label{eq:UU_propscore}
	p^{uu}(y_1,y_0,w;B) &= \begin{cases}
							\lc &\text{ if } y_1 \leq \overline{Q}_1, y_0 \geq \overline{Q}_0, (y_1,y_0) \neq (\overline{Q}_1,\overline{Q}_0)\\
							B &\text{ if } (y_1,y_0) = (\overline{Q}_1,\overline{Q}_0)\\
							\uc &\text{ if } y_1 \geq \overline{Q}_1, y_0 \leq \overline{Q}_0, (y_1,y_0) \neq (\overline{Q}_1,\overline{Q}_0)\\
							p_{1|w} &\text{ otherwise,}
						\end{cases}
\end{align}

\begin{align}\label{eq:LU_propscore}
	p^{lu}(y_1,y_0,w;B) &= \begin{cases}
							\uc &\text{ if } y_1\leq \overline{Q}_1, y_0 \leq,\underline{Q}_0, (y_1,y_0) \neq (\overline{Q}_1,\underline{Q}_0)\\
							B &\text{ if } (y_1,y_0) = (\overline{Q}_1,\underline{Q}_0)\\
							\lc &\text{ if } y_1 \geq \overline{Q}_1, y_0 \geq \underline{Q}_0, (y_1,y_0) \neq (\overline{Q}_1,\underline{Q}_0)\\
							p_{1|w} &\text{ otherwise,}
						\end{cases}
\end{align}

\begin{align}\label{eq:LL_propscore}
	p^{ll}(y_1,y_0,w;B) &= \begin{cases}
							\uc &\text{ if } y_1 \leq \overline{Q}_1, y_0 \geq \overline{Q}_0, (y_1,y_0) \neq (\overline{Q}_1,\overline{Q}_0)\\
							B &\text{ if } (y_1,y_0) = (\overline{Q}_1,\overline{Q}_0)\\
							\lc &\text{ if } y_1 \geq \overline{Q}_1, y_0 \leq \overline{Q}_0, (y_1,y_0) \neq (\overline{Q}_1,\overline{Q}_0)\\
							p_{1|w} &\text{ otherwise.}
						\end{cases}
\end{align}

By appropriately specifying the constant $B$ in these propensity scores, we can show that they correspond to the propensity scores $\P(X=1 \mid Y_1,Y_0,W=w)$ under joint $c$-dependence for all four pairs of cdf bounds. Before showing this, we state and prove three auxiliary lemmas.

\begin{lemma}\label{lemma:propscore_equal_condexp_v2}
	Let $w\in\supp(W)$. Suppose $m(\cdot)$ is a Borel measurable function and $\P(m(Y_1,Y_0) > \delta|W=w) = 1$ for some $\delta > 0$. The following statements are equivalent:
	\begin{enumerate}
	\item Conditional on $W=w$, the following statement holds almost surely:
	\begin{equation}
	\label{eq:prop_equality_1}
		m(Y_1,Y_0) = \E[X \mid Y_1,Y_0,W=w].
	\end{equation}
        \item For all $(y_1,y_0) \in \R^2$, the following equality holds:
            \begin{equation}
            \label{eq:prop_equality_3}
                \E\left[\1(Y_1\leq y_1, Y_0 \leq y_0) m(Y_1, Y_0) \mid W=w\right] = \P(Y_1\leq y_1, Y_0\leq y_0, X=1 \mid W=w).
            \end{equation}
        \end{enumerate}
\end{lemma}

\begin{proof}[Proof of Lemma \ref{lemma:propscore_equal_condexp_v2}]
We first show  \eqref{eq:prop_equality_1} implies \eqref{eq:prop_equality_3}, note that 
\begin{equation}
\label{eq:prop_equality_sufficiency}
\begin{aligned}
    \P\left(Y_1 \leq y_1, Y_0 \leq y_0, X=1 \mid W = w\right)
    &= \E\left[\1[Y_1\leq y_1, Y_0\leq y_0]X \mid W = w\right] \\
    &= \E\left(\E\left[\1[Y_1\leq y_1, Y_0\leq y_0]X| Y_1, Y_0, W = w\right]|W = w\right) \\
    &= \E\left(\1[Y_1\leq y_1, Y_0\leq y_0]\Exp[X| Y_1, Y_0, W=w]|W = w\right) \\
    &= \E\left(\1[Y_1 \leq y_1, Y_0 \leq y_0] m(Y_1, Y_0) \mid W = w\right).
\end{aligned}
\end{equation}
where we use the law of iterated expectation in the second line and use \eqref{eq:prop_equality_1} in the last line of derivation.

Next, we show that \eqref{eq:prop_equality_3} implies \eqref{eq:prop_equality_1}. To establish this result, we first note a few key facts:
\begin{enumerate}
        \item Following from the last two lines of \eqref{eq:prop_equality_sufficiency}, the law of iterated expectations implies       
        \[
            \E[\1[Y_1\leq y_1, Y_0\leq y_0]\E(X \mid Y_1, Y_0, W=w) \mid W=w] = \E[\1[Y_1\leq y_1, Y_0\leq y_0] m(Y_1,Y_0) \mid W=w]
        \]
        for each $(y_1, y_0) \in \R^2$.
        \item For $(y_1,y_0) \in \R^2$, define the preimage from a half-space on $\R^2$:
        \[
        	I_{y_1, y_0} = \{\omega\in \Omega: Y_1(\omega)\leq y_1, Y_0(\omega)\leq y_0\}
        \]
        and let $\mathcal{A}_2 \coloneqq \left\{I_{y_1, y_0}: (y_1, y_0) \in \R^2\right\}$. Similar to the proof of lemma \ref{lemma:propscore_equal_condexp}, the class of sets $\mathcal{A}_2$ is a $\pi$-system.
        \item The sample space can be written as a countable union of sets in $\mathcal{A}_2$:
        \[
        	\Omega = \{\omega\in\Omega: Y_1(\omega)<\infty, Y_0(\omega)<\infty\} = \bigcup_{n=1}^\infty I_{n,n}.
        \]
        \item The random variable $m(Y_1, Y_0)$ is measurable with respect to the $\sigma$-algebra generated by $(Y_1, Y_0)$ due to the Borel measurability of $m(\cdot)$, and it is integrable since $\E(m(Y_1,Y_0) \mid W=w) = \Prob(X=1 \mid W=w) < \infty$ by sending $y_1$ and $y_0$ to infinity in \eqref{eq:prop_equality_sufficiency}.
        
       	\item The $\sigma$-algebra generated by $\mathcal{A}_2$ equals the $\sigma$-algebra generated by $(Y_1, Y_0)$, i.e.,
	\[
		\sigma(\mathcal{A}_2) = \sigma(Y_1, Y_0).
	\] 
	To show this, define the mapping $f:\Omega \to \R^2$ as $f(\omega) \mapsto (Y_1(\omega), Y_0(\omega))$ and $\mathcal{F} = \{(-\infty, y_1]\times (-\infty, y_0]: (y_1, y_0)\in\R^2\}$. Note that 
	\[
	\sigma(\mathcal{A}_2) = \sigma(f^{-1}(\mathcal{F})) = f^{-1}(\sigma(\mathcal{\mathcal{F}})).
	\]
	Since the Borel $\sigma$-algebra on $\R^2$ can be generated by elements in $\mathcal{F}$, we have $\sigma(\mathcal{F}) = \mathcal{B}(\R^2)$. This implies
	\[
		f^{-1}(\sigma(\mathcal{F})) = f^{-1}(\mathcal{B}(\R^2)) \coloneqq \sigma(Y_1, Y_0).
	\]
	Therefore the desired conclusion holds.
\end{enumerate}

Given the above results, it follows by \citet[Theorem 34.1]{Billingsley1995} 
that
\[
	\Exp[X \mid Y_1,Y_0,W=w] = m(Y_1, Y_0),
\]
almost surely conditional on $W=w$, as desired.
\end{proof}

\begin{lemma}\label{lemma:cond_vs_uncond_comonotonicity}
Let $w\in\supp(W)$. Consider a probability distribution defined by
\begin{align*}
	\widetilde{\P}(Y_1 \leq y_1, Y_0 \leq y_0, X=x \mid W=w) 
	&= \min\{F_{Y \mid X,W}(y_1 \mid 1,w),\underline{F}_{Y_0 \mid X,W}(y_0 \mid 1,w)\}p_{1|w} x \\
	&\quad + \min\{\overline{F}_{Y_1 \mid X,W}(y_1 \mid 0,w),F_{Y \mid X,W}(y_0|0,w)\}p_{0|w} (1-x),
\end{align*}
then 
\begin{align*}
	\widetilde{\P}(Y_1 \leq y_1, Y_0 \leq y_0 \mid W=w) &= \min\{\overline{F}_{Y_1 \mid W}(y_1 \mid w),\underline{F}_{Y_0 \mid W}(y_0 \mid w)\}.
\end{align*}
Also for the following distribution,
\begin{align*}
	\widetilde{\P}(Y_1 \leq y_1, Y_0 \leq y_0, X=x \mid W=w) 
	&= \min\{F_{Y \mid X,W}(y_1 \mid 1,w),\overline{F}_{Y_0 \mid X,W}(y_0 \mid 1,w)\}p_{1|w} x \\
	&\quad + \min\{\underline{F}_{Y_1 \mid X,W}(y_1 \mid 0,w),F_{Y \mid X,W}(y_0|0,w)\}p_{0|w} (1-x)
\end{align*}
implies
\begin{align*}
	\widetilde{\P}(Y_1 \leq y_1, Y_0 \leq y_0 \mid W=w) &= \min\{\underline{F}_{Y_1}(y_1 \mid w),\overline{F}_{Y_0}(y_0 \mid w)\}.
\end{align*}
\end{lemma}

\begin{proof}[Proof of Lemma \ref{lemma:cond_vs_uncond_comonotonicity}]
	Consider the first statement with $p_{1|w} = \lc$. Then it follows that $\overline{F}_{Y_1 \mid X,W}(y_1 \mid 0,w) = \overline{F}_{Y_1 \mid W}(y_1 \mid w) = F_{Y \mid X,W}(y_1 \mid 1,w)$ and $\underline{F}_{Y_0 \mid X,W}(y_0 \mid 1,w) = \underline{F}_{Y_0 \mid W}(y_0 \mid w) = F_{Y \mid X,W}(y_0|0,w)$. Therefore,
    \[
        \widetilde{\P}(Y_1\leq y_1, Y_0\leq y_0, X=x \mid W=w) = \min\left\{F_{Y \mid X,W}(y_1 \mid 1,w), F_{Y \mid X,W}(y_0|0,w)\right\} p_{x|w}.
    \]
    This implies  
    \begin{align*}
        \widetilde{\P}(Y_1\leq y_1, Y_0\leq y_0 \mid W=w) 
        &= \min\left\{F_{Y \mid X,W}(y_1 \mid 1,w), F_{Y \mid X,W}(y_0|0,w)\right\} \\
        &= \min\left\{\overline{F}_{Y_1 \mid W}(y_1 \mid w), \underline{F}_{Y_0 \mid W}(y_0 \mid w)\right\}
    \end{align*}
    as desired. The proof for the case where $p_{1|w} = \uc$ follows the same arguments and thus omitted.
    
    Next consider $\lc < p_{1|w} < \uc$. We have that
	\begin{align*}
		&\widetilde{\P}(Y_1 \leq y_1, Y_0 \leq y_0 \mid W=w) \\
		&=  \min\{F_{Y \mid X,W}(y_1 \mid 1,w),\underline{F}_{Y_0 \mid X,W}(y_0 \mid 1,w)\}p_{1|w} + \min\{\overline{F}_{Y_1 \mid X,W}(y_1 \mid 0,w),F_{Y \mid X,W}(y_0|0,w)\}p_{0|w} \\
		&= F_{Y \mid X,W}(y_1 \mid 1,w)p_{1|w} \1(\underline{F}_{Y_0 \mid X,W}(y_0 \mid 1,w) \geq F_{Y \mid X,W}(y_1 \mid 1,w)) \\
		&\quad + \underline{F}_{Y_0 \mid X,W}(y_0 \mid 1,w)p_{1|w} \1(\underline{F}_{Y_0 \mid X,W}(y_0 \mid 1,w) < F_{Y \mid X,W}(y_1 \mid 1,w))\\
		&\quad + \overline{F}_{Y_1 \mid X,W}(y_1 \mid 0,w)p_{0|w} \1(F_{Y \mid X,W}(y_0|0,w) \geq \overline{F}_{Y_1 \mid X,W}(y_1 \mid 0,w)) \\
		&\quad + F_{Y \mid X,W}(y_0|0,w)p_{0|w} \1(F_{Y \mid X,W}(y_0|0,w) < \overline{F}_{Y_1 \mid X,W}(y_1 \mid 0,w)) \\
        &= \left(F_{Y \mid X,W}(y_1 \mid 1,w)p_{1|w} + \overline{F}_{Y_1 \mid X,W}(y_1 \mid 0,w)p_{0|w}\right)\1[\underline{F}_{Y_0 \mid W}(y_0 \mid w) \geq \overline{F}_{Y_1 \mid W}(y_1 \mid w)] \\
        & \quad + \left(\underline{F}_{Y_0 \mid X,W}(y_0 \mid 1,w)p_{1|w} + F_{Y \mid X,W}(y_0|0,w)p_{0|w}\right)\1[\underline{F}_{Y_0 \mid W}(y_0 \mid w) < \overline{F}_{Y_1 \mid W}(y_1 \mid w)] \\
        &= \min\{\overline{F}_{Y_1 \mid W}(y_1 \mid w), \underline{F}_{Y_0 \mid W}(y_0 \mid w) \}.
	\end{align*}
	The third equality follows by the first set of equivalences in Lemma \ref{lemma:preliminary_margcdep}.6 after setting $x = 1$. The last equality follows by Lemma \ref{lemma:cdf_properties}.2. The second statement follows similar arguments but instead uses Lemma \ref{lemma:preliminary_margcdep}.6 by setting $x = 0$. Therefore, the proof is complete.
\end{proof}

\begin{lemma}\label{lemma:cond_vs_uncond_counter-monotonicity}
Let $w\in \supp(W)$. Consider a probability distribution defined by
\begin{align*}
	\widetilde{\P}(Y_1 \leq y_1, Y_0 \leq y_0, X=x \mid W=w) 
	&= \max\{F_{Y \mid X,W}(y_1 \mid 1,w) + \overline{F}_{Y_0 \mid X,W}(y_0 \mid 1,w) - 1, 0\}p_{1|w} x \\
	&\quad + \max\{\overline{F}_{Y_1 \mid X,W}(y_1 \mid 0,w) + F_{Y \mid X,W}(y_0|0,w) - 1, 0\}p_{0|w} (1-x)
\end{align*}
then 
\begin{align*}
	\widetilde{\P}(Y_1 \leq y_1, Y_0 \leq y_0 \mid W=w) &= \max\{\overline{F}_{Y_1 \mid W}(y_1 \mid w) + \overline{F}_{Y_0 \mid W}(y_0 \mid w) - 1, 0\}.
\end{align*}
Also for the following distribution,
\begin{align*}
	\widetilde{\P}(Y_1 \leq y_1, Y_0 \leq y_0, X=x \mid W=w) 
	&= \max\{F_{Y \mid X,W}(y_1 \mid 1,w) + \underline{F}_{Y_0 \mid X,W}(y_0 \mid 1,w) - 1, 0\}p_{1|w}x \\
	&\quad + \max\{\underline{F}_{Y_1 \mid X,W}(y_1 \mid 0,w) + F_{Y \mid X,W}(y_0|0,w) - 1, 0\}p_{0|w}(1-x)
\end{align*}
implies 
\begin{align*}
	\widetilde{\P}(Y_1 \leq y_1, Y_0 \leq y_0 \mid W=w) &= \max\{\underline{F}_{Y_1 \mid W}(y_1 \mid w) + \underline{F}_{Y_0 \mid W}(y_0 \mid w) - 1, 0\}.
\end{align*}
\end{lemma}

\begin{proof}[Proof of Lemma \ref{lemma:cond_vs_uncond_counter-monotonicity}]

Consider the first statement. Similar arguments from the proof of Lemma \ref{lemma:cond_vs_uncond_comonotonicity} can be used to establish the desired result for $p_{1|w} = \lc$ or $p_{1|w} = \uc$. Thus we consider the case where $\lc < p_{1|w} < \uc$. Then we have that 
\begin{align*}
    &\widetilde{\P}(Y_1\leq y_1, Y_0\leq y_0 \mid W=w) \\
    &= \max\{F_{Y \mid X,W}(y_1 \mid 1,w) + \overline{F}_{Y_0 \mid X,W}(y_0 \mid 1,w) - 1, 0\}p_{1|w}  \\
    &\quad + \max\{\overline{F}_{Y_1 \mid X,W}(y_1 \mid 0,w) + F_{Y \mid X,W}(y_0|0,w) - 1, 0\}p_{0|w} \\
    &= \max\{p_{1|w}(F_{Y \mid X,W}(y_1 \mid 1,w) + \overline{F}_{Y_0 \mid X,W}(y_0 \mid 1,w) - 1),0\} \\
    &\quad + \max\left\{p_{0|w}(\overline{F}_{Y_1 \mid X,W}(y_1 \mid 0,w) + F_{Y \mid X,W}(y_0|0,w) - 1), 0\right\} \\
    &= \max\left\{\sum_{x=0,1}p_{x|w}F_{Y \mid X,W}(y_x|x,w) + p_{0|w}\overline{F}_{Y_1 \mid X,W}(y_1 \mid 0,w) + p_{1|w}\overline{F}_{Y_0 \mid X,W}(y_0 \mid 1,w) - (p_{1|w} + p_{0|w}), 0\right\} \\
    &= \max\{\overline{F}_{Y_1 \mid W}(y_1 \mid w), \overline{F}_{Y_0 \mid W}(y_0 \mid w) - 1, 0\}.
\end{align*}
The second equality follows by the first set of equivalences in Lemma \ref{lemma:preliminary_margcdep}.10 by setting $x=1$. The last equality holds by Lemma \ref{lemma:cdf_properties}.2. The second statement follows similar arguments but instead uses Lemma  \ref{lemma:preliminary_margcdep}.10 on the second set of equivalences regarding lower bounds of cdfs. Therefore, the proof is complete.
\end{proof}

\begin{lemma}\label{lemma:attainability_jointcdep}
Let assumptions \ref{assn:overlap} and \ref{assn:joint_cdep} hold. Let $\overline{C}_{1,0 \mid X,W}$ and $\underline{C}_{1,0 \mid X,W}$ denote classes of comonotonic (and counter-monotonic) copulas where $\overline{C}_{1,0|x,w}(u,v) = \min\{u,v\}$ and $\underline{C}_{1,0|x,w}(u,v) = \max\{u+v-1,0\}$ for all $(x,w) \in \{0,1\} \times \supp(W)$. Then each of the following terms is contained in the identified set $\mathcal{I}_0^j(F_{Y,X,W})$:
\begin{enumerate}
\item $(\overline{F}_{Y_1 \mid W},\underline{F}_{Y_0 \mid W}, \overline{C}_{1,0 \mid X,W})$;
\item $(\overline{F}_{Y_1 \mid W},\overline{F}_{Y_0 \mid W}, \underline{C}_{1,0 \mid X,W})$;
\item $(\underline{F}_{Y_1 \mid W},\overline{F}_{Y_0 \mid W}, \overline{C}_{1,0 \mid X,W})$;
\item $(\underline{F}_{Y_1 \mid W},\underline{F}_{Y_0 \mid W}, \underline{C}_{1,0 \mid X,W})$.
\end{enumerate}
\end{lemma}

\begin{proof}[Proof of Lemma \ref{lemma:attainability_jointcdep}]

\textbf{Proof of Part 1}: Fix a $w\in\supp(W)$. We prove the first statement by constructing a probability distribution $\widetilde{\P}$ for $(Y_1,Y_0,X)$ conditional on $W=w$ such that for all $y \in \R$ and $x \in \{0,1\}$, we have
\begin{enumerate}
	\item $\widetilde{\P}(Y_1 \leq y \mid W=w) = \overline{F}_{Y_1 \mid W}(y \mid w)$ and  $\widetilde{\P}(Y_0 \leq y \mid W=w) =  \underline{F}_{Y_0 \mid W}(y \mid w)$;
	\item $\widetilde{\P}(X=x \mid W=w) = p_{x|w}$;
	\item $\widetilde{\P}(Y_x \leq y \mid X=x,W=w) = F_{Y \mid X,W}(Y \mid X,w)$;
	\item $\widetilde{\P}(Y_1 \leq y_1, Y_0\leq y_0 \mid X=x, W=w) = \min\left\{\widetilde{\P}(Y_1\leq y_1 \mid X=x, W=w), \widetilde{\P}(Y_0\leq y_0 \mid X=x, W=w)\right\}$;
	\item $\widetilde{\P}(X=1 \mid Y_1,Y_0,W=w) \in [\lc,\uc]$ for $\widetilde{\P}$-almost surely.
\end{enumerate}

Similar to the arguments in the proof of Theorem \ref{thm:cdf_sharp_margcdep}, Conditions 1--5 ensures that the constructed distribution $\widetilde{\P}$ generates the desired marginal cdfs and copulas in Lemma \ref{lemma:attainability_jointcdep}.1 and staisfies all the requirements in the definition of identified set $\mathcal{I}_0^j(F_{Y,X,W})$. 

Let 
\begin{equation}
\label{eq:UL_sharpness_jointcdep}
\begin{aligned} 
	\widetilde{\P}(Y_1 \leq y_1, Y_0 \leq y_0, X = x \mid W=w) 
	&= x \min\{F_{Y \mid X,W}(y_1 \mid 1,w),\underline{F}_{Y_0 \mid X,W}(y_0 \mid 1,w)\} p_{1|w} \\
	&\quad + (1-x) \min\{\overline{F}_{Y_1 \mid X,W}(y_1 \mid 0,w),F_{Y \mid X,W}(y_0|0,w)\} p_{0|w}.
\end{aligned}
\end{equation}
By Lemma \ref{lemma:cdf_properties}.1, $\underline{F}_{Y_0 \mid X,W}(y_0 \mid 1,w)$ and $\overline{F}_{Y_1 \mid X,W}(y_1 \mid 0,w)$ are cdfs. Also note that $(u,v) \mapsto \min\{u,v\}$ is the comonotonic copula. By Sklar's Theorem, $\widetilde{\P}$ is a joint distribution function for $(Y_1,Y_0,X)$ conditional on $W=w$. 

Following the same steps as in the proof of Theorem \ref{thm:cdf_sharp_margcdep}, it can be shown show that conditions 1--4 are satisfied because the distribution in \eqref{eq:UL_sharpness_jointcdep} is the same as in \eqref{eq:sharpness_margcdep} but for a specific rather than an arbitrary choice of copulas. By Lemma \ref{lemma:cond_vs_uncond_comonotonicity}, $\widetilde{\P}$ implies the following co-monotonic joint distribution of $(Y_1, Y_0)$:
\begin{equation}\label{eq:joint_UL_cdf}
    \widetilde{\P}(Y_1 \leq y_1, Y_0 \leq y_0 \mid W=w) 
    =
    \min\left\{\overline{F}_{Y_1 \mid W}(y_1 \mid w), \underline{F}_{Y_0 \mid W}(y_0 \mid w)\right\}.
\end{equation}

To show condition 5 holds, and thus complete the proof, we construct a function $p^{ul}$ such that $p^{ul}(Y_1,Y_0) = \widetilde{\E}[X \mid Y_1,Y_0,W=w]$ and $p^{ul}(Y_1,Y_0,w) \in [\lc,\uc]$ almost surely under $\widetilde{\P}$. 

First consider $p_{1|w} = \lc$, then we have 
\[
	\underline{F}_{Y_0 \mid W}(y_0 \mid w) = \underline{F}_{Y_0 \mid X,W}(y_0 \mid 1,w) = F_{Y \mid X,W}(y_0|0,w)
\]
and 
\[
	\overline{F}_{Y_1 \mid W}(y_1 \mid w) = \overline{F}_{Y_1 \mid X,W}(y_1 \mid 0,w) = F_{Y \mid X,W}(y_1 \mid 1,w).
\] 
This implies the following derivation
\begin{align*}
    \widetilde{\E}\left[\1(Y_1\leq y_1, Y_0\leq y_0) \lc|W=w\right] 
    &= p_{1|w} \widetilde{\P}(Y_1\leq y_1, Y_0\leq y_0 \mid W=w) \\
    &= p_{1|w}\min\left\{F_{Y \mid X,W}(y_1 \mid 1,w), {F}_{Y \mid X,W}(y_0|0,w)\right\} \\
    &= \widetilde{\P}(Y_1 \leq y_1, Y_0\leq y_0, X=1 \mid W=w).
\end{align*}
The first line holds by $\lc = p_{1|w}$, the second line holds by \eqref{eq:joint_UL_cdf}, and the last line holds by \eqref{eq:UL_sharpness_jointcdep}. Following Lemma \ref{lemma:propscore_equal_condexp_v2}, we conclude that $\widetilde{\P}(X=1 \mid Y_1, Y_0, W=w) = \lc$ almost surely, which is naturally bounded between $\lc$ and $\uc$. The proof of the case where $p_{1|w} = \uc$ is similar and thus omitted. 

Next consider $\lc < p_{1|w} < \uc$. Let $p^{ul}(Y_1,Y_0,w) = p^{ul}(Y_1,Y_0,w;B^{ul})$ defined in \eqref{eq:UL_propscore}, where
\begin{align*}
	B^{ul} 
	&= \frac{\widetilde{\P}(Y_1 = \overline{Q}_1, Y_0 = \underline{Q}_0,X=1 \mid W=w)}{\widetilde{\P}(Y_1 = \overline{Q}_1, Y_0 = \underline{Q}_0 \mid W=w)}
\end{align*}
whenever $\widetilde{\P}(Y_1 = \overline{Q}_1, Y_0 = \underline{Q}_0 \mid W=w) > 0$. Let $B^{ul} = p_{1|w}$ otherwise.

We verify that $B^{ul} \in [\lc,\uc]$ if the denominator is nonzero. 

First, note that the denominator of $B^{ul}$ can be expanded below
\begin{align*}
	&\widetilde{\P}(Y_1 = \overline{Q}_1, Y_0 = \underline{Q}_0 \mid W=w) \\
	&= \widetilde{\P}(Y_1 \leq \overline{Q}_1, Y_0 \leq \underline{Q}_0 \mid W=w) - \widetilde{\P}(Y_1 \leq \overline{Q}_1, Y_0 < \underline{Q}_0 \mid W=w) \\
	&\quad - \widetilde{\P}(Y_1 < \overline{Q}_1, Y_0 \leq \underline{Q}_0 \mid W=w) + \widetilde{\P}(Y_1 < \overline{Q}_1, Y_0 < \underline{Q}_0 \mid W=w)\\
	&= \min\{\overline{F}_{Y_1 \mid W}(\overline{Q}_1 \mid w),\underline{F}_{Y_0 \mid W}(\underline{Q}_0|w)\} - \min\{\overline{F}_{Y_1 \mid W}(\overline{Q}_1 \mid w),\underline{F}_{Y_0 \mid W}(\underline{Q}_0-|w)\} \\
	&\quad - \min\{\overline{F}_{Y_1 \mid W}(\overline{Q}_1- \mid w),\underline{F}_{Y_0 \mid W}(\underline{Q}_0|w)\} + \min\{\overline{F}_{Y_1 \mid W}(\overline{Q}_1- \mid w),\underline{F}_{Y_0 \mid W}(\underline{Q}_0-|w)\},
\end{align*}
where the second equality holds via \eqref{eq:joint_UL_cdf}. By Lemma \ref{lemma:preliminary_margcdep}.8, this expression simplifies to
\begin{align*}
	&\widetilde{\P}(Y_1 = \overline{Q}_1, Y_0 = \underline{Q}_0 \mid W=w) \\
	&= \min\{\overline{F}_{Y_1 \mid W}(\overline{Q}_1 \mid w),\underline{F}_{Y_0 \mid W}(\underline{Q}_0|w)\} - \underline{F}_{Y_0 \mid W}(\underline{Q}_0-|w) \\
	&\quad  - \overline{F}_{Y_1 \mid W}(\overline{Q}_1- \mid w) + \min\{\overline{F}_{Y_1 \mid W}(\overline{Q}_1- \mid w),\underline{F}_{Y_0 \mid W}(\underline{Q}_0-|w)\}\\
	&= \min\{\overline{F}_{Y_1 \mid W}(\overline{Q}_1 \mid w),\underline{F}_{Y_0 \mid W}(\underline{Q}_0|w)\} - \max\{\overline{F}_{Y_1 \mid W}(\overline{Q}_1- \mid w),\underline{F}_{Y_0 \mid W}(\underline{Q}_0-|w)\}.
\end{align*}

Second, we expand the numerator of $B^{ul}$. We have that 
\[
	\widetilde{\P}(Y_1 = \overline{Q}_1, Y_0 = \underline{Q}_0, X=1 \mid W=w) = \widetilde{\P}(Y_1 = \overline{Q}_1, Y_0 = \underline{Q}_0 \mid X=1,W=w) p_{1|w}
\]
and that
\begin{align*}
	&\widetilde{\P}(Y_1 = \overline{Q}_1, Y_0 = \underline{Q}_0 \mid X=1,W=w)\\
	&= \min\{{F}_{Y \mid X,W}(\overline{Q}_1 \mid 1,w),\underline{F}_{Y_0 \mid X,W}(\underline{Q}_0 \mid 1,w)\} - \min\{{F}_{Y \mid X,W}(\overline{Q}_1 \mid 1,w),\underline{F}_{Y_0 \mid X,W}(\underline{Q}_0- \mid 1,w)\}\\
	& \quad - \min\{{F}_{Y \mid X,W}(\overline{Q}_1- \mid 1,w),\underline{F}_{Y_0 \mid X,W}(\underline{Q}_0 \mid 1,w)\} + \min\{{F}_{Y \mid X,W}(\overline{Q}_1- \mid 1,w),\underline{F}_{Y_0 \mid X,W}(\underline{Q}_0- \mid 1,w)\}\\
	&= \min\{{F}_{Y \mid X,W}(\overline{Q}_1 \mid 1,w),\underline{F}_{Y_0 \mid X,W}(\underline{Q}_0 \mid 1,w)\} - \underline{F}_{Y_0 \mid X,W}(\underline{Q}_0- \mid 1,w) \\
	&\quad - {F}_{Y \mid X,W}(\overline{Q}_1- \mid 1,w) + \min\{{F}_{Y \mid X,W}(\overline{Q}_1- \mid 1,w),\underline{F}_{Y_0 \mid X,W}(\underline{Q}_0- \mid 1,w)\}\\
	&= \min\{{F}_{Y \mid X,W}(\overline{Q}_1 \mid 1,w),\underline{F}_{Y_0 \mid X,W}(\underline{Q}_0 \mid 1,w)\} - \max\{{F}_{Y \mid X,W}(\overline{Q}_1- \mid 1,w),\underline{F}_{Y_0 \mid X,W}(\underline{Q}_0- \mid 1,w)\},
\end{align*}
where the second to last equality follows from Lemma \ref{lemma:preliminary_margcdep}.9.

From Part 6 and 7 of Lemma \ref{lemma:preliminary_margcdep}, we observe that $B^{ul}$ can take four possible values as follows:
\begin{align}
    B^{ul} 
        & =\frac{\widetilde{\P}(Y_1 = \overline{Q}_1, Y_0 = \underline{Q}_0 \mid X=1,W=w)p_{1|w}}{\widetilde{\P}(Y_1 = \overline{Q}_1, Y_0 = \underline{Q}_0 \mid W=w)} \notag \\
        &= \frac{(F_{Y \mid X,W}(\overline{Q}_1 \mid 1,w) - {F}_{Y \mid X,W}(\overline{Q}_1- \mid 1,w))p_{1|w}}{\overline{F}_{Y_1 \mid W}(\overline{Q}_1 \mid w) - \overline{F}_{Y_1 \mid W}(\overline{Q}_1- \mid w)}
        \1\left(
        \begin{array}{c}
        \overline{F}_{Y_1 \mid W}(\overline{Q}_1 \mid w) \leq \underline{F}_{Y_0 \mid W} (\underline{Q}_0|w), \\
        \overline{F}_{Y_1 \mid W}(\overline{Q}_1- \mid w) > \underline{F}_{Y_0 \mid W}(\underline{Q}_0-|w)
        \end{array}
        \right)\label{eq:Bul_term1}\\
	&+\frac{({F}_{Y \mid X,W}(\overline{Q}_1 \mid 1,w) - \underline{F}_{Y_0 \mid X,W}(\underline{Q}_0- \mid 1,w))p_{1|w}}{\overline{F}_{Y_1 \mid W}(\overline{Q}_1 \mid w) - \underline{F}_{Y_0 \mid W}(\underline{Q}_0-|w)}
		\1\left(
		\begin{array}{c}
		\overline{F}_{Y_1 \mid W}(\overline{Q}_1 \mid w) \leq \underline{F}_{Y_0 \mid W}(\underline{Q}_0|w), \\
		\overline{F}_{Y_1 \mid W}(\overline{Q}_1- \mid w) \leq \underline{F}_{Y_0 \mid W}(\underline{Q}_0-|w)
		\end{array}\right)\label{eq:Bul_term2}\\
	&+ \frac{(\underline{F}_{Y_0 \mid X,W}(\underline{Q}_0 \mid 1,w) - {F}_{Y \mid X,W}(\overline{Q}_1- \mid 1,w))p_{1|w}}{\underline{F}_{Y_0 \mid W}(\underline{Q}_0|w) - \overline{F}_{Y_1 \mid W}(\overline{Q}_1- \mid w)}	
		\1\left(
		\begin{array}{c}
		\overline{F}_{Y_1 \mid W}(\overline{Q}_1 \mid w) > \underline{F}_{Y_0 \mid W}(\underline{Q}_0|w), \\
		\overline{F}_{Y_1 \mid W}(\overline{Q}_1- \mid w) > \underline{F}_{Y_0 \mid W}(\underline{Q}_0-|w)
		\end{array}\right)\label{eq:Bul_term3}\\
	&+\frac{(\underline{F}_{Y_0 \mid X,W}(\underline{Q}_0 \mid 1,w) - \underline{F}_{Y_0 \mid X,W}(\underline{Q}_0- \mid 1,w))p_{1|w}}{\underline{F}_{Y_0 \mid W}(\underline{Q}_0|w) - \underline{F}_{Y_0 \mid W}(\underline{Q}_0-|w)}	
		\1\left(
		\begin{array}{c}
		\overline{F}_{Y_1 \mid W}(\overline{Q}_1 \mid w) > \underline{F}_{Y_0 \mid W}(\underline{Q}_0|w), \\
		\overline{F}_{Y_1 \mid W}(\overline{Q}_1- \mid w) \leq \underline{F}_{Y_0 \mid W}(\underline{Q}_0-|w)
		\end{array}\right).\label{eq:Bul_term4}
\end{align}
All the terms have positive denominators since we focus on the case where $\widetilde{\P}(Y_1 = \overline{Q}_1, Y_0 = \underline{Q}_0 \mid W=w) > 0$. As shown in Lemma \ref{lemma:preliminary_margcdep}.4, terms \eqref{eq:Bul_term1} and \eqref{eq:Bul_term4} lie in $[\lc,\uc]$. 

Next we examine the term \eqref{eq:Bul_term2}, which can be written as follows
\begin{align*}
   	&\frac{({F}_{Y \mid X,W}(\overline{Q}_1 \mid 1,w) - \underline{F}_{Y_0 \mid X,W}(\underline{Q}_0- \mid 1,w))p_{1|w}}{\overline{F}_{Y_1 \mid W}(\overline{Q}_1 \mid w) - \underline{F}_{Y_0 \mid W}(\underline{Q}_0-|w)} \\
    &= \frac{F_{Y \mid X,W}(\overline{Q}_1 \mid 1,w)p_{1|w} - F_{Y \mid X,W}(\underline{Q}_0-|0,w)\frac{p_{0|w}\lc}{1-\lc}}{\frac{\uc-p_{1|w}}{\uc} + F_{Y \mid X,W}(\overline{Q}_1 \mid 1,w)\frac{p_{1|w}}{\uc} - F_{Y \mid X,W}(\underline{Q}_0-|0,w)\frac{p_{0|w}}{1-\lc}} \\
    &= \lc + \frac{(\uc-\lc)p_{1|w}}{\uc}\frac{F_{Y \mid X,W}(\overline{Q}_1 \mid 1,w) - \overline{\tau}_1}{\frac{\uc-p_{1|w}}{\uc} + F_{Y \mid X,W}(\overline{Q}_1 \mid 1,w)\frac{p_{1|w}}{\uc} - F_{Y \mid X,W}(\underline{Q}_0-|0,w)\frac{p_{0|w}}{1-\lc}} \\
    &= \lc + \frac{(\uc-\lc)p_{1|w}}{\uc}\frac{F_{Y \mid X,W}(\overline{Q}_1 \mid 1,w) - \overline{\tau}_1}{\overline{F}_{Y_1 \mid W}(\overline{Q}_1 \mid w) - \underline{F}_{Y_0 \mid W}(\underline{Q}_0-|w)}\\
    &\geq \lc
\end{align*}
where the last line follows by $\uc \geq \lc$ and $F_{Y \mid X,W}(\overline{Q}_1 \mid 1,w) \geq \overline{\tau}_1$ via Lemma \ref{lemma:quantile_props}.2. Also note that 
\begin{align*}
   \frac{({F}_{Y \mid X,W}(\overline{Q}_1 \mid 1,w) - \underline{F}_{Y_0 \mid X,W}(\underline{Q}_0- \mid 1,w))p_{1|w}}{\overline{F}_{Y_1 \mid W}(\overline{Q}_1 \mid w) - \underline{F}_{Y_0 \mid W}(\underline{Q}_0-|w)}
   &= \uc + \frac{p_{0|w}(\uc-\lc)}{1-\lc}\frac{F_{Y \mid X,W}(\underline{Q}_0-|0,w) - \underline{\tau}_0}{\overline{F}_{Y_1 \mid W}(\overline{Q}_1 \mid w) - \underline{F}_{Y_0 \mid W}(\underline{Q}_0-|w)} \leq \uc,
\end{align*}
where the inequality follows by $\uc \geq \lc$ and $F_{Y \mid X,W}(\underline{Q}_0-|0,w) \leq \underline{\tau}_0$ via Lemma \ref{lemma:quantile_props}.3. Thus we have shown the term \eqref{eq:Bul_term2} is bounded within $[\lc,\uc]$.

Then consider the term \eqref{eq:Bul_term3}. Following the same arguments, we have
\begin{align*}
    \frac{(\underline{F}_{Y_0 \mid X,W}(\underline{Q}_0 \mid 1,w) - {F}_{Y \mid X,W}(\overline{Q}_1- \mid 1,w))p_{1|w}}{\underline{F}_{Y_0 \mid W}(\underline{Q}_0|w) - \overline{F}_{Y_1 \mid W}(\overline{Q}_1- \mid w)}
    &= \lc + \frac{p_{0|w}(\uc - \lc)}{1-\uc} \frac{F_{Y \mid X,W}(\underline{Q}_0|0,w) - \underline{\tau}_0}{\underline{F}_{Y_0 \mid W}(\underline{Q}_0|w) - \overline{F}_{Y_1 \mid W}(\overline{Q}_1- \mid w)} \geq \lc, 
\end{align*}
and
\begin{align*}
    &\frac{(\underline{F}_{Y_0 \mid X,W}(\underline{Q}_0 \mid 1,w) - {F}_{Y \mid X,W}(\overline{Q}_1- \mid 1,w))p_{1|w}}{\underline{F}_{Y_0 \mid W}(\underline{Q}_0|w) - \overline{F}_{Y_1 \mid W}(\overline{Q}_1- \mid w)} = \uc + \frac{p_{1|w}(\uc-\lc)}{\lc}\frac{F_{Y \mid X,W}(\overline{Q}_1- \mid 1.w) - \overline{\tau}_1}{\underline{F}_{Y_0 \mid W}(\underline{Q}_0|w) - \overline{F}_{Y_1 \mid W}(\overline{Q}_1- \mid w)} \leq \uc.
\end{align*}
So we have shown that all the four terms \eqref{eq:Bul_term1}--\eqref{eq:Bul_term4} are bounded within $[\lc,\uc]$, thus concluding $B^{ul} \in [\lc,\uc]$, which then establishes $p^{ul}(Y_1, Y_0, w) \in [\lc, \uc]$ almost surely.

To finish this proof, we demonstrate that $\widetilde{\E}[X \mid Y_1,Y_0,W=w] = p^{ul}(Y_1,Y_0,w)$ almost surely. To do so, we use Lemma \ref{lemma:propscore_equal_condexp_v2} and show that
\begin{equation}
\label{eq:Pul_verification}
\begin{aligned}
	\widetilde{\E}\left[\1(Y_1 \leq y_1, Y_0 \leq y_0)p^{ul}(Y_1,Y_0,w) \mid W=w\right] 
	&= \widetilde{\P}(Y_1 \leq y_1, Y_0 \leq y_0, X=1 \mid W=w)  \\
	&= p_{1|w}\min\{{F}_{Y \mid X,W}(y_1 \mid 1,w), \underline{F}_{Y_0 \mid X,W}(y_0 \mid 1,w)\}
\end{aligned}
\end{equation}
for all $(y_1,y_0)\in\R^2$. To complete the proof, we break this up into different cases.

\medskip

\noindent \textbf{(Part 1) Case 1:} $y_1 < \overline{Q}_1$ and $y_0 < \underline{Q}_0$. 

\medskip

In this case, $p^{ul}(Y_1, Y_0, w) \1[Y_1 \leq y_1, Y_0 \leq y_0] = \lc \1[Y_1 \leq y_1, Y_0 \leq y_0]$. Thus we have
\begin{align*}
    \widetilde{\E}\left[\1(Y_1\leq y_1, Y_0\leq y_0) p^{ul}(Y_1, Y_0, w) \mid W=w\right]
    &= \lc\min\left\{\overline{F}_{Y_1 \mid W}(y_1 \mid w), \underline{F}_{Y_0 \mid W}(y_0 \mid w)\right\} \\
    &= \lc\min\left\{F_{Y \mid X,W}(y_1 \mid 1,w) \frac{p_{1|w}}{\lc}, F_{Y \mid X,W}(y_0|0,w)\frac{p_{0|w}}{1-\lc}\right\} \\
    &= \min\left\{F_{Y \mid X,W}(y_1 \mid 1,w)p_{1|w}, F_{Y \mid X,W}(y_0|0,w)\frac{p_{0|w}\lc}{1-\lc}\right\} \\
    &= p_{1|w}\min\left\{F_{Y \mid X,W}(y_1 \mid 1,w), \underline{F}_{Y_0 \mid X,W}(y_0 \mid 1,w)\right\}.
\end{align*}
The second line holds by the assumption that $y_1 < \overline{Q}_1$ and $y_0 < \underline{Q}_0$. Therefore, we have shown that \eqref{eq:Pul_verification} holds.

\medskip

\noindent \textbf{(Part 1) Case 2:} $y_1 \geq \overline{Q}_1$ and $y_0 < \underline{Q}_0$. 

\medskip

First, note that the joint cdf from \eqref{eq:joint_UL_cdf} implies
\begin{equation}
\label{eq:Pul_supp}
    \widetilde{\P}(Y_1 > \overline{Q}_1, Y_0 < \underline{Q}_0 \mid W=w) = \widetilde{\P}(Y_1 < \overline{Q}_1, Y_0 > \underline{Q}_0 \mid W=w) = 0.
\end{equation}
These equalities follow by
\begin{align*}
    \widetilde{\P}(Y_1 > \overline{Q}_1, Y_0 < \underline{Q}_0 \mid W=w) 
    &= \widetilde{\P}(Y_0 < \underline{Q}_0 \mid W=w) - \widetilde{\P}(Y_0 < \underline{Q}_0, Y_1 \leq \overline{Q}_1 \mid W=w) \\
    &= \underline{F}_{Y_0 \mid W}(\underline{Q}_0-|w) - \min\left\{\overline{F}_{Y_1 \mid W}(\overline{Q}_1 \mid w), \underline{F}_{Y_0 \mid W}(\underline{Q}_0-|w)\right\} \\
    &= \underline{F}_{Y_0 \mid W}(\underline{Q}_0-|w) - \underline{F}_{Y_0 \mid W}(\underline{Q}_0-|w) \\
    &= 0,
\end{align*}
where the third line holds by Lemma \ref{lemma:preliminary_margcdep}.8. Similarly, 
\begin{align*}
    \widetilde{\P}(Y_1 < \overline{Q}_1, Y_0 > \underline{Q}_0 \mid W=w)
    &= \widetilde{\P}(Y_1 < \overline{Q}_1 \mid W=w) - \widetilde{\P}(Y_1 < \overline{Q}_1, Y_0 \leq \underline{Q}_0 \mid W=w) \\
    &= \overline{F}_{Y_1 \mid W}(\overline{Q}_1- \mid w) - \min\left\{\overline{F}_{Y_1 \mid W}(\overline{Q}_1- \mid w), \underline{F}_{Y_0 \mid W}(\underline{Q}_0|w)\right\} \\
    &= \overline{F}_{Y_1 \mid W}(\overline{Q}_1- \mid w) - \overline{F}_{Y_1 \mid W}(\overline{Q}_1- \mid w) \\
    &=0.
\end{align*}

Based on \eqref{eq:Pul_supp}, we can decompose the left-hand-side term of  \eqref{eq:Pul_verification} as below
\begin{align*}
    &\widetilde{\E}(\1(Y_1\leq y_1, Y_0\leq y_0)p^{ul}(Y_1,Y_0,w) \mid W=w) \\
    &= \lc\widetilde{\P}(Y_1 \leq \overline{Q}_1, Y_0 \leq y_0 \mid W=w) + p_{1|w}\cdot \widetilde{\P}(\overline{Q}_1 < Y_1 \leq y_1, Y_0 \leq y_0 \mid W=w) \\
    &= \lc \widetilde{\P}(Y_1 \leq \overline{Q}_1, Y_0 \leq y_0 \mid W=w) \\
    &= \lc \min\left\{\overline{F}_{Y_1 \mid W}(\overline{Q}_1 \mid w), \underline{F}_{Y_0 \mid W}(y_0 \mid w)\right\}.
\end{align*}
Note that $\overline{F}_{Y_1 \mid W}(\overline{Q}_1 \mid w) \geq \underline{F}_{Y_0 \mid W}(\underline{Q}_0-|w) \geq  \underline{F}_{Y_0 \mid W}(y_0 \mid w)$ by Lemma \ref{lemma:preliminary_margcdep}.8 and the condition that $y_0 < \underline{Q}_0$. We have 
\[
     \lc \min\left\{\overline{F}_{Y_1 \mid W}(\overline{Q}_1 \mid w), \underline{F}_{Y_0 \mid W}(y_0 \mid w)\right\}
     = \lc \underline{F}_{Y_0 \mid W}(y_0 \mid w) = p_{1|w} \underline{F}_{Y_0 \mid X,W}(y_0 \mid 1,w).
\]
By Lemma \ref{lemma:preliminary_margcdep}.6 (the second set of equivalent results), $\overline{F}_{Y_1 \mid W}(\overline{Q}_1 \mid w) \geq \underline{F}_{Y_0 \mid W}(y_0 \mid w)$ also implies 
\[
	\underline{F}_{Y_0 \mid X,W}(y_0 \mid 1,w) \leq {F}_{Y \mid X,W}(\overline{Q}_1 \mid 1,w) \leq F_{Y \mid X,W}(y_1 \mid 1,w),
\]
hence we have 
\[
    p_{1|w} \underline{F}_{Y_0 \mid X,W}(y_0 \mid 1,w) = p_{1|w} \min\left\{F_{Y \mid X,W}(y_1 \mid 1,w), \underline{F}_{Y_0 \mid X,W}(y_0 \mid 1,w)\right\}.
\]
Combining those results then yields \eqref{eq:Pul_verification}, as desired.

\medskip

\noindent \textbf{(Part 1) Case 3:} $y_1 < \overline{Q}_1$ and $y_0 \geq \underline{Q}_0$. 

\medskip

Following similar arguments in Case 2, we have the following equality
\begin{align*}
    &\widetilde{\E}\left(\1(Y_1\leq y_1, Y_0\leq y_0) p^{ul}(Y_1, Y_0, w) \mid W=w\right) \\
    &= \lc\widetilde{\P}(Y_1\leq y_1, Y_0 \leq \underline{Q}_0 \mid W=w) + p_{1|w} \cdot \widetilde{\P}(Y_1\leq y_1, \underline{Q}_0 < Y_0 \leq y_0 \mid W=w) \\
    &= \lc\min\left\{\overline{F}_{Y_1 \mid W}(y_1 \mid w), \underline{F}_{Y_0 \mid W}(\underline{Q}_0|w)\right\} \\
    &= \lc \overline{F}_{Y_1 \mid W}(y_1 \mid w) \\
    &= p_{1|w} F_{Y \mid X,W}(y_1 \mid 1,w) \\
    &= p_{1|w} \min\left\{F_{Y \mid X,W}(y_1 \mid 1,w), \underline{F}_{Y_0 \mid X,W}(y_0 \mid 1,w)\right\}
\end{align*}
where we use \eqref{eq:Pul_supp} in the second equality, the third equality follows by Lemma \ref{lemma:preliminary_margcdep}.8, and the condition that $y_1 < \overline{Q}_1$, and the last line holds by Lemma \ref{lemma:preliminary_margcdep}.9, where we deduce that 
\[
	F_{Y \mid X,W}(y_1 \mid 1,w) \leq F_{Y \mid X,W}(\overline{Q}_1- \mid 1,w) \leq \underline{F}_{Y_0 \mid X,W}(\underline{Q}_0 \mid 1,w) \leq  \underline{F}_{Y_0 \mid X,W}(y_0 \mid 1,w).
\] 
Therefore, we established \eqref{eq:Pul_verification}.

\medskip

\noindent \textbf{(Part 1) Case 4:} $y_1 = \overline{Q}_1$ and $y_0 = \underline{Q}_0$. 

\medskip

We can decompose the LHS probability of \eqref{eq:Pul_verification} as follows:
\begin{align*}
    &\widetilde{\E}(\1(Y_1\leq y_1, Y_0\leq y_0) p^{ul}(Y_1,Y_0,w) \mid W=w) \\
    &=          \widetilde{\E}(\1(Y_1 = y_1, Y_0 = y_0) p^{ul}(Y_1,Y_0) \mid W=w) +      \widetilde{\E}(\1(Y_1\leq y_1, Y_0 < y_0) p^{ul}(Y_1,Y_0,w) \mid W=w)  \\
    &\quad+      \widetilde{\E}(\1(Y_1 < y_1, Y_0\leq y_0) p^{ul}(Y_1,Y_0,w) \mid W=w) -      \widetilde{\E}(\1(Y_1 < y_1, Y_0 < y_0) p^{ul}(Y_1,Y_0,w) \mid W=w). \\
    &= B^{ul}\widetilde{\P}(Y_1 = \overline{Q}_1, Y_0 = \underline{Q}_0 \mid W=w) \\
    &\quad + \lim_{u\nearrow \underline{Q}_0} \widetilde{\E}(\1(Y_1\leq \overline{Q}_1, Y_0 \leq u \mid w=w) p^{ul}(Y_1,Y_0,w) \mid W=w) \\
    &\quad + \lim_{v\nearrow \overline{Q}_1} \widetilde{\E}(\1(Y_1\leq v, Y_0 \leq \underline{Q}_0) p^{ul}(Y_1,Y_0,w) \mid W=w) \\
    &\quad - \lim_{v\nearrow \overline{Q}_1, u\nearrow \underline{Q}_0} \widetilde{\E}(\1(Y_1\leq v, Y_0\leq u) p^{ul}(Y_1,Y_0,w) \mid W=w) \\
    &= \widetilde{\P}(Y_1 = \overline{Q}_1, Y_0 = \underline{Q}_0, X=1 \mid W=w) +     \lim_{u\nearrow \underline{Q}_0} \widetilde{\P}(Y_1\leq \overline{Q}_1, Y_0 \leq u, X=1 \mid W=w) \\
    &\quad+     \lim_{v\nearrow \overline{Q}_1} \widetilde{\P}(Y_1\leq v, Y_0\leq \underline{Q}_0, X=1 \mid W=w) - \lim_{v\nearrow \overline{Q}_1, u\nearrow \underline{Q}_0}\widetilde{\P}(Y_1\leq v, Y_0\leq u, X=1 \mid W=w) \\
    &= \widetilde{\P}(Y_1 \leq \overline{Q}_1, Y_0\leq \overline{Q}_0, X=1 \mid W=w).
\end{align*}
The second equality holds by the monotone convergence theorem. The third equality holds by the conclusion proved in Case 1--3. The last equality holds by the continuity of probability measure. Thus \eqref{eq:Pul_verification} has been verified.

\medskip

\noindent \textbf{(Part 1) Case 5:} $(y_1, y_0) \geq (\overline{Q}_1, \underline{Q}_0)$. 

\medskip

Given the above results, we have the following derivation:
\begin{align*}
    & \widetilde{\E}\left(\1(Y_1\leq y_1, Y_0\leq y_0) p^{ul}(Y_1, Y_0,w) \mid W=w\right) \\
    &=  \widetilde{\E}\left(1(Y_1\leq \overline{Q}_1, Y_0\leq \underline{Q}_0) p^{ul}(Y_1, Y_0,w) \mid W=w\right) \\
    &\quad  + \widetilde{\E}\left( \uc \left[\1(Y_1\leq y_1, Y_0\leq y_0 \mid W=w) - \1(Y_1\leq \overline{Q}_1, Y_0\leq \underline{Q}_0)\right]|W=w\right) \\
    &= \widetilde{\P}(Y_1\leq \overline{Q}_1, Y_0\leq \underline{Q}_0, X=1 \mid W=w) + \uc\left(\widetilde{\P}(Y_1\leq y_1, Y_0\leq y_0 \mid W=w) - \widetilde{\P}(Y_1\leq \overline{Q}_1, Y_0\leq \underline{Q}_0) \mid W=w\right) \\
    &= \uc \min\left\{\overline{F}_{Y_1 \mid W}(y_1 \mid w), \underline{F}_{Y_0 \mid W}(y_0 \mid w)\right\} + p_{1|w} \min\left\{F_{Y \mid X,W}(\overline{Q}_1 \mid 1,w), \underline{F}_{Y_0 \mid X,W}(\underline{Q}_0 \mid 1,w)\right\} \\
    &\quad  - \uc\min\left\{\overline{F}_{Y_1 \mid W}(\overline{Q}_1 \mid w), \underline{F}_{Y_0 \mid W}(\underline{Q}_0|w)\right\} \\
    &= \uc \min\left\{\overline{F}_{Y_1 \mid W}(y_1 \mid w), \underline{F}_{Y_0 \mid W}(y_0 \mid w)\right\} \\
    &\quad + \left[F_{Y \mid X,W}(\overline{Q}_1 \mid 1,w)p_{1|w} - \overline{F}_{Y_1 \mid W}(\overline{Q}_1 \mid w)\uc\right]\1(\overline{F}_{Y_1 \mid W}(\overline{Q}_1 \mid w) \leq \underline{F}_{Y_0 \mid W}(\underline{Q}_0|w)) \\
    &\quad + \left[\underline{F}_{Y_0 \mid X,W}(\underline{Q}_0 \mid 1,w)p_{1|w} - \underline{F}_{Y_0 \mid W}(\underline{Q}_0|w)\uc\right]\1(\overline{F}_{Y_1 \mid W}(\overline{Q}_1 \mid w) > \underline{F}_{Y_0 \mid W}(\underline{Q}_0|w)) \\
    &=  \uc \min\left\{\overline{F}_{Y_1 \mid W}(y_1 \mid w), \underline{F}_{Y_0 \mid W}(y_0 \mid w)\right\} - (\uc-p_{1|w}) \\
    &= \min\left\{F_{Y \mid X,W}(y_1 \mid 1,w)p_{1|w}, \frac{p_{1|w}-\uc}{1-\uc} + F_{Y \mid X,W}(y_0|0,w)\frac{p_{0|w}\uc}{1-\uc}\right\} \\
    &= p_{1|w}\min\left\{F_{Y \mid X,W}(y_1 \mid 1,w) , \underline{F}_{Y_0 \mid X,W}(y_0|0,w)\right\},
\end{align*}
where the first equality follows by \eqref{eq:Pul_supp} that $(Y_1,Y_0)$ has no mass on the off-diagonal area, the second equality follows by the result established in Case 4 above, and the fourth equality follows by Lemma \ref{lemma:preliminary_margcdep}.6. Thus we have verified \eqref{eq:Pul_verification}.

Since $\R^2$ is partitioned by these 5 cases, we have established that $\widetilde{\E}[X \mid Y_1,Y_0,W=w] = p^{ul}(Y_1,Y_0,w)$ almost surely, which concludes the proof of Part 1. 

\bigskip

\noindent \textbf{Proof of Part 2}: We prove this by constructing a probability distribution $\widetilde{\P}$ for $(Y_1,Y_0,X)$ conditional on $W=w$ such that for all $y \in \R$ and $x \in \{0,1\}$, we have
\begin{enumerate}
	\item $\widetilde{\P}(Y_1 \leq y \mid W=w) = \overline{F}_{Y_1 \mid W}(y \mid w)$ and  $\widetilde{\P}(Y_0 \leq y \mid W=w) =  \overline{F}_{Y_0 \mid W}(y \mid w)$;
	\item $\widetilde{\P}(X=x \mid W=w) = p_{x|w}$;
	\item $\widetilde{\P}(Y_x \leq y \mid X=x,W=w) = F_{Y \mid X,W}(Y \mid X,w)$;
	\item $\widetilde{\P}(Y_1 \leq y_1, Y_0 \leq y_0 \mid X=x, W=w) = \max\left\{\widetilde{\P}(Y_1\leq y_1 \mid X=x, W=w) + \widetilde{\P}(Y_0\leq y_0 \mid X=x, W=w) - 1, 0\right\}$;
	\item $\widetilde{\P}(X=1 \mid Y_1,Y_0,W=w) \in [\lc,\uc]$ for $\widetilde{\P}$-almost surely.
\end{enumerate}

Let 
\begin{equation}\label{eq:UU_sharpness_jointcdep}
\begin{aligned} 
    \widetilde{\P}(Y_1 \leq y_1, Y_0 \leq y_0, X = x \mid W=w) 
    &= x \max\{F_{Y \mid X,W}(y_1 \mid 1,w) + \overline{F}_{Y_0 \mid X,W}(y_0 \mid 1,w)- 1, 0\} p_{1|w} \\
    &\quad + (1-x) \max\{\overline{F}_{Y_1 \mid X,W}(y_1 \mid 0,w) + F_{Y \mid X,W}(y_0|0,w)-1,0\} p_{0|w}.
\end{aligned}
\end{equation}
By Lemma \ref{lemma:cdf_properties}.1, $\overline{F}_{Y_0 \mid X,W}(y_0 \mid 1,w)$ and $\overline{F}_{Y_1 \mid X,W}(y_1 \mid 0,w)$ are cdfs. Also note that $(u,v) \mapsto \max\{u + v - 1, 0\}$ is the counter-monotonic copula. Following Sklar's Theorem, $\widetilde{\P}$ is a joint distribution function for $(Y_1,Y_0,X)$ conditional on $W=w$. 

Following the same steps as in the proof of Theorem \ref{thm:cdf_sharp_margcdep}, we can show that conditions 1--4 are satisfied because the distribution in \eqref{eq:UU_sharpness_jointcdep} is the same as in \eqref{eq:sharpness_margcdep} but for a specific rather than an arbitrary choice of copulas. By Lemma \ref{lemma:cond_vs_uncond_counter-monotonicity}, $\widetilde{\P}$ leads to the following counter-monotonic joint distribution of $(Y_1,Y_0)$:
\begin{equation}\label{eq:joint_UU_cdf}
    \widetilde{\P}(Y_1\leq y_1, Y_0\leq y_0 \mid W=w)
    = 
    \max\left\{\overline{F}_{Y_1 \mid W}(y_1 \mid w) + \overline{F}_{Y_0 \mid W}(y_0 \mid w) - 1, 0\right\}.
\end{equation}

To show condition 5 holds, and thus complete the proof, we must find a function $p^{uu}$ such that $p^{uu}(Y_1,Y_0,w) = \widetilde{\E}[X \mid Y_1,Y_0,W=w]$ and $p^{uu}(Y_1,Y_0,w) \in [\lc,\uc]$ almost surely under $\widetilde{\P}$.

First consider $p_{1|w} = \lc$, then we have 
\[
	\overline{F}_{Y_1 \mid X,W}(y_1 \mid 0,w) = \overline{F}_{Y_1 \mid W}(y_1 \mid w) = F_{Y \mid X,W}(y_1 \mid 1,w)
\]
and 
\[
	\overline{F}_{Y_0 \mid X,W}(y_0 \mid 1,w) = \overline{F}_{Y_0 \mid W}(y_0 \mid w) = F_{Y \mid X,W}(y_0|0,w).
\] 
This implies the following derivations:
\begin{align*}
    \widetilde{\E}\left[\1(Y_1 \leq y_1, Y_0 \leq y_0)\lc|W=w\right]
    &= p_{1|w}\widetilde{\P}(Y_1\leq y_1, Y_0\leq y_0 \mid W=w) \\
    &= p_{1|w}\max\left\{F_{Y \mid X,W}(y_1 \mid 1,w) + F_{Y \mid X,W}(y_0|0,w) - 1, 0\right\} \\
    &= \widetilde{\P}(Y_1 \leq y_1, Y_0\leq y_0, X=1 \mid W=w).
\end{align*}
The first line holds by $\lc = p_{1|w}$, the second line holds by \eqref{eq:joint_UU_cdf}, and the last line holds by \eqref{eq:UU_sharpness_jointcdep}. Following Lemma \ref{lemma:propscore_equal_condexp_v2}, we conclude that $\widetilde{\P}(X=1 \mid Y_1, Y_0, W=w) = \lc \in [\lc, \uc]$ almost surely. The proof of the case where $p_{1|w} = \uc$ is similar and thus omitted.

Next consider $\lc < p_{1|w} < \uc$. Let $p^{uu}(Y_1,Y_0,w) = p^{uu}(Y_1,Y_0,w;B^{uu})$ defined in \eqref{eq:UU_propscore}, where
\begin{align*}
	B^{uu} &= \frac{\widetilde{\P}(Y_1 = \overline{Q}_1, Y_0 = \overline{Q}_0,X=1 \mid W=w)}{\widetilde{\P}(Y_1 = \overline{Q}_1, Y_0 = \overline{Q}_0 \mid W=w)}
\end{align*}
whenever $\widetilde{\P}(Y_1 = \overline{Q}_1, Y_0 = \overline{Q}_0 \mid W=w) > 0$. Set $B^{uu} = p_{1|w}$ otherwise.

We verify that $B^{uu} \in [\lc,\uc]$ if the denominator is nonzero. 

First, note that the denominator of $B^{uu}$ can be expanded below
\begin{align*}
	&\widetilde{\P}(Y_1 = \overline{Q}_1, Y_0 = \overline{Q}_0 \mid W=w)\\
	&= \widetilde{\P}(Y_1 \leq \overline{Q}_1, Y_0 \leq \overline{Q}_0 \mid W=w) - \widetilde{\P}(Y_1 \leq \overline{Q}_1, Y_0 < \overline{Q}_0 \mid W=w) \\
	&\quad - \widetilde{\P}(Y_1 < \overline{Q}_1, Y_0 \leq \overline{Q}_0 \mid W=w) + \widetilde{\P}(Y_1 < \overline{Q}_1, Y_0 < \overline{Q}_0 \mid W=w)\\
	&= \max\{\overline{F}_{Y_1 \mid W}(\overline{Q}_1 \mid w) + \overline{F}_{Y_0 \mid W}(\overline{Q}_0|w) - 1, 0\} - \max\{\overline{F}_{Y_1 \mid W}(\overline{Q}_1 \mid w) + \overline{F}_{Y_0 \mid W}(\overline{Q}_0-|w) - 1, 0\}\\
	&\quad - \max\{\overline{F}_{Y_1 \mid W}(\overline{Q}_1- \mid w) + \overline{F}_{Y_0 \mid W}(\overline{Q}_0|w)-1,0\} + \max\{\overline{F}_{Y_1 \mid W}(\overline{Q}_1- \mid w) + \overline{F}_{Y_0 \mid W}(\overline{Q}_0-|w)-1,0\}.
\end{align*}
where the second equality holds via \eqref{eq:joint_UU_cdf}. By Lemma \ref{lemma:preliminary_margcdep}.8, we observe that 
\begin{align*}
	\overline{F}_{Y_1 \mid W}(\overline{Q}_1 \mid w) + \overline{F}_{Y_0 \mid W}(\overline{Q}_0|w) - 1 &\geq \frac{\uc- p_{1|w}}{\uc - \lc} + \frac{p_{1|w} - \lc}{\uc - \lc} - 1 = 0\\
	\overline{F}_{Y_1 \mid W}(\overline{Q}_1- \mid w) + \overline{F}_{Y_0 \mid W}(\overline{Q}_0-|w) - 1 &\leq \frac{\uc- p_{1|w}}{\uc - \lc} + \frac{p_{1|w} - \lc}{\uc - \lc} - 1 = 0, 
\end{align*}
hence this expression simplifies to
\begin{align*}
	&\widetilde{\P}(Y_1 = \overline{Q}_1, Y_0 = \underline{Q}_0 \mid W=w) \\
	&= \overline{F}_{Y_1 \mid W}(\overline{Q}_1 \mid w) + \overline{F}_{Y_0 \mid W}(\overline{Q}_0|w) - 1 \\
	&\quad - \max\{\overline{F}_{Y_1 \mid W}(\overline{Q}_1 \mid w) + \overline{F}_{Y_0 \mid W}(\overline{Q}_0-|w) - 1, 0\}  -\max\{\overline{F}_{Y_1 \mid W}(\overline{Q}_1- \mid w) + \overline{F}_{Y_0 \mid W}(\overline{Q}_0|w)-1,0\} \\
	&= \min\left\{ 1-  \overline{F}_{Y_0 \mid W}(\overline{Q}_0-|w), \overline{F}_{Y_1 \mid W}(\overline{Q}_1 \mid w)\right\} + \min\left\{ 1-  \overline{F}_{Y_1 \mid W}(\overline{Q}_1- \mid w), \overline{F}_{Y_0 \mid W}(\overline{Q}_0|w)\right\} -1.
\end{align*}

Second, we expand the numerator of $B^{uu}$. We have that 
\[
	\widetilde{\P}(Y_1 = \overline{Q}_1, Y_0 = \overline{Q}_0,X=1 \mid W=w) = \widetilde{\P}(Y_1 = \overline{Q}_1, Y_0 = \overline{Q}_0 \mid X=1,W=w) p_{1|w}
\]
and that
\begin{align*}
	&\widetilde{\P}(Y_1 = \overline{Q}_1, Y_0 = \overline{Q}_0 \mid X=1, W=w)\\
	&= \max\{F_{Y \mid X,W}(\overline{Q}_1 \mid 1,w) + \overline{F}_{Y_0 \mid X,W}(\overline{Q}_0 \mid 1,w) - 1, 0\} \\
	&\quad - \max\{F_{Y \mid X,W}(\overline{Q}_1 \mid 1,w) + \overline{F}_{Y_0 \mid X,W}(\overline{Q}_0- \mid 1,w) - 1, 0\}\\
	&\quad - \max\{F_{Y \mid X,W}(\overline{Q}_1- \mid 1,w) + \overline{F}_{Y_0 \mid X,W}(\overline{Q}_0 \mid 1,w)-1,0\} \\
	&\quad + \max\{F_{Y \mid X,W}(\overline{Q}_1- \mid 1,w) + \overline{F}_{Y_0 \mid X,W}(\overline{Q}_0- \mid 1,w)-1,0\}\\
	&= F_{Y \mid X,W}(\overline{Q}_1 \mid 1,w) + \overline{F}_{Y_0 \mid X,W}(\overline{Q}_0 \mid 1,w) - 1 - \max\{F_{Y \mid X,W}(\overline{Q}_1 \mid 1,w) + \overline{F}_{Y_0 \mid X,W}(\overline{Q}_0- \mid 1,w) - 1, 0\}\\
	&\quad - \max\{F_{Y \mid X,W}(\overline{Q}_1- \mid 1,w) + \overline{F}_{Y_0 \mid X,W}(\overline{Q}_0 \mid 1,w)-1,0\} + 0\\
	&= \min\left\{ 1-  \overline{F}_{Y_0 \mid X,W}(\overline{Q}_0- \mid 1,w), F_{Y \mid X,W}(\overline{Q}_1 \mid 1,w) \right\} + \min\left\{ 1-  F_{Y \mid X,W}(\overline{Q}_1- \mid 1,w), \overline{F}_{Y_0 \mid X,W}(\overline{Q}_0 \mid 1,w)\right\} -1.
\end{align*}
where the second to last equality follows from Lemma \ref{lemma:preliminary_margcdep}.9, where we note that $\overline{\tau}_1 + \underline{\tau}_1 = 1$.

From Lemma \ref{lemma:preliminary_margcdep}.11, $B^{uu}$ can take four possible values as follows:
\begin{align}
	B^{uu} &= \frac{\widetilde{\P}(Y_1 = \overline{Q}_1, Y_0 = \overline{Q}_0 \mid X=1,W=w)p_{1|w}}{\widetilde{\P}(Y_1 = \overline{Q}_1, Y_0 = \overline{Q}_0 \mid W=w)} \notag\\
        &= \frac{\left( 1-  \overline{F}_{Y_0 \mid X,W}(\overline{Q}_0- \mid 1,w) -  F_{Y \mid X,W}(\overline{Q}_1- \mid 1,w)\right) p_{1|w}}{1-  \overline{F}_{Y_0 \mid W}(\overline{Q}_0-|w) -  \overline{F}_{Y_1 \mid W}(\overline{Q}_1- \mid w)}
        \1\left(
        \begin{array}{c}
        \overline{F}_{Y_1 \mid W}(\overline{Q}_1 \mid w) +  \overline{F}_{Y_0 \mid W}(\overline{Q}_0-|w) \geq 1, \\
        \overline{F}_{Y_1 \mid W}(\overline{Q}_1- \mid w) +  \overline{F}_{Y_0 \mid W}(\overline{Q}_0|w) \geq 1
        \end{array}\right) \label{eq:Buu_term1}\\
	&+ \frac{\left( F_{Y \mid X,W}(\overline{Q}_1 \mid 1,w) -  F_{Y \mid X,W}(\overline{Q}_1- \mid 1,w)\right) p_{1|w}}{\overline{F}_{Y_1 \mid W}(\overline{Q}_1 \mid w) -  \overline{F}_{Y_1 \mid W}(\overline{Q}_1- \mid w)}
		\1\left(
		\begin{array}{c}
		\overline{F}_{Y_1 \mid W}(\overline{Q}_1 \mid w) +  \overline{F}_{Y_0 \mid W}(\overline{Q}_0-|w) < 1,  \\
		\overline{F}_{Y_1 \mid W}(\overline{Q}_1- \mid w) +  \overline{F}_{Y_0 \mid W}(\overline{Q}_0|w) \geq 1
		\end{array}\right)\label{eq:Buu_term2}\\
	&+ \frac{\left( \overline{F}_{Y_0 \mid X,W}(\overline{Q}_0 \mid 1,w) -  \overline{F}_{Y_0 \mid X,W}(\overline{Q}_0- \mid 1,w)\right) p_{1|w}}{\overline{F}_{Y_0 \mid W}(\overline{Q}_0|w) -  \overline{F}_{Y_0 \mid W}(\overline{Q}_0-|w)}
		\1\left(
		\begin{array}{c}
		\overline{F}_{Y_1 \mid W}(\overline{Q}_1 \mid w) +  \overline{F}_{Y_0 \mid W}(\overline{Q}_0-|w) \geq 1, \\
		\overline{F}_{Y_1 \mid W}(\overline{Q}_1- \mid w) +  \overline{F}_{Y_0 \mid W}(\overline{Q}_0|w) < 1
		\end{array}\right)\label{eq:Buu_term3}\\
	&+ \frac{\left( \overline{F}_{Y_0 \mid X,W}(\overline{Q}_0 \mid 1,w) +  F_{Y \mid X,W}(\overline{Q}_1 \mid 1,w) - 1\right) p_{1|w}}{\overline{F}_{Y_0 \mid W}(\overline{Q}_0|w) +  \overline{F}_{Y_1 \mid W}(\overline{Q}_1 \mid w) - 1}
		\1\left(
		\begin{array}{c}
		\overline{F}_{Y_1 \mid W}(\overline{Q}_1 \mid w) +  \overline{F}_{Y_0 \mid W}(\overline{Q}_0-|w) < 1, \\
		\overline{F}_{Y_1 \mid W}(\overline{Q}_1- \mid w) +  \overline{F}_{Y_0 \mid W}(\overline{Q}_0|w) < 1
		\end{array}\right)\label{eq:Buu_term4}
\end{align}
As shown in Lemma \ref{lemma:preliminary_margcdep}.3, terms \eqref{eq:Buu_term2} and \eqref{eq:Buu_term3} lie in $[\lc,\uc]$. 

Next we examine the term \eqref{eq:Buu_term1}, which can be written as below
\begin{align*}
    &\frac{\left( 1-  \overline{F}_{Y_0 \mid X,W}(\overline{Q}_0- \mid 1,w) -  F_{Y \mid X,W}(\overline{Q}_1- \mid 1,w)\right) p_{1|w}}{1-  \overline{F}_{Y_0 \mid W}(\overline{Q}_0-|w) -  \overline{F}_{Y_1 \mid W}(\overline{Q}_1- \mid w)} \\
    &= \frac{p_{1|w}\left(1-F_{Y \mid X,W}(\overline{Q}_0-|0,w) \frac{p_{0|w}\uc}{p_{1|w}(1-\uc)} - F_{Y \mid X,W}(\overline{Q}_1- \mid 1,w)\right)}{1-F_{Y \mid X,W}(\overline{Q}_0-|0,w) \frac{p_{0|w}}{1-\uc} - F_{Y \mid X,W}(\overline{Q}_1- \mid 1,w) \frac{p_{1|w}}{\lc}} \\
    &= \lc + \frac{p_{1|w}-\lc - F_{Y \mid X,W}(\overline{Q}_0-|0,w)\frac{p_{0|w}(\uc-\lc)}{1-\uc}}{\widetilde{\P}(Y_1 = \overline{Q}_1, Y_0 = \overline{Q}_0 \mid W=w)} \\
    &\geq \lc + \frac{p_{1|w}-\lc - \overline{\tau}_0\frac{p_{0|w}(\uc-\lc)}{1-\uc}}{\widetilde{\P}(Y_1 = \overline{Q}_1, Y_0 = \overline{Q}_0 \mid W=w)} \\
    &=\lc,
\end{align*}
where the inequality follows by Lemma \ref{lemma:quantile_props}.3 that $F_{Y \mid X,W}(\overline{Q}_0-|0,w) = F_{Y \mid X,W}(Q_{Y \mid X,W}(\overline{\tau}_0|0,w)-|0,w) \leq \overline{\tau}_0$. Also note that
\begin{align*}
    \frac{\left( 1-  \overline{F}_{Y_0 \mid X,W}(\overline{Q}_0- \mid 1,w) -  F_{Y \mid X,W}(\overline{Q}_1- \mid 1,w)\right) p_{1|w}}{1-  \overline{F}_{Y_0 \mid W}(\overline{Q}_0-|w) -  \overline{F}_{Y_1 \mid W}(\overline{Q}_1- \mid w)}
    &= \uc + \frac{(p_{1|w}-\uc) + F_{Y \mid X,W}(\overline{Q}_1- \mid 1,w)\frac{p_{1|w}(\uc-\lc)}{\lc}}{\widetilde{\P}(Y_1=\overline{Q}_1, Y_0 = \overline{Q}_0 \mid W=w)} \\
    &\leq \uc + \frac{(p_{1|w}-\uc) + \overline{\tau}_1\frac{p_{1|w}(\uc-\lc)}{\lc}}{\widetilde{\P}(Y_1=\overline{Q}_1, Y_0 = \overline{Q}_0 \mid W=w)} \\
    &= \uc,
\end{align*}
where the inequality follows by Lemma \ref{lemma:quantile_props}.3 that $F_{Y \mid X,W}(\overline{Q}_1- \mid 1,w) = F_{Y \mid X,W}(Q_{Y \mid X,W}(\overline{\tau}_1 \mid 1,w)- \mid 1,w) \leq \overline{\tau}_1$. Then we have shown that the term \eqref{eq:Buu_term1} is bounded between $\lc$ and $\uc$. 

Then consider the term \eqref{eq:Buu_term4}. Following the same arguments, we have
\begin{align*}
    &\frac{\left( \overline{F}_{Y_0 \mid X,W}(\overline{Q}_0 \mid 1,w) +  F_{Y \mid X,W}(\overline{Q}_1 \mid 1,w) - 1\right) p_{1|w}}{\overline{F}_{Y_0 \mid W}(\overline{Q}_0|w) +  \overline{F}_{Y_1 \mid W}(\overline{Q}_1 \mid w) - 1}
    = \lc + \frac{\frac{p_{1|w}(\uc - \lc)}{\uc}F_{Y \mid X,W}(\overline{Q}_1 \mid 1,w) - \frac{(\uc-p_{1|w})\lc}{\uc}}{\widetilde{\P}(Y_1 = \overline{Q}_1, Y_0 = \overline{Q}_0 \mid W=w)} \geq \lc \\
    &\frac{\left( \overline{F}_{Y_0 \mid X,W}(\overline{Q}_0 \mid 1,w) +  F_{Y \mid X,W}(\overline{Q}_1 \mid 1,w) - 1\right) p_{1|w}}{\overline{F}_{Y_0 \mid W}(\overline{Q}_0|w) +  \overline{F}_{Y_1 \mid W}(\overline{Q}_1 \mid w) - 1}
    = \uc + \frac{\frac{(p_{1|w}-\lc)(1-\uc)}{1-\lc}-F_{Y \mid X,W}(\overline{Q}_0|0,w)\frac{p_{0|w}(\uc-\lc)}{1-\lc}}{\widetilde{\P}(Y_1 = \overline{Q}_1, Y_0 = \overline{Q}_0 \mid W=w)} \leq \uc,
\end{align*}
where inequalities follow by Lemma \ref{lemma:quantile_props}.2 that $F_{Y \mid X,W}(\overline{Q}_x|x,w) \geq \overline{\tau}_x$ for $x=0,1$. So we have shown that all four terms \eqref{eq:Buu_term1}--\eqref{eq:Buu_term4} are bounded within $[\lc,\uc]$, thus concluding $B^{uu} \in [\lc,\uc]$, which then establishes $p^{uu}(Y_1, Y_0, w) \in [\lc,\uc]$ almost surely.

To finish this proof, we demonstrate that $\widetilde{\E}[X \mid Y_1,Y_0,W=w] = p^{uu}(Y_1,Y_0,w)$ almost surely. To do so, we use Lemma \ref{lemma:propscore_equal_condexp_v2} and show that
\begin{equation}
\label{eq:Puu_verification}
\begin{aligned}
	\widetilde{\E}\left[\1(Y_1 \leq y_1, Y_0 \leq y_0) ~ p^{uu}(Y_1,Y_0,w) \mid W=w\right] 
	&= \widetilde{\P}(Y_1 \leq y_1, Y_0 \leq y_0, X=1 \mid W=w) \\
	&= p_{1|w}\max\{{F}_{Y \mid X,W}(y_1 \mid 1,w) + \overline{F}_{Y_0 \mid X,W}(y_0 \mid 1,w) -1 , 0\}
\end{aligned}
\end{equation}
for all $(y_1,y_0)\in\R^2$. To complete the proof, we break this up into different cases.

\medskip

\noindent \textbf{(Part 2) Case 1:} $y_1 < \overline{Q}_1$ and $y_0 < \overline{Q}_0$.

\medskip

First, note that the joint cdf from \eqref{eq:joint_UU_cdf} implies
\begin{equation}
\label{eq:Puu_supp}
    \widetilde{\P}(Y_1 < \overline{Q}_1, Y_0 < \overline{Q}_0 \mid W=w) =
    \widetilde{\P}(\overline{Q}_1 < Y_1, \overline{Q}_0 < Y_0 \mid W=w) = 0.
\end{equation}
These equalities can be verified by the arguments below:
\begin{align*}
    \widetilde{\P}(Y_1 < \overline{Q}_1, Y_0 < \overline{Q}_0 \mid W=w)
    &= \max\left\{\overline{F}_{Y_1 \mid W}(\overline{Q}_1- \mid w) + \overline{F}_{Y_0 \mid W}(\overline{Q}_0-|w) - 1, 0\right\} \\
    &\leq \max\left\{\frac{\uc-p_{1|w}}{\uc-\lc} + \frac{p_{1|w}-\lc}{\uc-\lc} - 1, 0\right\} \\
    &= 0,
\end{align*}
where the inequality follows by Lemma \ref{lemma:preliminary_margcdep}.8, and similarly,
\begin{align*}
    \widetilde{\P}(\overline{Q}_1 < Y_1, \overline{Q}_0 < Y_0 \mid W=w)
    &= 1 - \widetilde{\P}(Y_1 \leq \overline{Q}_1 \mid W=w) - \widetilde{\P}(Y_0 \leq \overline{Q}_0 \mid W=w) \\
    &\quad + \widetilde{\P}(Y_1\leq \overline{Q}_1, Y_0 \leq \overline{Q}_0 \mid W=w) \\
    &= 1 - \overline{F}_{Y_1 \mid W}(\overline{Q}_1 \mid w) - \overline{F}_{Y_0 \mid W}(\overline{Q}_0|w) \\
    &\quad + \max\left\{\overline{F}_{Y_1 \mid W}(\overline{Q}_1 \mid w) + \overline{F}_{Y_0 \mid W}(\overline{Q}_0|w) - 1, 0\right\} \\
    &= 1 - \min\left\{1, \overline{F}_{Y_1 \mid W}(\overline{Q}_1 \mid w) + \overline{F}_{Y_0 \mid W}(\overline{Q}_0|w)\right\} \\
    &\leq 1 - \min\left\{1,\frac{\uc-p_{1|w}}{\uc-\lc} + \frac{p_{1|w}-\lc}{\uc-\lc}\right\} \\
    &=0,
\end{align*}
where the inequality follows by Lemma \ref{lemma:preliminary_margcdep}.8. 

On the one hand, \eqref{eq:Puu_supp} implies
\begin{align*}
    \widetilde{\E}(\1(Y_1\leq y_1, Y_0\leq y_0)p^{uu}(Y_1, Y_0, w) \mid W=w)
    &= p_{1|w} \widetilde{\P}(Y_1\leq y_1, Y_0\leq y_0 \mid W=w) \\
    &\leq p_{1|w} \widetilde{\P}(Y_1 < \overline{Q}_1, Y_0 < \overline{Q}_0 \mid W=w) \\
    &= 0.
\end{align*}
This shows $\widetilde{\E}(\1(Y_1\leq y_1, Y_0\leq y_0)p^{uu}(Y_1, Y_0, w) \mid W=w) = 0$ due to the construction that $p^{uu}(Y_1, Y_0, w)$ is non-negative. On the other hand, 
\begin{align*}
\widetilde{\P}(Y_1\leq y_1, Y_0\leq y_0, X=1 \mid W=w) 
    &=p_{1|w} \max\left\{F_{Y \mid X,W}(y_1 \mid 1,w) + \overline{F}_{Y_0 \mid X,W}(y_0 \mid 1,w) - 1, 0\right\} \\
    & \leq  p_{1|w} \max\left\{F_{Y \mid X,W}(\overline{Q}_1- \mid 1,w) + \overline{F}_{Y_0 \mid X,W}(\overline{Q}_0- \mid 1,w) - 1, 0\right\} \\
    & \leq p_{1|w} \max\left\{\overline{\tau}_1 + \underline{\tau}_1 - 1, 0\right\} \\
    &= 0,
\end{align*}
where the second inequality follows by Lemma \ref{lemma:preliminary_margcdep}.9. 
This implies $\widetilde{\P}(Y_1\leq y_1, Y_0\leq y_0, X=1 \mid W=w) = 0$, thus  establishing \eqref{eq:Puu_verification}, as desired.

\medskip

\noindent \textbf{(Part 2) Case 2:} $y_1 \geq \overline{Q}_1$, $y_0 < \overline{Q}_0$.

\medskip

First, note that \eqref{eq:Puu_supp} implies that
\begin{align*}
    &\widetilde{\E}(\1(Y_1\leq y_1, Y_0\leq y_0)p^{uu}(Y_1, Y_0, w) \mid W=w) \\
    &= \widetilde{\E}(\1(\overline{Q}_1 \leq Y_1\leq y_1, Y_0\leq y_0)\uc|W=w) + \widetilde{\E}(\1(Y_1 < \overline{Q}_1, Y_0\leq y_0)p_{1|w}|W=w) \\
    &= \widetilde{\E}(\1(\overline{Q}_1 \leq Y_1\leq y_1, Y_0\leq y_0)\uc|W=w) + \widetilde{\E}(\1(Y_1 < \overline{Q}_1, Y_0\leq y_0)\uc|W=w) \\
    &= \widetilde{\E}(\1(Y_1\leq y_1, Y_0\leq y_0)\uc|W=w),
\end{align*}
where the second equality follows by the fact that $\widetilde{\P}$ takes no mass on $\{Y_1 < \overline{Q}_1, Y_0 \leq y_0\} \subseteq \{Y_1 < \overline{Q}_1, Y_0 < \overline{Q}_0\}$ by \eqref{eq:Puu_supp}. Next we expand the last expression 
\begin{align*}
    & \widetilde{\E}(\1(Y_1\leq y_1, Y_0\leq y_0)\uc|W=w) \\
    &= \uc \widetilde{\P}(Y_1 \leq y_1, Y_0 \leq y_0 \mid W=w) \\
    &= \uc \max\left\{\overline{F}_{Y_1 \mid W}(y_1 \mid w) + \overline{F}_{Y_0 \mid W}(y_0 \mid w) - 1, 0\right\} \\
    &= \uc \max\left\{\frac{\uc-p_{1|w}}{\uc} + F_{Y \mid X,W}(y_1 \mid 1,w) \frac{p_{1|w}}{\uc} + F_{Y \mid X,W}(y_0|0,w)\frac{p_{0|w}}{1-\uc} - 1, 0\right\} \\
    &= p_{1|w} \max\left\{F_{Y \mid X,W}(y_1 \mid 1,w) + \frac{\uc-p_{1|w}}{p_{1|w}} + F_{Y \mid X,W}(y_0|0,w) \frac{p_{0|w}\uc}{p_{1|w}(1-\uc)} - \frac{\uc}{p_{1|w}}, 0\right\} \\
    &= p_{1|w} \max\left\{F_{Y \mid X,W}(y_1 \mid 1,w)  + F_{Y \mid X,W}(y_0|0,w) \frac{p_{0|w}\uc}{p_{1|w}(1-\uc)} - 1, 0\right\} \\
    &= p_{1|w} \max\left\{F_{Y \mid X,W}(y_1 \mid 1,w) + \overline{F}_{Y_0 \mid X,W}(y_0 \mid 1,w) -1, 0\right\},
\end{align*}
thus establishing \eqref{eq:Puu_verification}, as desired.

\medskip

\noindent \textbf{(Part 2) Case 3:} $y_1 < \overline{Q}_1$, $y_0 \geq \overline{Q}_0$.

\medskip

Similar to the proof of case 2 above, we have
\[
    \widetilde{\E}(\1(Y_1\leq y_1, Y_0\leq y_0) p^{uu}(Y_1,Y_0,w) \mid W=w)
    =
    \widetilde{\E}(\1(Y_1\leq y_1, Y_0\leq y_0) \lc|W=w)
\]
due to \eqref{eq:Puu_supp}. Next we expand the expression on the right hand side.
\begin{align*}
    &\widetilde{\E}(\1(Y_1\leq y_1, Y_0\leq y_0) \lc|W=w) \\
    &= \lc \max\left\{\overline{F}_{Y_1 \mid W}(y_1 \mid w) + \overline{F}_{Y_0 \mid W}(y_0 \mid w) - 1, 0\right\} \\
    &= \lc \max\left\{F_{Y \mid X,W}(y_1 \mid 1,w) \frac{p_{1|w}}{\lc} + \frac{p_{1|w} - \lc}{1-\lc} + F_{Y \mid X,W}(y_0|0,w) \frac{p_{0|w}}{1-\lc} - 1, 0\right\} \\
    &= p_{1|w} \max\left\{F_{Y \mid X,W}(y_1 \mid 1,w) + \frac{(p_{1|w}-\lc)\lc}{p_{1|w}(1-\lc)} + F_{Y \mid X,W}(y_0|0,w)\frac{p_{0|w}\lc}{p_{1|w}(1-\lc)} - \frac{\lc}{p_{1|w}}, 0\right\} \\
    &= p_{1|w} \max\left\{F_{Y \mid X,W}(y_1 \mid 1,w) + \frac{p_{1|w}-\lc}{p_{1|w}(1-\lc)} + F_{Y \mid X,W}(y_0|0,w) \frac{p_{0|w}\lc}{p_{1|w}(1-\lc)} - 1, 0\right\} \\
    &= p_{1|w} \max\left\{F_{Y \mid X,W}(y_1 \mid 1,w) + \overline{F}_{Y_0 \mid X,W}(y_0 \mid 1,w) - 1, 0\right\},
\end{align*}
thus establishing \eqref{eq:Puu_verification}, as desired.

\medskip
\noindent \textbf{(Part 2) Case 4:} $y_1 = \overline{Q}_1$, $y_0 = \overline{Q}_0$.
\medskip

Note that the equality \eqref{eq:Puu_verification} can be established following the same arguments from the proof of Part 1, case 4. Once the results are established for cases 1--3, the equality \eqref{eq:Puu_verification} holds for $y_1 = \overline{Q}_1$, $y_0 = \overline{Q}_0$ by applying monotone convergence theorem and continuity of measure. To this end, the proof is omitted.

\medskip
\noindent \textbf{(Part 2) Case 5:} $(y_1,y_0) \geq (\overline{Q}_1, \overline{Q}_0)$.
\medskip

We start by noting that 
\begin{align*}
    & \widetilde{\E}(\1(Y_1\leq y_1, Y_0\leq y_0)p^{uu}(Y_1,Y_0,w) \mid W=w) \\
    &= p_{1|w} \widetilde{\P}(Y_1 \in (\overline{Q}_1, y_1], Y_0 \in (\overline{Q}_0, y_0]|W=w) + \lc \widetilde{\P}(Y_1 \leq \overline{Q}_1, Y_0 \in (\overline{Q}_0, y_0]|W=w) \\
    &\quad + \uc \widetilde{\P}(Y_1 \in (\overline{Q}_1, y_1], Y_0 \leq \overline{Q}_0) \mid W=w)  + \widetilde{\E}(\1(Y_1\leq \overline{Q}_1, Y_0\leq \overline{Q}_0)p^{uu}(Y_1,Y_0,w) \mid W=w) \\
    &= \lc \left(\widetilde{\P}(Y_1 \leq \overline{Q}_1, Y_0 \in (\overline{Q}_0, y_0]|W=w) + \widetilde{\P}(Y_1 > \overline{Q}_1, Y_0 \in (\overline{Q}_0, y_0]) \mid W=w\right) \\
    &\quad + \uc \left(\widetilde{\P}(Y_1 \in (\overline{Q}_1, y_1], Y_0 \leq \overline{Q}_0) \mid W=w) + \widetilde{\P}(Y_1 \in (\overline{Q}_1, y_1], Y_0 > \overline{Q}_0)) \mid W=w\right) \\
    &\quad + \widetilde{\E}(\1(Y_1\leq \overline{Q}_1, Y_0\leq \overline{Q}_0)p^{uu}(Y_1,Y_0,w) \mid W=w) \\
    &= \lc \widetilde{\P}(Y_0\leq (\overline{Q}_0, y_0]|W=w) + \uc\widetilde{\P}(Y_1 \in (\overline{Q}_1, y_1]|W=w) + \widetilde{\E}(\1(Y_1\leq \overline{Q}_1, Y_0\leq \overline{Q}_0)p^{uu}(Y_1,Y_0,w) \mid W=w),
\end{align*}
where the second equality follows by \eqref{eq:Puu_supp} that $\widetilde{\P}$ takes no mass on diagonal area. Next we expand the last line.
\begin{align*}
    &\lc \widetilde{\P}(Y_0\leq (\overline{Q}_0, y_0]|W=w) + \uc\widetilde{\P}(Y_1 \in (\overline{Q}_1, y_1]|W=w) + \widetilde{\E}(\1(Y_1\leq \overline{Q}_1, Y_0\leq \overline{Q}_0)p^{uu}(Y_1,Y_0) \mid W=w) \\
    &= \lc\left[\overline{F}_{Y_0 \mid W}(y_0 \mid w) - \overline{F}_{Y_0 \mid W}(\overline{Q}_0|w)\right]
    	+ \uc\left[\overline{F}_{Y_1 \mid W}(y_1 \mid w) - \overline{F}_{Y_1 \mid W}(\overline{Q}_0|w)\right] \\
    &\quad + \widetilde{\P}(Y_1 \leq \overline{Q}_1, Y_0\leq \overline{Q}_0, X=1 \mid W=w) \\
    &= \frac{p_{0|w}\lc}{1-\lc}\left(F_{Y \mid X,W}(y_0|0,w) - F_{Y \mid X,W}(\overline{Q}_0|0,w)\right) + p_{1|w}\left(F_{Y \mid X,W}(y_1 \mid 1,w) - F_{Y \mid X,W}(\overline{Q}_1 \mid 1,w)\right) \\
    &\quad + p_{1|w}\max\left\{F_{Y \mid X,W}(\overline{Q}_1 \mid 1,w) + \overline{F}_{Y_0 \mid X,W}(\overline{Q}_0 \mid 1,w) - 1, 0\right\} \\
    &= \frac{p_{0|w}\lc}{1-\lc}\left(F_{Y \mid X,W}(y_0|0,w) - F_{Y \mid X,W}(\overline{Q}_0|0,w)\right) + p_{1|w}\left(F_{Y \mid X,W}(y_1 \mid 1,w) - F_{Y \mid X,W}(\overline{Q}_1 \mid 1,w)\right)  \\
    &\quad + p_{1|w} \left(F_{Y \mid X,W}(\overline{Q}_1 \mid 1,w) + \overline{F}_{Y_0 \mid X,W}(\overline{Q}_0 \mid 1,w) - 1\right) \\
    &= p_{1|w}\left[F_{Y \mid X,W}(y_1 \mid 1,w) + \frac{p_{1|w}-\lc}{p_{1|w}(1-\lc)} + \frac{p_{0|w}\lc}{p_{1|w}(1-\lc)} F_{Y \mid X,W}(y_0|0,w) - 1\right] \\
    &= p_{1|w}\left[F_{Y \mid X,W}(y_1 \mid 1,w) + \overline{F}_{Y_0 \mid X,W}(y_0 \mid 1,w) - 1\right] \\
    &= p_{1|w}\max\left\{F_{Y \mid X,W}(y_1 \mid 1,w) + \overline{F}_{Y_0 \mid X,W}(y_0 \mid 1,w) - 1, 0\right\}.
\end{align*}
The third and the last equality hold by the following derivation 
\begin{align*}
    F_{Y \mid X,W}(y_1 \mid 1,w) + \overline{F}_{Y_0 \mid X,W}(y_0 \mid 1,w) - 1
    &\geq F_{Y \mid X,W}(\overline{Q}_1 \mid 1,w) + \overline{F}_{Y_0 \mid X,W}(\overline{Q}_0 \mid 1,w) - 1 \\
    &\geq \overline{\tau}_1 + \underline{\tau}_1 - 1 \\
    &= 0,
\end{align*}
where the second inequality follows by Lemma \ref{lemma:preliminary_margcdep}.9. Hence we have established  \eqref{eq:Puu_verification}, as desired.

Since $\R^2$ is partitioned by these 5 cases, we haven shown that $\widetilde{\E}[X \mid Y_1,Y_0,W=w] = p^{uu}(Y_1,Y_0,w)$ almost surely, which concludes the proof of Part 2.

\bigskip

\noindent \textbf{Proof of Part 3:} One can show that $(\underline{F}_{Y_1 \mid W}, \overline{F}_{Y_0 \mid W}, \overline{C}_{1,0 \mid X,W})$ can be achieved by the joint distribution of $(Y_1, Y_0, X)$ conditional on $W=w$ constructed as below:
\begin{align*}
    &\widetilde{\P}(Y_1\leq y_1, Y_0\leq y_0, X=1 \mid W=w) \\  
    &= x\min\left\{F_{Y \mid X,W}(y_1 \mid 1,w), \overline{F}_{Y_0 \mid X,W}(y_0 \mid 1,w)\right\}p_{1|w} + 
    (1-x)\min\left\{\underline{F}_{Y_1 \mid X,W}(y_1 \mid 0,w), F_{Y \mid X,W}(y_0|0,w)\right\}p_{0|w}.
\end{align*}
It can be verified that this joint distribution satisfies the following 5 conditions: for all $y \in \R$ and $x \in \{0,1\}$, 
\begin{enumerate}
    \item $\widetilde{\P}(Y_1\leq y \mid W=w) = \underline{F}_{Y_1 \mid W}(y \mid w)$ and $\widetilde{\P}(Y_0\leq y \mid W=w) = \overline{F}_{Y_0}(y \mid w)$;
    \item $\widetilde{\P}(X=x \mid W=w) = p_{x|w}$;
    \item $\widetilde{\P}(Y_x\leq y \mid X=x,W=w) = F_{Y \mid X,W}(Y \mid X,w)$;
    \item $\widetilde{\P}(Y_1\leq y_1, Y_0\leq y_0 \mid X=x, W=w) = \min\left\{\widetilde{\P}(Y_1\leq y_1 \mid X=x,W=w), \widetilde{\P}(Y_0\leq y_0 \mid X=x,W=w)\right\}$;
    \item $\widetilde{\E}(X \mid Y_1, Y_0, W=w) = p^{lu}(Y_1,Y_0,w;B^{lu}) \in [\lc,\uc]$, for $\widetilde{\P}$-almost surely with 
    \[
        B^{lu} = \frac{\widetilde{\P}(Y_1 = y_1, Y_0 = y_0, X=1 \mid W=w)}{\widetilde{\P}(Y_1 = y_1, Y_0 = y_0 \mid W=w)}.
    \]
\end{enumerate}
The arguments are similar to the proof of Part 1 and thus omitted.

\bigskip

\noindent \textbf{Proof of Part 4}: One can show that $(\underline{F}_{Y_1 \mid W}, \underline{F}_{Y_0 \mid W}, \underline{C}_{1,0 \mid X,W})$ can be achieved by the joint distribution of $(Y_1, Y_0, X)$ conditional on $W=w$ constructed as below:
\begin{align*}
    \widetilde{\P}(Y_1\leq y_1, Y_0\leq y_0, X=1 \mid W=w) 
    &= 
    x\max\left\{F_{Y \mid X,W}(y_1 \mid 1,w) + \underline{F}_{Y_0 \mid X,W}(y_0 \mid 1,w) - 1, 0\right\}p_{1|w} \\
    &\quad + (1-x)\max\left\{\underline{F}_{Y_1 \mid X,W}(y_1 \mid 0,w) + F_{Y \mid X,W}(y_0|0,w) - 1, 0\right\}p_{0|w}.
\end{align*}
It can be verified that this joint distribution satisfies the following four conditions: for all $y \in \R$ and $x \in \{0,1\}$,
\begin{enumerate}
    \item $\widetilde{\P}(Y_1\leq y \mid W=w) = \underline{F}_{Y_1 \mid W}(y \mid w)$ and $\widetilde{\P}(Y_0\leq y \mid W=w) = \underline{F}_{Y_0 \mid W}(y \mid w)$;
    \item $\widetilde{\P}(X=x \mid W=w) = p_{x|w}$;
    \item $\widetilde{\P}(Y_x\leq y \mid X=x,W=w) = F_{Y \mid X,W}(Y \mid X,w)$;
    \item $\widetilde{\P}(Y_1\leq y_1, Y_0\leq y_0 \mid X=x, W=w) = \max\left\{\widetilde{\P}(Y_1\leq y_1 \mid X=x,W=w) + \widetilde{\P}(Y_0\leq y_0 \mid X=x,W=w) - 1, 0\right\}$;
    \item $\widetilde{\E}(X \mid Y_1, Y_0, W=w) = p^{ll}(Y_1,Y_0, w;B^{ll}) \in [\lc,\uc]$, for $\widetilde{\P}$-almost surely with 
    \[
        B^{ll} = \frac{\widetilde{\P}(Y_1 = y_1, Y_0 = y_0, X=1 \mid W=w)}{\widetilde{\P}(Y_1 = y_1, Y_0 = y_0 \mid W=w)}.
    \]
\end{enumerate}
The arguments are similar to the proof of Part 2 and thus omitted.
\end{proof}

Suppose $\{F^{k}_{Y_1,Y_0,X}\}_{k=1}^K$ is a set of valid cdfs of $(Y_1, Y_0, X)$ whose support of $X$ is $\{0,1\}$. Consider a mixture of cdfs defined as 
\begin{equation}
\label{eq:mix_dist}
	F^{\text{mix}}_{Y_1,Y_0,X}(y_1, y_0, x) = \sum_{k=1}^K a_k F^k_{Y_1,Y_0,X}(y_1, y_0, x)\end{equation}
where $a_k \in [0,1]$ for all $k \in \{1,\ldots,K\}$, $\sum_{k=1}^K a_k = 1$, and $(y_1, y_0, x)\in \R^2\times \{0,1\}$. 

Let $p^k(Y_1, Y_0) \coloneqq \Exp^k[X\mid Y_1, Y_0]$, where $\E^k$ denotes the expectation under cdf $F_{Y_1, Y_0, X}^k$ for $k \in \{1,\ldots,K,\text{mix}\}$. The next lemma will be used to show that mixtures of distributions satisfying joint $c$-dependence also satisfy joint $c$-dependence.
\begin{lemma}
\label{lemma:propscore_cvx_comb}
There exists a sequence of Borel measurable function $\{f_k(\cdot,\cdot)\}_{k=1}^K: \R^2 \to \R$ such that
\begin{enumerate}
	\item $f_k(Y_1, Y_0) \in [0,1]$ for each $k \in \{1,\ldots,K\}$, and $\sum_{k=1}^K f_k(Y_1, Y_0) = 1$,
	\item $p^{\text{mix}}(Y_1, Y_0) = \sum_{k=1}^K f_k(Y_1, Y_0) p^k(Y_1, Y_0)$,
\end{enumerate} 
almost surely under the distribution $F^{\text{mix}}_{Y_1,Y_0,X}$.
\end{lemma}

\begin{proof}[Proof of Lemma \ref{lemma:propscore_cvx_comb}]
Let $\P^k$ denote the probability taken under cdf $F_{Y_1, Y_0, X}^k$ for $k \in \{1,\ldots,K,\text{mix}\}$. Then it follows by the definition of conditional probability that
\begin{align*}
	\P^\text{mix}(Y_1\leq y_1, Y_0\leq y_0, X = 1)
&= \Exp^\text{mix}\left[\indicator[Y_1\leq y_1, Y_0\leq y_0]p^\text{mix}(Y_1, Y_0)\right] \\
&= \int_{u \leq y_1, v\leq y_0} p^\text{mix}(u, v) \, d F^\text{mix}_{Y_1, Y_0},
\end{align*}
where the last line denotes the Lebesgue-Stieltjes integral with respect to the cdf $F_{Y_1, Y_0}^\text{mix}$.

Likewise, for each $k \in \{1,\ldots,K\}$, we have 
\begin{align*}
 \P^k(Y_1 \leq y_1, Y_0 \leq y_0, X = 1) = \int_{u \leq y_1, v\leq y_0} p^k(u, v) \, d F^k_{Y_1, Y_0}.
\end{align*}

Since \eqref{eq:mix_dist} implies that 
\[
	\P^\text{mix}(Y_1 \leq y_1, Y_0 \leq y_0, X = 1) = \sum_{k=1}^K a_k \P^k(Y_1\leq y_1, Y_0\leq y_0, X=1),
\]
we have
\begin{equation}
\label{eq:prob_cvx_comb}
	\int_{w \leq y_1, v\leq y_0} p^\text{mix}(w, v) \, d F^\text{mix}_{Y_1, Y_0} = \sum_{k=1}^K a_k \int_{u\leq y_1, v\leq y_0} p^k(u,v) \, dF^k_{Y_1,Y_0}.
\end{equation}

Following the Carathéodory extension theorem (e.g., \citet[Theorem 1.4.9]{Ash2000}), there exists unique Lebesgue-Stieltjes measures $\nu^{\text{mix}}$ and $\{\nu^k\}_{k=1}^K$ defined on $(\R^2, \mathcal{B}(\R^2))$ that  are consistent with cdfs $F^{\text{mix}}_{Y_1, Y_0}$ and $\{F^k_{Y_1,Y_0}\}_{k=1}^K$, respectively. Combined with \eqref{eq:mix_dist}, this implies
\begin{equation}\label{eq:measure_cvx_comb_rectangle}
	\nu^{\text{mix}}(A) =  \sum_{k=1}^K a_k \nu^k(A), \quad \text{for all } A \in \mathcal{A} \coloneqq \{(-\infty, y_0]\times (-\infty, y_1]: (y_0, y_1) \in \R^2 \}.
\end{equation}
It can be seen that $\mathcal{A}$ is a $\pi$-system, $\sigma(\mathcal{A}) = \mathcal{B}(\R^2)$, and the class of sets satisfying \eqref{eq:measure_cvx_comb_rectangle} constitutes a $\lambda$-system. Following from $\pi-\lambda$ Theorem, 
\begin{equation}\label{eq:measure_cvx_comb}
	\nu^{\text{mix}}(A) = \sum_{k=1}^K a_k \nu^k(A), \quad \text{for all } A \in \mathcal{B}(\R^2).
\end{equation}
From the above identity \eqref{eq:measure_cvx_comb}, we note that $a_k \nu^k \ll \nu^{\text{mix}}$ for all $k \in \{1,\ldots,K\}$. By Radon-Nikodym Theorem (e.g., \citet[p.386, Problem 54.1]{Royden2010}), there exist nonnegative Borel measurable functions $d (a_k \nu^k)/d \nu^{\text{mix}}$ such that the following equalities hold for all Borel sets $A \in \mathcal{B}(\R^2)$ and for each $k \in \{1,\ldots,K\}$:
\begin{align*}
	\int_A  p^k(u,v)  \frac{d (a_k \nu^k)}{d \nu^{\text{mix}}} \, d\nu^{\text{mix}} = \int_A p^{k}(u,v) \, d(a_k \nu^k)  = \int_A a_k p^{k}(u,v) \, d\nu^k.
\end{align*}

Taking $A = (-\infty, y_1] \times (-\infty, y_0]$ and combining these equalities across $k \in \{1,\ldots, K\}$ then gives 
\begin{align*}
\int_{u\leq y_1, v\leq y_0}  \sum_{k=1}^K p^{k}(u,v)  \frac{d (a_k \nu^k)}{d \nu^{\text{mix}}} \, d\nu^{\text{mix}}  &= \sum_{k=1}^K a_k \int_{u\leq y_1, v\leq y_0} p^k(u,v) \, d\nu^k\\
&= \int_{u\leq y_1, v\leq y_0} p^{\text{mix}}(u,v) \, d\nu^{\text{mix}}
\end{align*}
The second equality holds by \eqref{eq:prob_cvx_comb}. Since $\mathcal{A}$ constitutes a $\pi$-system, and $\R^2$ can be written as a countable union of elements in the class $\mathcal{A}$, applying \citet[Theorem 16.10.(iii)]{Billingsley1995} then leads to the following equality: 
\[
	p^{\text{mix}}(u,v) = \sum_{k=1}^K p^k(u,v)\frac{d(a_k\nu^k)}{d\nu^{\text{mix}}}(u,v)
\]
almost surely with respect to the measure $\nu^{\text{mix}}$ on $(\R^2, \mathcal{B}(\R^2))$. Replacing $(u,v)$ with random variables $(Y_1,Y_0)$, the above equality is equivalent to
\[
	p^{\text{mix}}(Y_1, Y_0) = \sum_{k=1}^K p^k(Y_1,Y_0)\frac{d(a_k\nu^k)}{d\nu^{\text{mix}}}(Y_1,Y_0)
\]
almost surely under the distribution $F^{\text{mix}}_{Y_1, Y_0}$. Next we show that the weights add up to one and they are non-negative.

Following from \eqref{eq:measure_cvx_comb}, the conclusion in \citet[p.386, Problem 54.2]{Royden2010} implies
\[
	\sum_{k=1}^K \frac{d (a_k \nu^k)}{d \nu^{\text{mix}}} = \frac{d\nu^{\text{mix}}}{d\nu^{\text{mix}}} = 1, \quad \text{almost surely-}\nu^{\text{mix}}.
\]
By definition, Radon-Nikodym derivatives are nonnegative. So we have
\[
	\frac{d (a_k \nu^k)}{d \nu^{\text{mix}}} = 1 - \sum_{j\neq k}\frac{d (a_j \nu^j)}{d \nu^m} \in [0,1],  \quad \text{almost surely-}\nu^{\text{mix}}.
\]
Therefore, we have shown that $p^{\text{mix}}(Y_1, Y_0)$ is a convex combination of $\{p^k(Y_1,Y_0)\}_{k=1}^K$ almost surely, with weights $d(a_k\nu^k)/d\nu^{\text{mix}}$ being a non-negative measurable function of $(Y_1, Y_0)$, as desired.
\end{proof}

\begin{proof}[Proof of Theorem \ref{thm:cdf_sharp_jointcdep}]
Fix a $w\in\supp(W)$ and $(\varepsilon, \gamma) \in [0,1]^2$, we prove this by constructing a probability distribution $\widetilde{\P}$ for $(Y_1,Y_0,X)$ conditional on $W=w$ such that for all $y \in \R$ and $x \in \{0,1\}$, the following conditions hold
\begin{enumerate}
	\item $\widetilde{\P}(Y_1 \leq y \mid W=w) = \varepsilon \underline{F}_{Y_1 \mid W}(y \mid w) + (1-\varepsilon) \overline{F}_{Y_1 \mid W}(y \mid w)$ and  \\
	$\widetilde{\P}(Y_0 \leq y \mid W=w) =  \gamma \underline{F}_{Y_0 \mid W}(y \mid w) + (1-\gamma)\overline{F}_{Y_0 \mid W}(y \mid w)$;
	\item $\widetilde{\P}(X=x \mid W=w) = p_{x|w}$;
	\item $\widetilde{\P}(Y_x \leq y \mid X=x,W=w) = F_{Y \mid X,W}(Y \mid X,w)$;
	\item $\widetilde{\P}(X=1 \mid Y_1,Y_0,W=w) \in [\lc,\uc]$ for $\widetilde{\P}$-almost surely.
\end{enumerate}
Compared to the proof of Theorem \ref{thm:cdf_sharp_margcdep}, we remove the requirement that the copulas between $Y_1$ and $Y_0$ conditional on $(X,W)$ can be arbitrary.

As in Lemma \ref{lemma:attainability_jointcdep}, we have constructed the following four joint (conditional) cdfs of $(Y_1,Y_0,X) \mid W=w$:
\begin{align*}
    F^{ul}(y_1, y_0, x|w) 
    &= x\min\left\{F_{Y \mid X,W}(y_1 \mid 1,w), \underline{F}_{Y_0 \mid X}(y_0 \mid 1)\right\}p_{1|w} \\
    &\quad + (1-x)\min\left\{\overline{F}_{Y_1 \mid X,W}(y_1 \mid 0,w), F_{Y \mid X,W}(y_0|0,w)\right\}p_{0|w} \\
    F^{lu}(y_1, y_0, x|w) 
    &= x\min\left\{F_{Y \mid X,W}(y_1 \mid 1), \overline{F}_{Y_0 \mid X,W}(y_0 \mid 1,w)\right\}p_{1|w} \\
    &\quad + (1-x)\min\left\{\underline{F}_{Y_1 \mid X,W}(y_1 \mid 0,w), F_{Y \mid X,W}(y_0|0,w)\right\}p_{0|w} \\
    F^{uu}(y_1, y_0, x|w) 
    &= x\max\left\{F_{Y \mid X,W}(y_1 \mid 1,w) + \overline{F}_{Y_0 \mid X,W}(y_0 \mid 1,w) - 1, 0\right\}p_{1|w} \\
    &\quad + (1-x)\min\left\{\overline{F}_{Y_1 \mid X,W}(y_1 \mid 0,w) + F_{Y \mid X,W}(y_0|0,w) - 1, 0\right\}p_{0|w} \\
    F^{ll}(y_1, y_0, x|w) 
    &= x\max\left\{F_{Y \mid X,W}(y_1 \mid 1,w) + \underline{F}_{Y_0 \mid X,W}(y_0 \mid 1,w) - 1, 0\right\}p_{1|w} \\
    &\quad + (1-x)\min\left\{\underline{F}_{Y_1 \mid X,W}(y_1 \mid 0,w) + F_{Y \mid X,W}(y_0|0,w) - 1, 0\right\}p_{0|w}.
\end{align*}
Let the joint distribution of $(Y_1, Y_0, X) \mid W=w$ be defined as below:
\begin{equation}\label{eq:mixture_sharpness_jointcdep}
\begin{aligned}
    \widetilde{\P}(Y_1\leq y_1, Y_0\leq y_0, X=x \mid W=w)  
    &= \gamma \left[ \varepsilon F^{ll}(y_1,y_0,x|w) + (1-\varepsilon)F^{ul}(y_1,y_0,x|w) \right] \\
    &\quad + (1-\gamma) \left[\varepsilon F^{lu}(y_1,y_0,x|w) + (1-\varepsilon) F^{uu}(y_1,y_0,x|w)\right].
\end{aligned}
\end{equation}

As shown in Lemma \ref{lemma:attainability_jointcdep}, the functions $F^{ll}$, $F^{ul}$, $F^{lu}$, and $F^{uu}$ are valid cdfs, hence their convex combination \eqref{eq:mixture_sharpness_jointcdep} also yields a valid cdf. Next we verify that this cdf satisfies conditions 1-4 listed above, thus concluding that $(\varepsilon\underline{F}_{Y_1 \mid W} + (1-\varepsilon)\overline{F}_{Y_1 \mid W}, \gamma\underline{F}_{Y_0 \mid W} + (1-\gamma)\overline{F}_{Y_0 \mid W})$ is in the identified set.

\medskip
\noindent \textbf{Verifying Condition 1:} For $y \in \R$, we have 
\begin{align*}
    &\widetilde{\P}(Y_1 \leq y \mid W=w) \\
    &= \lim_{y_0 \to +\infty}\widetilde{\P}(Y_1 \leq y, Y_0 \leq y_0, X=1 \mid W=w) + \lim_{y_0 \to +\infty}\widetilde{\P}(Y_1 \leq y, Y_0 \leq y_0, X=0 \mid W=w) \\
    &= \sum_{x=0,1}\gamma \left[ \varepsilon \lim_{y_0 \to +\infty}F^{ll}(y,y_0,x|w) + (1-\varepsilon)\lim_{y_0 \to +\infty}F^{ul}(y,y_0,x|w) \right] \\
    &\quad + \sum_{x=0,1} (1-\gamma) \left[\varepsilon \lim_{y_0 \to +\infty} F^{lu}(y,y_0,x|w) + (1-\varepsilon) \lim_{y_0 \to +\infty}F^{uu}(y,y_0,x|w)\right] \\
    &= \gamma \left[\varepsilon \underline{F}_{Y_1 \mid W}(y \mid w) + (1-\varepsilon) \overline{F}_{Y_1 \mid W}(y \mid w)\right] + (1-\gamma) \left[\varepsilon \underline{F}_{Y_1 \mid W}(y \mid w) + (1-\varepsilon) \overline{F}_{Y_1 \mid W}(y \mid w)\right] \\
    &= \varepsilon \underline{F}_{Y_1 \mid W}(y \mid w) + (1-\varepsilon) \overline{F}_{Y_1 \mid W}(y \mid w),
\end{align*}
where the third equality uses the conclusion from condition 1 in the proof of Lemma \ref{lemma:attainability_jointcdep}.

Likewise, 
\begin{align*}
    &\widetilde{\P}(Y_0 \leq y \mid W=w) \\
    &= \lim_{y_1 \to +\infty}\widetilde{\P}(Y_1 \leq y_1, Y_0 \leq y, X=1 \mid W=w) + \lim_{y_1 \to +\infty}\widetilde{\P}(Y_1 \leq y_1, Y_0 \leq y_0, X=0 \mid W=w) \\
    &= \sum_{x=0,1} \gamma \left[ \varepsilon \lim_{y_1 \to +\infty} F^{ll}(y_1,y,x|w) + (1-\varepsilon)\lim_{y_1 \to +\infty}F^{ul}(y_1,y,x|w) \right] \\
    &\quad + \sum_{x=0,1} (1-\gamma) \left[\varepsilon \lim_{y_1 \to +\infty} F^{lu}(y_1,y,x|w) + (1-\varepsilon)\lim_{y_1 \to +\infty} F^{uu}(y_1,y,x|w)\right] \\
    &= \gamma\left[\varepsilon \underline{F}_{Y_0 \mid W}(y \mid w) + (1-\varepsilon) \underline{F}_{Y_0 \mid W}(y \mid w)\right] + (1-\gamma) \left[\varepsilon \overline{F}_{Y_0 \mid W}(y \mid w) + (1-\varepsilon) \overline{F}_{Y_0 \mid W}(y \mid w)\right] \\
    &= \gamma \underline{F}_{Y_0 \mid W}(y \mid w) + (1-\gamma) \overline{F}_{Y_0 \mid W}(y \mid w).
\end{align*}

\medskip
\noindent \textbf{Verifying Condition 2:} For $x\in \{0,1\}$, we have that 
\begin{align*}
    \widetilde{\P}(X=x \mid W=w) 
    &= \lim_{y_0,y_1 \to +\infty}\widetilde{\P}(Y_1 \leq y_1, Y_0 \leq y_0, x \mid W=w) \\
    &= \gamma \left[ \varepsilon \lim_{y_0,y_1 \to +\infty} F^{ll}(y_1,y_0,x|w) + (1-\varepsilon)\lim_{y_0,y_1 \to +\infty} F^{ul}(y_1,y_0,x|w) \right]  \\
    &\quad + (1-\gamma) \left[\varepsilon \lim_{y_0,y_1 \to +\infty} F^{lu}(y_1,y_0,x|w) + (1-\varepsilon) \lim_{y_0,y_1 \to +\infty} F^{uu}(y_1,y_0,x|w)\right] \\
    &= \gamma \left[ \varepsilon p_{x|w} + (1-\varepsilon) p_{x|w} \right] +
    (1-\gamma) \left[\varepsilon p_{x|w} + (1-\varepsilon) p_{x|w} \right] \\
    &= p_{x|w},
\end{align*}
where the third equality follows by the condition 2 in the proof of Lemma \ref{lemma:attainability_jointcdep}.

\medskip
\noindent \textbf{Verifying Condition 3:} Similar to the proof of condition 2, condition 3 follows by the fact that all the cdfs $F^{ul}$, $F^{lu}$, $F^{uu}$, and $F^{ll}$ satisfy condition 3 as argued in Lemma \ref{lemma:attainability_jointcdep}. Hence it follows that their convex combination $\widetilde{\P}$ also satisfies this condition.

\medskip
\noindent \textbf{Verifying Condition 4:} As in the proof of Lemma \ref{lemma:attainability_jointcdep}, we established propensity score functions $p^{ul}$, $p^{lu}$, $p^{uu}$, and $p^{ll}$ under the cdfs $F^{ul}$, $F^{lu}$, $F^{uu}$, and $F^{ll}$, respectively. Applying Lemma \ref{lemma:propscore_cvx_comb} to the conditional mixture distribution $\widetilde{\P}$ gives the propensity score as below
\begin{align*}
    &\widetilde{\E}(X \mid Y_1,Y_0, W=w) \\
    &= \omega^{ul}(Y_1,Y_0)p^{ul}(Y_1,Y_0,w) + \omega^{lu}(Y_1,Y_0)p^{lu}(Y_1,Y_0,w) + \omega^{uu}(Y_1,Y_0)p^{uu}(Y_1,Y_0,w) + \omega^{ll}(Y_1,Y_0)p^{ll}(Y_1,Y_0,w)
\end{align*}
almost surely under $\widetilde{\P}$, where $\omega^k(Y_1, Y_0) \in [0,1]$, and $\sum_{k} \omega^k(Y_1,Y_0) = 1$ almost surely under $\widetilde{\P}$ for $k \in \{ul, lu, uu, ll\}$. Since we have argued that $p^k(u,v,w) \in [\lc,\uc]$ for $(u,v) \in \R^2$, therefore, $\widetilde{\E}(X \mid Y_1,Y_0,W=w) \in [\lc,\uc]$ almost surely under $\widetilde{\P}$, which concludes the proof.
\end{proof}

\subsection{Proofs for Section \ref{subsec:monotonic_parameter_bounds}}\label{appendix:parameter_bounds_proofs}

\begin{proof}[Proof of Theorem \ref{prop:IDset_monotonic_param}]
First, we prove this proposition when Assumption \ref{assn:marginal_cdep} holds. Fix an arbitrary $(F_1,F_0,C) \in \mathcal{I}_0^\text{marg}(F_{Y,X,W};c)$. By the definition of $\mathcal{I}_0^\text{marg}(F_{Y,X,W};c)$, there exists a joint cdf $F_{Y_1,Y_0 \mid X,W}$ that generates $(F_1,F_0,C)$ and satisfies Assumption \ref{assn:marginal_cdep}. By Lemma \ref{lemma:cdf_bounds_margcdep} and the monotonicity assumption, we have 
\[
	\theta(\overline{F}_{Y_1 \mid W},\underline{F}_{Y_0 \mid W},F_{Y,X,W})
	\leq
	\theta(F_1,F_0,F_{Y,X,W})
	\leq
	\theta(\underline{F}_{Y_1 \mid W},\overline{F}_{Y_0 \mid W},F_{Y,X,W}).
\]
Since $(F_1, F_0, C)$ is arbitrary, we have 
\[
	\theta(\overline{F}_{Y_1 \mid W},\underline{F}_{Y_0 \mid W},F_{Y,X,W})
	\leq
	\inf_{(F_1,F_0,C)\in \mathcal{I}_0^\text{marg}(F_{Y,X,W;c})} \theta(F_1,F_0,F_{Y,X,W})
\]
and
\[
\sup_{(F_1,F_0,C)\in \mathcal{I}_0^\text{marg}(F_{Y,X,W;c})} \theta(F_1,F_0,F_{Y,X,W})
	\leq
	\theta(\underline{F}_{Y_1 \mid W},\overline{F}_{Y_0 \mid W},F_{Y,X,W}).
\]
Furthermore, as demonstrated by Theorem \ref{thm:cdf_sharp_margcdep}, $(\overline{F}_{Y_1 \mid W}, \underline{F}_{Y_0 \mid W}, C)$ and $(\underline{F}_{Y_1 \mid W}, \overline{F}_{Y_0 \mid W}, C)$ are contained in the identified set for any copula $C\in\mathcal{C}_{1,0 \mid X,W}$. This implies
\[
	\inf_{(F_1,F_0,C)\in \mathcal{I}_0^\text{marg}(F_{Y,X,W;c})} \theta(F_1,F_0,F_{Y,X,W}) \leq \theta(\overline{F}_{Y_1 \mid W},\underline{F}_{Y_0 \mid W},F_{Y,X,W})
\]
and
\[
	\sup_{(F_1,F_0,C)\in \mathcal{I}_0^\text{marg}(F_{Y,X,W;c})} \theta(F_1,F_0,F_{Y,X,W}) \geq \theta(\underline{F}_{Y_1 \mid W},\overline{F}_{Y_0 \mid W},F_{Y,X,W}).
\]
Thus we conclude that 
\begin{align*}
	& \inf_{(F_1,F_0,C)\in \mathcal{I}_0^\text{marg}(F_{Y,X,W;c})} \theta(F_1,F_0,F_{Y,X,W}) = \theta(\overline{F}_{Y_1 \mid W},\underline{F}_{Y_0 \mid W},F_{Y,X,W}) \\
	& \sup_{(F_1,F_0,C)\in \mathcal{I}_0^\text{marg}(F_{Y,X,W;c})} \theta(F_1,F_0,F_{Y,X,W}) = \theta(\underline{F}_{Y_1 \mid W},\overline{F}_{Y_0 \mid W},F_{Y,X,W}).
\end{align*}

Note that Theorem \ref{thm:cdf_sharp_margcdep} also implies that 
$(\varepsilon \underline{F}_{Y_1 \mid W} + (1-\varepsilon) \overline{F}_{Y_1 \mid W},\gamma \underline{F}_{Y_0 \mid W} + (1-\gamma) \overline{F}_{Y_0 \mid W}, C_{1,0 \mid X,W})$ belongs to the identified set $I^\text{marg}_0(F_{Y,X,W};c)$ for each $(\varepsilon, \gamma)\in [0,1]^2$. By the continuity of the mapping $(\epsilon,\gamma)\mapsto\theta(\varepsilon \underline{F}_{Y_1 \mid W} + (1-\epsilon)\overline{F}_{Y_1 \mid W}, \gamma \underline{F}_{Y_0 \mid W} + (1-\gamma)\overline{F}_{Y_0 \mid W})$ and the definition of $\mathcal{I}_\theta^\text{marg}(F_{Y,X,W};c)$, the sharpness of the interior then follows by the intermediate value theorem. The proof follows the same arguments when imposing Assumption \ref{assn:joint_cdep}, where we use Theorem \ref{thm:cdf_sharp_jointcdep} instead of Theorem \ref{thm:cdf_sharp_margcdep}.
\end{proof}

\begin{proof}[Proof of Lemma \ref{lemma:monotonic parameters}]
	
\textbf{Part 1:} Because $Q_{Y_x}(U) \sim  Y_x$ when $U \sim \text{Unif}(0,1)$, we can write 
	\[
	\E[Y_x] = \E[Q_{Y_x}(U)] = \int_0^1 Q_{Y_x}(\tau) \, d\tau = \int_0^1 \theta_Q(F_{Y_x};\tau) \, d\tau.
	\] 
	By the proof of Part 2 below, $\theta_Q(F_{Y_x};\tau)$ is increasing in $F_{Y_x}$ for all $\tau \in (0,1)$. This implies  $\theta_{\E}(F_{Y_x}) = \int_0^1 \theta_Q(F_{Y_x};\tau) \, d\tau$ is also increasing. Continuity follows from
\begin{align*}
	\theta_{\E}(\varepsilon \underline{F}_{Y_x} + (1-\varepsilon)\overline{F}_{Y_x}) &= \int y \, d\left(\varepsilon \underline{F}_{Y_x}(y) + (1-\varepsilon)\overline{F}_{Y_x}(y)\right)\\
	&= \int y \, d\left(\varepsilon \underline{F}_{Y_x}(y)\right) + \int y \, d\left((1-\varepsilon) \overline{F}_{Y_x}(y)\right)\\
	&= \varepsilon \theta_{\E}(\underline{F}_{Y_x}) + (1-\varepsilon) \theta_{\E}(\overline{F}_{Y_x})
\end{align*}
being continuous in $\varepsilon$ over $\varepsilon \in [0,1]$.	
	
\textbf{Part 2:}
Suppose $F_{Y_x}(y) \leq F_{Y_x}'(y)$ for all $y \in \R$. Therefore, for any $\tau \in (0,1)$, $\{y \in \R: F_{Y_x}(y) \geq \tau\} \subseteq \{y \in \R: F_{Y_x}'(y) \geq \tau\}$. Hence, 
\[
	\theta_Q(F_{Y_x};\tau) = \inf\{y \in \R: F_1(y) \geq \tau\} \geq \inf\{y \in \R: F_{Y_x}'(y) \geq \tau\} = \theta_Q(F_{Y_x}';\tau).
\] 
Since $F_{Y_x} \succeq F_{Y_x}'$, we have that $\theta_Q(F_{Y_x};\tau)$ is increasing in $F_{Y_x}$.

\textbf{Part 3:} 
Suppose $F_{Y_x}(y) \leq F_{Y_x}'(y)$ for all $y \in \R$. Denote by 
\begin{align*}
& F_{Y_x \mid X}(y \mid 1-x) = \frac{F_{Y_x}(y) - F_{Y \mid X}(Y \mid X)p_x}{p_{1-x}}, \\
& F_{Y_x \mid X}'(y \mid 1-x) = \frac{F_{Y_x}'(y) - F_{Y \mid X}(Y \mid X)p_x}{p_{1-x}}.
\end{align*} 
Then $F_{Y_x \mid X}(\cdot \mid 1-x) \succeq F_{Y_x \mid X}'(\cdot \mid 1-x)$ and, by Part 2, 
\begin{align*}
	\theta_{CQ}(F_{Y_x};\tau) &= \theta_Q(F_{Y_x \mid X}(\cdot \mid 1-x);\tau) \geq \theta_Q(F_{Y_x \mid X}'(\cdot \mid 1-x);\tau) = \theta_{CQ}(F_{Y_x}';\tau)
\end{align*}
for any $\tau \in (0,1)$. Therefore, $\theta_{CQ}$ is increasing in $F_{Y_x}$.

\textbf{Part 4:} Suppose $F_{Y_1 \mid W} \succeq F_{Y_1'|W}$. This implies
\[
	\int y \, dF_{Y_1 \mid W}(y \mid w) \geq \int y \, dF_{Y_1'|W}(y \mid w)  \text{ for all } w \in \supp(W)
\]
which in turn implies
\[
	\1\left(\int y \, dF_{Y_1 \mid W}(y \mid w) - \E[Y_0 \mid W=w] \leq z\right) \leq \1\left(\int y \, dF_{Y_1'|W}(y \mid w) - \E[Y_0 \mid W=w] \leq z\right)
\]
for all $w \in \supp(W)$ and hence
\[
	\P\left(\int y \, dF_{Y_1 \mid W}(y \mid W) - \E[Y_0 \mid W] \leq z\right) \leq \P\left(\int y \, dF_{Y_1'|W}(y \mid W) - \E[Y_0 \mid W] \leq z\right).
\]
The first statement holds by by Part 1 of this lemma, the second holds directly, and the third by integrating over the distribution of $W$. Therefore, this last cdf is decreasing in $F_{Y_1 \mid W}$. By Part 2 of this lemma, its corresponding quantile will be decreasing in $F_{Y_1 \mid W}$. This parameter is decreasing in $F_{Y_0 \mid W}$ because of the minus sign inside the CATE.

\textbf{Part 5:} Suppose $F_{Y_x} \succeq F_{Y_x}'$ for $x \in \{0,1\}$. Then, $F_{Y_x}(y) \leq F_{Y_x}'(y)$ for all $y \in \R$. Therefore, for any $(y_1,y_0) \in \R^2$ and copula $C$
\begin{align*}
	\theta(F_{Y_1},F_{Y_0},C;y_1,y_0) = C(F_{Y_1}(y_0),F_{Y_0}(y_0)) \leq C(F_{Y_1}'(y_0),F_{Y_0}'(y_0)) = \theta(F_{Y_1}',F_{Y_0}',C;y_1,y_0)
\end{align*}
because $C$, as a copula, is nondecreasing in its arguments. We conclude that this parameter is decreasing in both $F_{Y_1}$ and $F_{Y_0}$.

\textbf{Part 6:} We begin by showing that $(Y_1,Y_0) \sim (Q_{Y_1}(U_1),Q_{Y_0}(U_0))$ where $(U_1,U_0)$ have joint cdf $C$. To see this, note that $F_{Y_1,Y_0}(y_1,y_0) =  C(F_{Y_1}(y_1),F_{Y_0}(y_0))$ by Sklar's Theorem. Also,
\begin{align*}
	F_{Y_1,Y_0}(y_1,y_0) &= C(F_{Y_1}(y_1),F_{Y_0}(y_0))\\
	&= \P(U_1 \leq F_{Y_1}(y_1), U_0 \leq F_{Y_0}(y_0))\\
	&= \P(Q_{Y_1}(U_1) \leq y_1, Q_{Y_0}(U_0) \leq y_0),
\end{align*}
where the third equality follows from Lemma \ref{lemma:quantile_props}.1.

Based on this, we can write the functional as
	\begin{align*}
		\theta_{\text{DTE}}(F_{Y_1},F_{Y_0},C;z) &= \E[\1(Y_1 - Y_0 \leq z)]\\
		&=  \E[\1(Q_{Y_1}(U_1) - Q_{Y_0}(U_0) \leq z)]\\
		&=  \int \1(Q_{Y_1}(u_1) - Q_{Y_0}(u_0) \leq z)\, dC(u_1,u_0).
	\end{align*}
Now suppose that $F_{Y_1} \succeq F_{Y_1}'$. By Part 2 above, this implies that $Q_{Y_1}(u_1) \geq Q_{Y_1}'(u_1)$ for all $u_1 \in (0,1)$ and thus
	\begin{align*}
		\theta_{\text{DTE}}(F_{Y_1},F_{Y_0},C;z) &=  \int \1(Q_{Y_1}(u_1) - Q_{Y_0}(u_0) \leq z)\, dC(u_1,u_0)\\
		&\leq \int \1(Q_{Y_1}'(u_1) - Q_{Y_0}(u_0) \leq z)\, dC(u_1,u_0)\\
		&= \theta_{\text{DTE}}(F_{Y_1}',F_{Y_0},C;z).
	\end{align*}
	Therefore, $\theta_{\text{DTE}}(F_{Y_1},F_{Y_0},C;z)$ is decreasing in $F_{Y_1}$. An analogous argument shows that it is increasing in $F_{Y_0}$.
\end{proof}

\section{Proofs for Section \ref{sec:mainTreatmentEffectBounds}}

\begin{proof}[Proof of Proposition \ref{prop:joint_cdf_bounds}]
	By Lemma \ref{lemma:cdf_bounds_margcdep} and the monotonicity of copulas in their arguments, we have that
	\begin{align*}
		&\sup_{(F_1,F_0,C) \in \mathcal{I}_0^\text{marg}(F_{Y,X,W};c)} \theta_\text{CDF}(F_1,F_0,C,F_{Y,X,W};y_1,y_0) \\
		&\leq \sup_{C \in \mathcal{C}_{1,0 \mid X,W}}\theta_\text{CDF}(\overline{F}_{Y_1 \mid X,W},\overline{F}_{Y_0 \mid X,W},C,F_{Y,X,W};y_1,y_0)\\
		&= \theta_\text{CDF}(\overline{F}_{Y_1 \mid X,W},\overline{F}_{Y_0 \mid X,W},\overline{C},F_{Y,X,W};y_1,y_0).
	\end{align*}
	The equality follows from the Fr\'echet-Hoeffding bounds. Similarly, 
	\begin{align*}
		&\inf_{(F_1,F_0,C) \in \mathcal{I}_0^\text{marg}(F_{Y,X,W};c)} \theta_\text{CDF}(F_1,F_0,C,F_{Y,X,W};y_1,y_0) \\
		&\geq \inf_{C \in \mathcal{C}_{1,0 \mid X,W}}\theta_\text{CDF}(\underline{F}_{Y_1 \mid X,W},\underline{F}_{Y_0 \mid X,W},C,F_{Y,X,W};y_1,y_0)\\
		&= \theta_\text{CDF}(\underline{F}_{Y_1 \mid X,W},\underline{F}_{Y_0 \mid X,W},\underline{C},F_{Y,X,W};y_1,y_0).
	\end{align*}
	To show sharpness of the interval \eqref{eq:IDset_jointcdf}, consider the following choices for conditional cdfs and copulas:
	\begin{align*}
		(\varepsilon_1 \underline{F}_{Y_1 \mid X,W} + (1 - \varepsilon_1)\overline{F}_{Y_1 \mid X,W},\varepsilon_2 \underline{F}_{Y_0 \mid X,W} + (1 - \varepsilon_2)\overline{F}_{Y_0 \mid X,W},  \varepsilon_3\underline{C} + (1-\varepsilon_3)\overline{C})
	\end{align*}
	for $\varepsilon \coloneqq (\varepsilon_1,\varepsilon_2,\varepsilon_3)  \in [0,1]^3$. For any $\varepsilon \in [0,1]^3$, this triple belongs to $\mathcal{I}_0^\text{marg}(F_{Y,X,W};c)$. Setting $\varepsilon = (0,0,0)$ and $\varepsilon = (1,1,1)$ yields the upper and lower bounds of the interval, so the bounds are sharp. To show the interior is sharp, consider the function 
	\begin{align*}
		\varepsilon &\mapsto 
		\theta_\text{CDF}\left(
		\begin{array}{c}
		\varepsilon_1 \underline{F}_{Y_1 \mid X,W} + (1 - \varepsilon_1)\overline{F}_{Y_1 \mid X,W}, \\
		\varepsilon_2 \underline{F}_{Y_0 \mid X,W} + (1 - \varepsilon_2)\overline{F}_{Y_0 \mid X,W}, \\
		\varepsilon_3 \underline{C} + (1-\varepsilon_3)\overline{C}, \\
		F_{Y,X,W}
		\end{array};y_1,y_0\right)\\
		&= \varepsilon_3 
		\E\left[
		\max\left\{
		\begin{array}{c}
		\varepsilon_1 \underline{F}_{Y_1 \mid X,W}(y_1 \mid X,W) + (1 - \varepsilon_1)\overline{F}_{Y_1 \mid X,W}(y_1 \mid X,W) \\
		+\varepsilon_2 \underline{F}_{Y_0 \mid X,W}(y_0 \mid X,W) + (1 - \varepsilon_2)\overline{F}_{Y_0 \mid X,W}(y_0 \mid X,W) - 1
		\end{array},
		0
		\right\}\right]\\
		&\quad + (1-\varepsilon_3) 
		\E
		\left[\min
		\left\{
		\begin{array}{c}
		\varepsilon_1 \underline{F}_{Y_1 \mid X,W}(y_1 \mid X,W) + (1 - \varepsilon_1)\overline{F}_{Y_1 \mid X,W}(y_1 \mid X,W), \\
		\varepsilon_2 \underline{F}_{Y_0 \mid X,W}(y_0 \mid X,W) + (1 - \varepsilon_2)\overline{F}_{Y_0 \mid X,W}(y_0 \mid X,W)
		\end{array}
		\right\}\right].
	\end{align*}
	This mapping is continuous in $\varepsilon_3$. It is also continuous in $\varepsilon_1$ and $\varepsilon_2$ since the functions $(u,v) \mapsto \underline{C}(u,v)$ and $(u,v) \mapsto \overline{C}(u,v)$ are both continuous, and by the dominated convergence theorem. Therefore, by the intermediate value theorem, all values in the interval \eqref{eq:IDset_jointcdf} are attained and thus the identified set is this interval.
\end{proof}

\begin{proof}[Proof of Proposition \ref{prop:DTE_bounds}]
	By lemmas \ref{lemma:cdf_bounds_margcdep} and \ref{lemma:monotonic parameters}.5
	\begin{align*}
		&\sup_{(F_1,F_0,C) \in \mathcal{I}_0^\text{marg}(F_{Y,X,W};c)} \theta_\text{DTE}(F_{Y_1 \mid X,W},F_{Y_0 \mid X,W},C_{1,0 \mid X,W},F_{Y,X,W};z) \\
		&\leq \sup_{C \in \mathcal{C}_{1,0 \mid X,W}}\theta_\text{DTE}(\overline{F}_{Y_1 \mid X,W},\underline{F}_{Y_0 \mid X,W},C,F_{Y,X,W};z).
	\end{align*}
	By Lemma 2.1 in \cite{FanPark2010},
	\begin{equation}
	\label{eq:DTE_upper_bound}
	\begin{aligned}
		&\sup_{C \in \mathcal{C}_{1,0 \mid X,W}} \theta_\text{DTE}(\overline{F}_{Y_1 \mid X,W},\underline{F}_{Y_0 \mid X,W},C,F_{Y,X,W};z)	\\	
	 	&\leq  1 +  \E\left[\min\left\{\inf_{y \in \R} \left(\overline{F}_{Y_1 \mid X,W}(y \mid X,W) - \underline{F}_{Y_0 \mid X,W}(y - z \mid X,W)\right),0\right\}\right].
	\end{aligned}
	\end{equation}
	
	This bound can be attained because the cdf pair $(\overline{F}_{Y_1 \mid X,W},\underline{F}_{Y_0 \mid X,W})$ is attainable by Theorem \ref{thm:cdf_sharp_margcdep}, and the bound in \eqref{eq:DTE_upper_bound} is attained by Lemma 2.1 in \cite{FanPark2010} since the set of conditional copulas under marginal $c$-dependence is unrestricted.
	
	Similar proof can be used to show that the lower bound 
	\[
	\E\left[\max\left\{\sup_{y \in \R} \left(\underline{F}_{Y_1 \mid X,W}(y \mid X,W) - \overline{F}_{Y_0 \mid X,W}(y - z \mid X,W)\right),0\right\}\right]
	\]
	is sharp as well.
\end{proof}

\section{Appendix: Explicit Bounds on Expected Potential Outcomes}
\label{appendix:bounds_exp_outcomes}

\begin{lemma}\label{lemma:integrals of quantiles}
	Let $Y$ be random variable with cdf $F$ and quantile function $Q$. Suppose $\E(|Y|) < \infty$. Then, for $a \in (0,1)$:
	\begin{align*}
		\int_0^a Q(u) \, du &= \E[Y \mid Y \leq Q(a)] F(Q(a)) - Q(a)(F(Q(a)) - a)\\
		\int_a^1 Q(u) \, du &= \E[Y \mid Y > Q(a)](1 - F(Q(a))) + Q(a)(F(Q(a)) - a).
	\end{align*}
	If $\P(Y = Q(a)) = 0$, then 
	\begin{align*}
		\int_0^a Q(u) \, du &= \E[Y \mid Y \leq Q(a)] a\\
		\int_a^1 Q(u) \, du &= \E[Y \mid Y > Q(a)](1 - a).
	\end{align*}
\end{lemma}

\begin{lemma}\label{lemma:expectation bounds}
	Let Assumption \ref{assn:overlap} hold. Then,
	\begin{equation}
	\label{eq:Y1ubExp}
	\begin{aligned}
	&\int y \, d\overline{F}_{Y_1 \mid W}(y \mid w) \\
	&= \frac{p_{1|w}}{\lc} \left[\E[Y \mid Y \leq \overline{Q}_1, X=1,W=w] {F}_{Y \mid X,W}(\overline{Q}_1 \mid 1,w) - \overline{Q}_1(F_{Y \mid X,W}(\overline{Q}_1 \mid 1,w) - \overline{\tau}_1)\right]\\
	& +  \frac{p_{1|w}}{\uc} \left[\E[Y \mid Y > \overline{Q}_1, X=1, W=w] (1 - {F}_{Y \mid X,W}(\overline{Q}_1 \mid 1,w)) + \overline{Q}_1(\overline{F}_{Y \mid X,W}(\overline{Q}_1 \mid 1,w) - \overline{\tau}_1)\right],
	\end{aligned}
	\end{equation}
	\begin{equation}
	\label{eq:Y1lbExp}
	\begin{aligned}
	&\int y \, d\underline{F}_{Y_1 \mid W}(y \mid w) \\
	&= \frac{p_{1|w}}{\uc} \left[\E[Y \mid Y \leq \underline{Q}_1, X=1,W=w] F_{Y \mid X,W}(\underline{Q}_1 \mid 1,w) - \underline{Q}_1(F_{Y \mid X,W}(\underline{Q}_1 \mid 1,w) - \underline{\tau}_1)\right]\\
	& +  \frac{p_{1|w}}{\lc} \left[\E[Y \mid Y > \underline{Q}_1, X=1, W=w] (1 - F_{Y \mid X,W}(\underline{Q}_1 \mid 1,w)) + \underline{Q}_1(F_{Y \mid X,W}(\underline{Q}_1 \mid 1,w) - \underline{\tau}_1)\right]
	\end{aligned}
	\end{equation}
and
	\begin{equation}
	\label{eq:Y0ubExp}
	\begin{aligned}
	&\int y \, d\overline{F}_{Y_0 \mid W}(y \mid w) \\
	&= \frac{p_{0|w}}{1-\uc} \left[\E[Y \mid Y \leq \overline{Q}_0, X=0,W=w] F_{Y \mid X,W}(\overline{Q}_0|0,w) - \overline{Q}_0(F_{Y \mid X,W}(\overline{Q}_0|0,w) - \overline{\tau}_0)\right]\\
	& +  \frac{p_{0|w}}{1-\lc} \left[\E[Y \mid Y > \overline{Q}_0, X=0, W=w] (1 - F_{Y \mid X,W}(\overline{Q}_0|0,w)) + \overline{Q}_0(F_{Y \mid X,W}(\overline{Q}_0|0,w) - \overline{\tau}_0)\right],
	\end{aligned}
	\end{equation}
	\begin{equation}
	\label{eq:Y0lbExp}
	\begin{aligned}
	&\int y \, d\underline{F}_{Y_0 \mid W}(y \mid w) \\
	&= \frac{p_{0|w}}{1-\lc} \left[\E[Y \mid Y \leq \underline{Q}_0, X=0,W=w] F_{Y \mid X,W}(\underline{Q}_0|0,w) - \underline{Q}_0(F_{Y \mid X,W}(\underline{Q}_0|0,w) - \underline{\tau}_0)\right]\\
	& +  \frac{p_{0|w}}{1-\uc} \left[\E[Y \mid Y > \underline{Q}_0, X=0, W=w] (1 - F_{Y \mid X,W}(\underline{Q}_0|0,w)) + \underline{Q}_0(F_{Y \mid X,W}(\underline{Q}_0|0,w) - \underline{\tau}_0)\right].
	\end{aligned}
	\end{equation}
	If $Y$ is continuously distributed conditionally on $(X,W)$, then these expressions simplify to
	\begin{align*}
		\int y \, d\overline{F}_{Y_1 \mid W}(y \mid w) &=  \E[Y \mid Y \leq \overline{Q}_1, X=1,W=w] \frac{\uc - p_{1|w}}{\uc - \lc}\\
	& +   \E[Y \mid Y > \overline{Q}_1, X=1, W=w]\frac{p_{1|w} - \lc}{\uc - \lc} \\
	\int y \, d\underline{F}_{Y_1 \mid W}(y \mid w) &=  \E[Y \mid Y \leq \underline{Q}_1, X=1,W=w] \frac{p_{1|w} - \lc}{\uc - \lc}\\
	& +   \E[Y \mid Y > \underline{Q}_1, X=1, W=w] \frac{\uc - p_{1|w}}{\uc - \lc}
	\end{align*}
and
	\begin{align*}
		\int y \, d\overline{F}_{Y_0 \mid W}(y \mid w) &=  \E[Y \mid Y \leq \overline{Q}_0, X=0,W=w] \frac{p_{1|w} - \lc}{\uc - \lc}\\
	& +   \E[Y \mid Y > \overline{Q}_0, X=0, W=w]\frac{\uc - p_{1|w}}{\uc - \lc} \\
	\int y \, d\underline{F}_{Y_0 \mid W}(y \mid w) &=  \E[Y \mid Y \leq \underline{Q}_0, X=0,W=w] \frac{\uc - p_{1|w}}{\uc - \lc}\\
	& +   \E[Y \mid Y > \underline{Q}_0, X=0, W=w] \frac{p_{1|w} - \lc}{\uc - \lc}
	\end{align*}
where $\overline{Q}_x$, $\underline{Q}_x$, $\overline{\tau}_x$, and $\underline{\tau}_x$ for $x=0,1$ are defined in Appendix \ref{appendix:notation}.
\end{lemma}

\subsection{Proofs of Analytical Bounds on Expectations}

\begin{proof}[Proof of Lemma \ref{lemma:integrals of quantiles}]
		First we consider the equality involving $\int_0^a Q(u) du$. Note that
\begin{align*}
	\int_0^a Q(u) \, du
	&= \int_0^1 Q(u) \1[u\leq a] \, du \\
	&= \int_0^1 Q(u) \1[Q(u) \leq Q(a), u\leq a] \, du \\
	&= \int_0^1 Q(u) \1[Q(u) \leq Q(a)] du - \int_0^1 Q(u) \1[Q(u) \leq Q(a), u>a] \, du \\
	&= \int_0^1 Q(u) \1[Q(u) \leq Q(a)] du - \int_0^1 Q(u) \1[Q(u) = Q(a), u>a] \, du \\
	&= \int_0^1 Q(u) \1[Q(u) \leq Q(a)] du - Q(a) \int_0^1 \1[Q(u) \leq Q(a), u>a] \, du,
\end{align*}
where the second, the fourth, and the last line follow by the monotonicity of quantile function $Q(\cdot)$.

The first term in the last line can be written as below:
\begin{align*}
	\int_0^1 Q(u) \1[Q(u) \leq Q(a)] \, du = \E(Y \1[Y\leq Q(a)]) = \E(Y \mid Y\leq Q(a))F(Q(a)),
\end{align*}
where the first equality follows by that $Q(U)$ has the same distribution as $Y$ if $U$ is uniformly distributed over $[0,1]$. To expand the second term, note that $\{u: Q(u)\leq Q(a), u>a\}$ is a half-open interval with the left endpoint $a$ excluded, and right endpoint $\sup\{u: Q(u) \leq Q(a)\}$ included in the interval due to the left-continuity of quantile function $Q(\cdot)$. So we have
\begin{align*}
	Q(a)\int_0^1 \1[Q(u)\leq Q(a), u>a] \, du
	&= Q(a)(\sup\{u: Q(u)\leq Q(a)\} - a) \\
	&= Q(a)(\sup\{u: u\leq F(Q(a))\} - a) \\
	&= Q(a)(F(Q(a)) - a),
\end{align*}
where the second line holds by Lemma \ref{lemma:quantile_props}.1. Given the above derivations, we conclude that 
\[
	\int_0^a Q(u) \, du = \E[Y \mid Y\leq Q(a)] F(Q(a)) - Q(a)(F(Q(a)) - a),
\]
as desired. 

Regarding the second equality involving $\int_a^1 Q(u) du$, note that $\int_0^1 Q(u) \, du = \E[Q(U)] = \E[Y]$. This implies
\begin{align*}
	\int_a^1 Q(u) \, du 
	&= \int_0^1 Q(u)du - \int_0^a Q(u) \, du \\
	&= \E[Y] - \E[Y \mid Y\leq Q(a)]F(Q(a)) + Q(a)(F(Q(a)) - a) \\
	&= \E[Y \mid Y>Q(a)](1-F(Q(a))) + Q(a)(F(Q(a)) - a),
\end{align*}
where the last line follows by the law of iterated expectation. So the second equality is established.

When $\P(Y=Q(a))=0$, the the CDF $F$ is continuous at $Q(a)$, which implies that 
$F(Q(a)) = a$ by Lemma \ref{lemma:quantile_props}.2. Therefore, 
\begin{align*}
	\int_0^a Q(u) \, du 
	= \E[Y \mid Y \leq Q(a)] F(Q(a)) - Q(a)(F(Q(a)) - a) 
	= \E[Y \mid Y \leq Q(a)] a,
\end{align*}
and similar arguments can be applied to $\int_a^1 Q(u) \, du$ as well. Therefore we have established the desired result.
\end{proof}

\begin{proof}[Proof of Lemma \ref{lemma:expectation bounds}]
We prove the claim for $\int_0^1 y \, d\overline{F}_{Y_1 \mid W}(y \mid w)$, and note that the claims for the other terms can be derived analogously. 

Let $U \sim \text{Unif}(0,1)$, then $\underline{Q}_{Y_1 \mid W}(U \mid 1,w)$ has the distribution $\overline{F}_{Y_1 \mid W}(\cdot \mid 1,w)$, which implies
\begin{align}
	\int y \, d\overline{F}_{Y_1 \mid W}(y \mid w) 
	&= \int_0^1 \underline{Q}_{Y_1 \mid W}(\tau \mid w) \, d\tau  \notag \\
	&= \int_0^1 Q_{Y \mid X,W}\left(\frac{\lc \tau}{p_{1|w}} \mid 1,w\right)\1\left[\tau \leq \frac{\uc - p_{1|w}}{\uc - \lc}\right] \, d\tau \label{eq:term1_Y1ubExp}\\
	&\quad + \int_0^1 Q_{Y \mid X,W}\left(\frac{p_{1|w} - \uc + \uc\tau}{p_{1|w}} \mid 1,w\right) \1\left[\tau > \frac{\uc - p_{1|w}}{\uc - \lc}\right] \, d\tau \label{eq:term2_Y1ubExp}
\end{align}
We expand the term \eqref{eq:term1_Y1ubExp} below:
\begin{align}
	&\int_0^1 Q_{Y \mid X,W}\left(\frac{\lc \tau}{p_{1|w}} \mid 1,w\right)\1\left[\tau \leq \frac{\uc - p_{1|w}}{\uc - \lc}\right] \, d\tau  \notag\\
	&= \frac{p_{1|w}}{\lc}\int_0^{\overline{\tau}_1} Q_{Y \mid X,W}(u \mid 1,w)\,du \notag \\
	&= \frac{p_{1|w}}{\lc}
	\left[\E[Y \mid Y\leq Q_{Y \mid X,W}(\overline{\tau}_1 \mid 1,w), X=1, W=w] F_{Y \mid X,W}(Q_{Y \mid X,W}(\overline{\tau}_1 \mid 1,w) \mid 1,w) \right] \notag \\
	&\quad - \frac{p_{1|w}}{\lc}Q_{Y \mid X,W}(\overline{\tau}_1 \mid 1,w) \left[F_{Y \mid X,W}(Q_{Y \mid X,W}(\overline{\tau}_1 \mid 1,w) \mid 1,w) - \overline{\tau}_1\right] \notag \\
	&= \frac{p_{1|w}}{\lc}\left[\E[Y \mid Y\leq \overline{Q}_1, X=1, W=w]F_{Y \mid X,W}(\overline{Q}_1 \mid 1,w) - \overline{Q}_1(F_{Y \mid X,W}(\overline{Q}_1 \mid 1,w) - \overline{\tau}_1)\right]. \label{eq:term1_final_Y1ubExp}
\end{align}
The first equality uses the changes of variable $u = \lc \tau/p_{1|w}$ and recall that
\[
	\overline{\tau}_1 = \frac{\lc}{p_{1|w}} \frac{\uc - p_{1|w}}{\uc - \lc}.
\]
The second equality follows by Lemma \ref{lemma:integrals of quantiles}, and the last line holds by recalling that $\overline{Q}_1 = Q_{Y \mid X,W}(\overline{\tau}_1 \mid 1,w)$.

Similarly, we can expand the term \eqref{eq:term2_Y1ubExp} below:
\begin{align}
&\int_0^1 Q_{Y \mid X,W}\left(\frac{p_{1|w} - \uc + \uc\tau}{p_{1|w}} \mid 1,w\right) \1\left[\tau > \frac{\uc - p_{1|w}}{\uc - \lc}\right] \, d\tau  \notag\\
&= \int_{\overline{\tau}_1}^1 Q_{Y \mid X,W}\left(u \mid 1,w\right)\, du \notag \\
&= \frac{p_{1|w}}{\uc}\left[\E[Y \mid Y > \overline{Q}_1, X=1, W=w] (1-F_{Y \mid X,W}(\overline{Q}_1 \mid 1,w)) + \overline{Q}_1\left(F_{Y \mid X,W}(\overline{Q}_1 \mid 1,w) - \overline{\tau}_1\right)\right], \label{eq:term2_final_Y1ubExp}
\end{align}
where we use the change of variable $u = 1 - \frac{(1-\tau)\uc}{p_{1|w}}$ in the second line.

Combining the above results, we con combine \eqref{eq:term1_final_Y1ubExp} and \eqref{eq:term2_final_Y1ubExp} to obtain the analytical formula of $\int y \, d\overline{F}_{Y_1 \mid W}$:
\begin{align*}
	& \int y \, d\overline{F}_{Y_1 \mid W}(y \mid w) \\
	& =
	\frac{p_{1|w}}{\lc}\left[\E[Y \mid Y\leq \overline{Q}_1, X=1, W=w]F_{Y \mid X,W}(\overline{Q}_1 \mid 1,w) - \overline{Q}_1(F_{Y \mid X,W}(\overline{Q}_1 \mid 1,w) - \overline{\tau}_1)\right]  \\
	&\quad + \frac{p_{1|w}}{\uc}\left[\E[Y \mid Y > \overline{Q}_1, X=1, W=w] (1-F_{Y \mid X,W}(\overline{Q}_1 \mid 1,w)) + \overline{Q}_1\left(F_{Y \mid X,W}(\overline{Q}_1 \mid 1,w) - \overline{\tau}_1\right)\right].
\end{align*}
Finally, we note that if $Y$ is continuously distributed conditional on $(X,W)$, then 
\[
	F_{Y \mid X,W}(\overline{Q}_1 \mid 1,w) = \overline{\tau}_1,
\]
which implies
\begin{align*}
	\int y \, d\overline{F}_{Y_1 \mid W}(y \mid w)
	&=
	\frac{p_{1|w}}{\lc} \E[Y \mid Y\leq \overline{Q}_1, X=1, W=w]\overline{\tau}_1 + \frac{p_{1|w}}{\uc} \E[Y \mid Y> \overline{Q}_1, X=1, W=w] (1-\overline{\tau}_1) \\
	&= \E[Y \mid Y\leq \overline{Q}_1, X=1, W=w] \frac{\uc - p_{1|w}}{\uc - \lc} + \E[Y \mid Y> \overline{Q}_1, X=1, W=w] \frac{p_{1|w} - \lc}{\uc-\lc},
\end{align*}
as desired.
\end{proof}

\end{document}